\documentclass[11pt,leqno]{article}
\usepackage{amssymb,amsfonts}
\usepackage{amsmath,amsthm,amsxtra}
\usepackage[all]{xy}
\usepackage{mathabx}
\usepackage{color}
\usepackage{verbatim}

\usepackage[margin=1.34in]{geometry}
\usepackage[linktocpage=true,colorlinks=true, linkcolor=blue, citecolor=red, urlcolor=green]{hyperref}

\usepackage{enumerate}

\newcommand\bigcheck[1]{#1 \raise1ex\hbox{$\hspace{-1ex}{}^\vee$}}
\newcommand\sucheck[1]{#1 \raise0.5ex\hbox{$\hspace{-1ex}{}^\vee$}}

\newcommand{\wcheck}{\!\!\!\!\!\!\widecheck{}\,\,\,\,\,\,}
\newcommand{\wwcheck}{\!\!\!\!\!\!\!\!\!\!\!\!\!\!\!\!\!\!\widecheck{}\,\,\,\,\,\,\,\,\,\,\,\,\,\,\,\,\,\,}

\makeatletter
\@addtoreset{equation}{section}
\makeatother

\newtheorem{theorem}{Theorem}[section]
\newtheorem{lemma}[theorem]{Lemma}
\newtheorem{corollary}[theorem]{Corollary}
\newtheorem{proposition}[theorem]{Proposition}
\newtheorem*{lemma*}{Lemma}

\theoremstyle{definition}
\newtheorem{definition}[theorem]{Definition}

\theoremstyle{remark}
\newtheorem{remark}[theorem]{Remark}
\newtheorem{example}[theorem]{Example}

\newcommand{\mc}[1]{{\mathcal #1}}
\newcommand{\mf}[1]{{\mathfrak #1}}
\newcommand{\mb}[1]{{\mathbb #1}}

\newcommand\tint{{\textstyle\int}}

\newcommand{\id}{{1 \mskip -5mu {\rm I}}}

\renewcommand{\tilde}{\widetilde}

\newcommand{\ad}{\mathop{\rm ad }}

\newcommand{\Mat}{\mathop{\rm Mat }}

\renewcommand{\ker}{\mathop{\rm Ker }}
\newcommand{\im}{\mathop{\rm Im }}

\newcommand{\Span}{\mathop{\rm Span }}

\newcommand{\ass}[1]{\stackrel{#1}{\longleftrightarrow}}

\definecolor{light}{gray}{.9}

\begin{document}


\title{Non-local Hamiltonian structures and applications to the theory 
of integrable systems I}

\author{
Alberto De Sole
\thanks{Dipartimento di Matematica, Universit\`a di Roma ``La Sapienza'',
00185 Roma, Italy ~~
desole@mat.uniroma1.it ~~~~
Supported in part by Department of Mathematics, M.I.T.},~~
Victor G. Kac
\thanks{Department of Mathematics, M.I.T.,
Cambridge, MA 02139, USA.~~
kac@math.mit.edu~~~~
Supported in part by an NSF grant, and the Center for Mathematics and Theoretical Physics (CMTP)
in Rome.~~ 
}~~
}

\maketitle

\vspace{4pt}

\begin{center}
\emph{Dedicated to Minoru Wakimoto on his 70-th birthday.}
\end{center}

\vspace{2pt}


\begin{abstract}
\noindent 
We develop a rigorous theory of non-local Hamiltonian structures,
built on the notion of a non-local Poisson vertex algebra.
As an application, we find conditions that guarantee applicability of the Lenard-Magri scheme
of integrability to a pair of compatible non-local Hamiltonian structures.
\end{abstract}


\section{Introduction}
\label{sec:intro}

Local Poisson brackets play a fundamental role in the theory of integrable systems.
Recall that a local Poisson bracket is defined by (see e.g. \cite{TF86}):
\begin{equation}\label{intro:eq1.1}
\{u_i(x),u_j(y)\}=H_{ij}\big(u(y),u'(y),\dots;\partial_y\big)\delta(x-y)\,,
\end{equation}
where $u=(u_1,\dots,u_\ell)$ is a vector valued function on a 1-dimensional manifold $M$,
$\delta(x-y)$ is the $\delta$-function: $\tint_Mf(y)\delta(x-y)dy=f(x)$,
and $H(\partial)=\big(H_{ij}(\partial)\big)_{i,j=1}^\ell$
is an $\ell\times\ell$ matrix differential operator,
whose coefficients are functions in $u,u',\dots,u^{(k)}$.
One requires, in addition, that \eqref{intro:eq1.1} ``satisfies the Lie algebra axioms''.

One of the ways to formulate this condition is as follows.
Let $\mc V$ be an algebra of differential polynomials in $u_1,\dots,u_\ell$,
i.e. the algebra of polynomials in $u_i^{(n)},\,i\in I=\{1,\dots,\ell\},\,n\in\mb Z_+$,
with $u_i^{(0)}=u_i$ and the derivation $\partial$, defined by $\partial u_i^{(n)}=u_i^{(n+1)}$,
or its algebra of differential functions extension.
The bracket \eqref{intro:eq1.1} extends, by the Leibniz rule and bilinearity,
to arbitrary $f,g\in\mc V$:
\begin{equation}\label{intro:eq1.2}
\{f(x),g(y)\}=
\sum_{i,j\in I}\sum_{m,n\in\mb Z_+}
\frac{\partial f(x)}{\partial u_i^{(m)}} \frac{\partial g(y)}{\partial u_j^{(n)}}
\partial_x^m \partial_y^n\{u_i(x),u_j(y)\}\,.
\end{equation}
Applying integration by parts, we get the following bracket on $\mc V/\partial\mc V$:
\begin{equation}\label{intro:eq1.3}
\{\tint f,\tint g\}=
\int \frac{\delta g}{\delta u}\cdot H(\partial)\frac{\delta f}{\delta u}\,,
\end{equation}
where $\tint$ is the canonical quotient map $\mc V\to\mc V/\partial\mc V$
and $\frac{\delta f}{\delta u}$ is the vector of variational derivatives
$\frac{\delta f}{\delta u_i}=\sum_{n\in\mb Z_+}(-\partial)^n\frac{\partial f}{\partial u_i^{(n)}}$.
Then one requires that the bracket \eqref{intro:eq1.3}
satisfies the Lie algebra axioms.
(The skewsymmetry of this bracket is equivalent to the skewadjointness of $H(\partial)$,
but the Jacobi identity is a complicated system of non-linear PDE on its coefficients.)
In this case the matrix differential operator $H(\partial)$ is called a \emph{Hamiltonian structure}.

Given an element $\tint h\in\mc V/\partial\mc V$, called a \emph{Hamiltonian functional},
the \emph{Hamiltonian equation} associated to $H(\partial)$ is the following evolution equation:
\begin{equation}\label{intro:eq1.4}
\frac{du}{dt}=H(\partial)\frac{\delta h}{\delta u}\,.
\end{equation}
For example, taking $H(\partial)=\partial$ and $h=\frac12(u^3+cuu'')$,
we obtain the KdV equation: $\frac{du}{dt}=3uu'+cu'''$.

Equation \eqref{intro:eq1.4} is called \emph{integrable} if $\tint h$
is contained in an infinite dimensional abelian subalgebra $A$ of the Lie algebra $\mc V/\partial\mc V$
with bracket \eqref{intro:eq1.3}.
Picking a basis $\{\tint h_n\}_{n\in\mb Z_+}$ of $A$,
we obtain a hierarchy of compatible integrable Hamiltonian equations:
$$
\frac{du}{dt_n}=H(\partial)\frac{\delta h_n}{\delta u}\,\,,\,\,\,\,n\in\mb Z_+\,.
$$

An alternative approach, proposed in \cite{BDSK09},
is to apply the Fourier transform $F(x,y)\mapsto\tint_Mdxe^{\lambda(x-y)}F(x,y)$
to both sides of \eqref{intro:eq1.2}
to obtain the following ``Master formula'' \cite{DSK06}:
\begin{equation}\label{intro:eq1.5}
\{f_\lambda g\}=
\sum_{i,j\in I}\sum_{m,n\in\mb Z_+}
\frac{\partial g}{\partial u_j^{(n)}}
(\lambda+\partial)^n
H_{ji}(\lambda+\partial)
(-\lambda-\partial)^m
\frac{\partial f}{\partial u_i^{(m)}}\,. 
\end{equation}
For an arbitrary $\ell\times\ell$ matrix differential operator $H(\partial)$ 
this $\lambda$-\emph{bracket}
is polynomial in $\lambda$, i.e. it takes values in $\mc V[\partial]$,
satisfies the \emph{left} and \emph{right Leibniz rules}:
\begin{equation}\label{intro:eq1.6}
\{f_\lambda gh\}=
g\{f_\lambda h\}+h\{f_\lambda g\}
\,\,,\,\,\,\,
\{fg_\lambda h\}=
\{f_{\lambda+\partial} g\}_\to h+\{f_{\lambda+\partial} h\}_\to g\,,
\end{equation}
where the arrow means that $\lambda+\partial$ should be moved to the right,
and the  \emph{sesquilinerity} axioms:
\begin{equation}\label{intro:eq1.7}
\{\partial f_\lambda g\}=-\lambda\{f_\lambda g\}
\,\,,\,\,\,\,
\{f_\lambda\partial g\}=(\lambda+\partial)\{f_\lambda g\}\,.
\end{equation}
It is proved in \cite{BDSK09} that the requirement that \eqref{intro:eq1.3} satisfies the Lie algebra axioms
is equivalent to the following two properties of \eqref{intro:eq1.5}:
\begin{equation}\label{intro:eq1.8}
\{g_\lambda f\}=-\{f_{-\lambda-\partial} g\}\,,
\end{equation}
\begin{equation}\label{intro:eq1.9}
\{f_\lambda \{g_\mu h\}\}=\{g_\mu\{f_\lambda h\}\}
+\{\{f_\lambda g\}_{\lambda+\mu}h\}\,.
\end{equation}
A differential algebra $\mc V$,
endowed with a polynomial $\lambda$-bracket, satisfying axioms \eqref{intro:eq1.6}--\eqref{intro:eq1.9},
is called a \emph{Poisson vertex algebra} (PVA).

It was demonstrated in \cite{BDSK09} that the PVA approach greatly simplifies the theory
of integrable Hamiltonian PDE, based on local Poisson brackets.
For example, equation \eqref{intro:eq1.4} becomes, in terms of the $\lambda$-bracket associated to $H$:
$$
\frac{du}{dt}=\{h_\lambda u\}\big|_{\lambda=0}\,,
$$
and the Lie bracket \eqref{intro:eq1.3} becomes
$$
\{\tint f,\tint g\}=\tint \{f_\lambda g\}\big|_{\lambda=0}\,.
$$

It is the purpose of the present paper to demonstrate that the adequate (and in fact indispensable)
tool for understanding non-local Poisson brackets is the theory of ``non-local'' PVA.

We define a \emph{non-local} $\lambda$-\emph{bracket} on the differential algebra $\mc V$
to take its values in $\mc V((\lambda^{-1}))$,
formal Laurent series in $\lambda^{-1}$ with coefficients in $\mc V$,
and to satisfy properties \eqref{intro:eq1.6} and \eqref{intro:eq1.7}.
The main example is the $\lambda$-bracket given by the Master Formula \eqref{intro:eq1.5},
where $H(\partial)$ is a matrix pseudodifferential operator.
The only problem with this definition is the interpretation of the operator $\frac1{\lambda+\partial}$;
this is defined by the geometric progression
$$
\frac1{\lambda+\partial}=\sum_{n\in\mb Z_+}(-1)^n\lambda^{-n-1}\partial^n\,.
$$

Property \eqref{intro:eq1.8} of the $\lambda$-bracket is interpreted in the same way,
but the interpretation of property \eqref{intro:eq1.9} is more subtle.
Indeed, in general, we have $\{f_\lambda\{g_\mu h\}\}\in\mc V((\lambda^{-1}))((\mu^{-1}))$,
but $\{g_\mu\{f_\lambda h\}\}\in\mc V((\mu^{-1}))((\lambda^{-1}))$,
and $\{\{f_\lambda g\}_{\lambda+\mu} h\}\in\mc V(((\lambda+\mu)^{-1}))((\lambda^{-1}))$,
so that all three terms of \eqref{intro:eq1.9} lie in different spaces.
Our key idea is to consider the space
$$
\mc V_{\lambda,\mu}=\mc V[[\lambda^{-1},\mu^{-1},(\lambda+\mu)^{-1}]][\lambda,\mu]\,,
$$
which is canonically embedded in all three of the above spaces.
We say that a $\lambda$-bracket is \emph{admissible} if 
$$
\{f_\lambda\{g_\mu h\}\}\in\mc V_{\lambda,\mu}
\,\,,\,\,\,\,
\text{ for all } f,g,h\in\mc V\,.
$$
It is immediate to see that then the other two terms of \eqref{intro:eq1.9} 
lie in $\mc V_{\lambda,\mu}$ as well,
hence \eqref{intro:eq1.9} is an identity in $\mc V_{\lambda,\mu}$.

We call the differential algebra $\mc V$, endowed with a non-local $\lambda$-bracket,
a \emph{non-local PVA}, if it satisfies \eqref{intro:eq1.8},
is admissible, and satisfies \eqref{intro:eq1.9}.

For an arbitrary pseudodifferential operator $H(\partial)$ the $\lambda$-bracket \eqref{intro:eq1.5}
is not admissible, but it is admissible for any \emph{rational} pseudodifferential operator,
i.e. contained in the subalgebra generated by differential operators and their inverses.
We show that, as in the local case (see \cite{BDSK09}),
this $\lambda$-bracket satisfies conditions \eqref{intro:eq1.8} and \eqref{intro:eq1.9}
if and only if \eqref{intro:eq1.8} holds for any pair $u_i,u_j$,
and \eqref{intro:eq1.9} holds for any triple $u_i,u_j,u_k$.
Also, \eqref{intro:eq1.8} is equivalent to skewadjointness of $H(\partial)$.

The simplest example of a non-local PVA corresponds 
to the skewadjoint operator $H(\partial)=\partial^{-1}$.
Then 
$$
\{u_\lambda u\}=\lambda^{-1}\,,
$$
and equation \eqref{intro:eq1.9} trivially holds for the triple $u,u,u$.
Note that \eqref{intro:eq1.1} in this case reads:
$\{u(x),u(y)\}=\partial_y^{-1}\delta(x-y)$,
which is quite difficult to work with.

The next example corresponds to Sokolov's skewadjoint operator \cite{Sok84},
$H(\partial)=u'\partial^{-1}\circ u'$.
The corresponding $\lambda$-bracket is
$$
\{u_\lambda u\}=u'\frac1{\lambda+\partial}u'\,.
$$
The verification of \eqref{intro:eq1.9} for the triple $u,u,u$ is straightforward.

A rational pseudodifferential operator $H(\partial)$ is a \emph{Hamiltonian structure} on $\mc V$
if the $\lambda$-bracket \eqref{intro:eq1.5} endows $\mc V$ with a structure of a non-local PVA
(in other words $H(\partial)$ should be skewadjoint and \eqref{intro:eq1.9} should hold
for any triple $u_i,u_j,u_k$).

Fix a ``minimal fractional decomposition'' $H=AB^{-1}$.
This means that $A,B$ are differential operators over $\mc V$,
such that $\ker A\cap\ker B=0$
in any differential domain extension of $\mc V$.
It is shown in \cite{CDSK12b} that such a decomposition always exists and that the above property
is equivalent to the property that any common right factor of $A$ and $B$
is invertible over the field of fractions of $\mc V$.
Then the basic notions of the theory of integrable systems are defined as follows.
A \emph{Hamiltonian functional} (for $H=AB^{-1}$)
is an element $\tint h\in\mc V/\partial\mc V$ such that $\frac{\delta\tint h}{\delta u}=B(\partial)F$
for some $F\in\mc V^\ell$.
Then the element $P=A(\partial)F$ is called an associated \emph{Hamiltonian vector field}.
Denote by $\mc F(H)\subset\mc V/\partial\mc V$ the subspace of all Hamiltonian functionals,
and by $\mc H(H)\subset\mc V^\ell$ the subspace of all Hamiltonian vector fields
(they are independent of the choice of the minimal fractional decomposition for $H$):
$$
\mc F(H)=\Big(\frac{\delta}{\delta u}\Big)^{-1}\Big(\im B\Big)\subset\mc V/\partial\mc V
\,\,,\,\,\,\,
\mc H(H)=A\Big(B^{-1}\Big(\im\frac{\delta}{\delta u}\Big)\Big)\subset\mc V^\ell\,.
$$
Then it is easy to show that $\mc F(H)$ is a Lie algebra with respect to the well-defined
bracket \eqref{intro:eq1.3}, and $\mc H(H)$ is a subalgebra of the Lie algebra $\mc V^\ell$
with bracket $[P,Q]=D_Q(\partial)P-D_P(\partial)Q$, where $D_P(\partial)$ is the Frechet derivative.

A \emph{Hamiltonian equation}, associated to the Hamiltonian structure $H$
and a Hamiltonian functional $\tint h\in\mc F(H)$,
with an associated Hamiltonian vector field $P\in\mc H(H)$,
is the following evolution equation:
\begin{equation}\label{intro:eq1.10}
\frac{du}{dt}=P\,.
\end{equation}
Note that \eqref{intro:eq1.10} coincides with \eqref{intro:eq1.4} in the local case.
The Hamiltonian equation \eqref{intro:eq1.10} is called \emph{integrable} if there exist
linearly independent infinite sequences
$\tint h_n\in\mc F(H)$ and $P_n\in\mc H(H)$, $n\in\mb Z_+$,
such that $\tint h_0=\tint h$, $P_0=P$,
$P_n$ is associated to $\tint h_n$,
and $\{\tint h_m,\tint h_n\}=0,\,[P_m,P_n]=0$
for all $m,n\in\mb Z_+$.
In this case we have a hierarchy of compatible integrable equations
$$
\frac{du}{dt_n}=P_n
\,\,,\,\,\,\, n\in\mb Z_+\,.
$$

Having given rigorous definitions of the basic notions of the theory of Hamiltonian equations
with non-local Hamiltonian structures,
we proceed to establish some basic results of the theory.

The first result is Theorem \ref{20111021:thm},
which states that if $H$ and $K$ are compatible non-local Hamiltonian structures and $K$ is invertible
(as a pseudodifferential operator), then the sequence
of rational pseudodifferential operators
$H^{[0]}=K, H^{[n]}=(H\circ K^{-1})^{n-1}\circ H,\,n\geq1$,
is a compatible family of non-local Hamiltonian structures.
(As usual \cite{Mag78,Mag80} a collection of non-local Hamiltonian structures
is called compatible if any their finite linear combination is again a non-local Hamiltonian structure.)
This result was first stated in \cite{Mag80}
and its partial proof was given in \cite{FF81}
(of course, without having rigorous definitions).

Next, we give a rigorous definition of a non-local symplectic structure and prove 
(the ``well-known'' fact) that, if $S$ is invertible as pseudodifferential operator,
then it is a non-local symplectic structure if and only if $S^{-1}$ is a non-local Hamiltonian structure
(Theorem \ref{20111012:thm}).
Since we completely described (local) symplectic structures in \cite{BDSK09},
this result provides a large collection of non-local Hamiltonian structures.
We also establish a connection between Dirac structures (see \cite{Dor93} and \cite{BDSK09})
with non-local Hamiltonian structures
(Theorems \ref{20111020:thm} and \ref{20120126:prop2}).

After that we discuss the Lenard-Margi scheme of integrability
for a pair of compatible non-local Hamiltonian structures,
similar to that discussed in \cite{Mag78,Mag80,Dor93,BDSK09} in the local case,
and give sufficient conditions when this scheme works
(Theorem \ref{20120127:thm} and Remark \ref{20120201:rem2}).

Finally, we demonstrate on the example of the NLS hierarchy how this scheme works.
The corresponding pair of compatible non-local Hamiltonian structures has been written down 
already in \cite{Mag80}:
$$
H=\left(\begin{array}{cc} 
v\partial^{-1}\circ v & -v\partial^{-1}\circ u \\
-u\partial^{-1}\circ v & u\partial^{-1}\circ u
\end{array}\right)+c\partial\id
\,\,,\,\,\,\,
K=\left(\begin{array}{cc} 0 & -1 \\ 1 & 0 \end{array}\right)\,.
$$
In our next paper \cite{DSK12} we go systematically over well known 
non-local bi-Hamiltonian structures
and construct the corresponding integrable hierarchies of Hamiltonian PDE's.

Throughout the paper, all vector spaces are considered over a field $\mb F$ of characteristic zero.

We wish to thank Pavel Etingof and Andrea Maffei for useful discussions.


\section{Rational pseudodifferential operators}
\label{sec:2}

\subsection{The space $V_{\lambda,\mu}$}
\label{sec:2.1}

Throughout the paper we shall use the following standard notation.
Given a vector space $V$, we denote by $V[\lambda]$ the space of polynomials in $\lambda$ with coefficients in $V$,
by $V[[\lambda^{-1}]]$ the space of formal power series in $\lambda^{-1}$ with coefficients in $V$,
and by $V((\lambda^{-1}))=V[[\lambda^{-1}]][\lambda]$ the space of formal Laurent series in $\lambda^{-1}$ 
with coefficients in $V$.

We have the obvious identifications $V[\lambda,\mu]=V[\lambda][\mu]=V[\mu][\lambda]$
and $V[[\lambda^{-1},\mu^{-1}]]=V[[\lambda^{-1}]][[\mu^{-1}]]=V[[\mu^{-1}]][[\lambda^{-1}]]$.
However the space $V((\lambda^{-1}))((\mu^{-1})$ does not coincide 
with $V((\mu^{-1}))((\lambda^{-1}))$.
Both spaces contain naturally the subspace $V[[\lambda^{-1},\mu^{-1}]][\lambda,\mu]$.
In fact, this subspace is their intersection in the ambient space $V[[\lambda^{\pm1},\mu^{\pm1}]]$
consisting of all series of the form $\sum_{m,n\in\mb Z}a_{m,n}\lambda^m\mu^n$.

The most important for this paper will be the space
$$
V_{\lambda,\mu}:=V[[\lambda^{-1},\mu^{-1},(\lambda+\mu)^{-1}]][\lambda,\mu]\,,
$$
namely, the quotient of the $\mb F[\lambda,\mu,\nu]$-module
$V[[\lambda^{-1},\mu^{-1},\nu^{-1}]][\lambda,\mu,\nu]$
by the submodule 
$(\nu-\lambda-\mu)V[[\lambda^{-1},\mu^{-1},\nu^{-1}]][\lambda,\mu,\nu]$.
By definition, the space $V_{\lambda,\mu}$ consists of elements which can be written (NOT uniquely) in the form
\begin{equation}\label{20111006:eq1}
A=\sum_{m=-\infty}^M\sum_{n=-\infty}^N\sum_{p=-\infty}^P a_{m,n,p}\lambda^m\mu^n(\lambda+\mu)^p\,,
\end{equation}
for some $M,N,P\in\mb Z$ (in fact, we can always choose $P\leq 0$), and $a_{m,n,p}\in V$.

In the space $V[[\lambda^{-1},\mu^{-1},\nu^{-1}]][\lambda,\mu,\nu]$
we have a natural notion of degree, by letting $\deg(\lambda)=\deg(\mu)=\deg(\nu)=1$.
Every element $A\in V[[\lambda^{-1},\mu^{-1},\nu^{-1}]][\lambda,\mu,\nu]$
decomposes as a sum
$A=\sum_{d=-\infty}^NA^{(d)}$
(possibly infinite), where $A^{(d)}$ is a finite linear combination of monomials of degree $d$.
Since $\nu-\lambda-\mu$ is homogenous (of degree 1),
this induces a well-defined notion of degree on the quotient space $V_{\lambda,\mu}$,
and we denote by $V_{\lambda,\mu}^d$, for $d\in\mb Z$,
the span of elements of degree $d$ in $V_{\lambda,\mu}$.
If $A\in\mc V_{\lambda,\mu}$ has the form \eqref{20111006:eq1}, then it decomposes as $A=\sum_{d=-\infty}^{M+N+P}A^{(d)}$,
where $A^{(d)}\in V_{\lambda,\mu}^d$ is given by
$$
A^{(d)}=\sum_{\substack{m\leq M,n\leq N,p\leq P \\ (m+n+p=d)}} a_{m,n,p}\lambda^m\mu^n(\lambda+\mu)^p\,.
$$
The coefficients $a_{m,n,p}\in V$ are still not uniquely defined,
but now the sum in $A^{(d)}$ is finite (since $d-2K\leq m,n,p\leq K:=\max(M,N,P)$).
Hence, we have the following equality
$$
V^d_{\lambda,\mu}=V[\lambda^{\pm1},\mu^{\pm1},(\lambda+\mu)^{-1}]^d\,,
$$
where, as before, the superscript $d$ denotes the subspace consisting of polynomials
in $\lambda^{\pm1},\mu^{\pm1},(\lambda+\mu)^{-1}$, of degree $d$.
\begin{lemma}\label{20110919:lem1}
The following is a basis of the space $V^d_{\lambda,\mu}$ over $V$:
$$
\lambda^{d-i}\mu^i,\,i\in\mb Z
\,\,\,\,;\,\,\,\,\,\,\,\,
\lambda^{d+i}(\lambda+\mu)^{-i},\,i\in\mb Z_{>0}=\{1,2,\dots\}\,,
$$
in the sense that any element of the space $V^d_{\lambda,\mu}$
can be written uniquely as a finite linear combination 
of the above elements with coefficients in $V$.
\end{lemma}
\begin{proof}
First, it suffices to prove the claim for $d=0$.
In this case, letting $t=\mu/\lambda$, the elements of $V^0_{\lambda,\mu}$
are rational functions in $t$ with poles at 0 and -1.
But any such rational functions can be uniquely written,
by partial fractions decomposition, as a linear combination of $t^i$, with $i\in\mb Z$,
and of $(1+t)^{-i}$, with $i\in\mb Z_{>0}$.
\end{proof}
\begin{remark}
One has natural embeddings of $V_{\lambda,\mu}$
in all the spaces 
$V((\lambda^{-1}))((\mu^{-1}))$, $V((\mu^{-1}))((\lambda^{-1}))$,
$V((\lambda^{-1}))(((\lambda+\mu)^{-1}))$, $V((\mu^{-1}))(((\lambda+\mu)^{-1}))$,
$V(((\lambda+\mu)^{-1}))((\lambda^{-1}))$, $V(((\lambda+\mu)^{-1}))((\mu^{-1}))$,
defined by expanding one of the variables $\lambda,\mu$ or $\nu=\lambda+\mu$
in terms of the other two.
For example, we have the embedding
\begin{equation}\label{20110919:eq1}
\iota_{\mu,\lambda}:\,V_{\lambda,\mu}\hookrightarrow V((\lambda^{-1}))((\mu^{-1}))\,,
\end{equation}
obtained by expanding all negative powers of $\lambda+\mu$ in the region $|\mu|>|\lambda|$:
\begin{equation}\label{20110919:eq1b}
\iota_{\mu,\lambda}(\lambda+\mu)^{-n-1}
=
\sum_{k=0}^\infty\binom{-n-1}k \lambda^k\mu^{-n-k-1}\,.
\end{equation}
Similarly in all other cases.
Note that, even though $V_{\lambda,\mu}$ is naturally embedded in both spaces
$V((\lambda^{-1}))((\mu^{-1}))$ and $V((\mu^{-1}))((\lambda^{-1}))$,
it is not contained in their intersection $V[[\lambda^{-1},\mu^{-1}]][\lambda,\mu]$.
\end{remark}

\begin{lemma}\label{20111006:lem}
If $V$ is an algebra, then $V_{\lambda,\mu}$ is also an algebra, with the obvious product.
Namely, if $A(\lambda,\mu),B(\lambda,\mu)\in V_{\lambda,\mu}$, then $A(\lambda,\mu)B(\lambda,\mu)\in V_{\lambda,\mu}$.
More generally, 
if $S,T:\,V\to V$ are endomorphisms of $V$ (viewed as a vector space), then
$$
A(\lambda+S,\mu+T)B(\lambda,\mu)\in V_{\lambda,\mu}\,,
$$
where we expand the negative powers of $\lambda+S$ and $\mu+T$ in non-negative powers of $S$ and $T$, 
acting on the coefficients of $B$.
\end{lemma}
\begin{proof}
We expand $A$ and $B$ as in \eqref{20111006:eq1}:
$$
\begin{array}{l}
\displaystyle{
A(\lambda,\mu)=\sum_{m=-\infty}^M\sum_{n=-\infty}^N\sum_{p=-\infty}^P a_{m,n,p}\lambda^m\mu^n(\lambda+\mu)^p\,,
} \\
\displaystyle{
B(\lambda,\mu)=\sum_{m'=-\infty}^{M'}\sum_{n'=-\infty}^{N'}\sum_{p'=-\infty}^{P'} b_{m',n',p'}\lambda^{m'}\mu^{n'}(\lambda+\mu)^{p'}\,.
}
\end{array}
$$
Using binomial expansion, we then get
$$
A(\lambda+S,\mu+T)B(\lambda,\mu)
=
\sum_{\bar{m}=-\infty}^{M+M'}\sum_{\bar{n}=-\infty}^{N+N'}\sum_{\bar{p}=-\infty}^{P+P'} c_{\bar{m},\bar{n},\bar{p}}
\lambda^{\bar{m}}\mu^{\bar{n}}(\lambda+\mu)^{\bar{p}}\,,
$$
where
$$
\begin{array}{r}
\displaystyle{
c_{\bar{m},\bar{n},\bar{p}}=
\sum_{\substack{m\leq M,m'\leq M',i\geq0 \\ (m+m'-i=\bar{m})}}
\,
\sum_{\substack{n\leq N,n'\leq N',j\geq0 \\ (n+n'-j=\bar{n})}}
\sum_{\substack{p\leq P,p'\leq P',k\geq0 \\ (p+p'-k=\bar{p})}}
} \\
\displaystyle{
\binom{m}{i} \binom{n}{j} \binom{p}{k}
a_{m,n,p} \big(S^{i} T^{j} (S+T)^{k} b_{m',n',p'}\big)\,.
}
\end{array}
$$
To conclude, we just observe that each sum in the RHS is finite,
since, for example, in the first sum we have $i=m+m'-\bar{m}$, $\bar{m}-M'\leq m\leq M$ and $\bar{m}-M\leq m'\leq M'$.
\end{proof}
\begin{lemma}\label{20120131:lem1}
Let $V$ be a vector space and let $U\subset V$ be a subspace.
Then we have:
$$
\begin{array}{l}
\big\{A\in V_{\lambda,\mu}\,\big|\,\iota_{\mu,\lambda}A\in U((\lambda^{-1}))((\mu^{-1}))\big\}
\\
=
\big\{A\in V_{\lambda,\mu}\,\big|\,\iota_{\lambda,\mu}A\in U((\mu^{-1}))((\lambda^{-1}))\big\}
\\
=
\big\{A\in V_{\lambda,\mu}\,\big|\,\iota_{\lambda+\mu,\lambda}A
\in U((\lambda^{-1}))(((\lambda+\mu)^{-1}))\big\}
\\
=
\big\{A\in V_{\lambda,\mu}\,\big|\,\iota_{\lambda+\mu,\mu}A
\in U((\mu^{-1}))(((\lambda+\mu)^{-1}))\big\}
\\
=
\big\{A\in V_{\lambda,\mu}\,\big|\,\iota_{\lambda,\lambda+\mu}A
\in U(((\lambda+\mu)^{-1}))((\lambda^{-1}))\big\}
\\
=
\big\{A\in V_{\lambda,\mu}\,\big|\,\iota_{\mu,\lambda+\mu}A
\in U(((\lambda+\mu)^{-1}))((\mu^{-1}))\big\}
= U_{\lambda,\mu}\,.
\end{array}
$$
\end{lemma}
\begin{proof}
We only need to prove that 
$\big\{A\in V_{\lambda,\mu}\,\big|\,\iota_{\mu,\lambda}A
\in U((\lambda^{-1}))((\mu^{-1}))\big\}\subset U_{\lambda,\mu}$.
Indeed, the opposite inclusion is obvious,
and the argument for the other equalities is the same.

Let $A\in V^d_{\lambda,\mu}$ be such that its expansion 
$\iota_{\mu,\lambda}A\in V((\lambda^{-1}))((\mu^{-1}))$ has coefficients in $U$.
We want to prove that $A$ lies in $U_{\lambda,\mu}$.
By Lemma \ref{20110919:lem1} $A$ can be written uniquely as
$$
A=\sum_{i=-M}^Nv_i\lambda^{d+i}\mu^{-i}+\sum_{j=1}^N w_j\lambda^{d+j}(\lambda+\mu)^{-j}\,,
$$
with $v_i,w_j\in V$.
Its expansion in $V((\lambda^{-1}))((\mu^{-1}))$ is
$$
\iota_{\mu,\lambda}A=
\sum_{i=-M}^Nv_i\lambda^{d+i}\mu^{-i}
+\sum_{j=1}^N\sum_{k=0}^\infty\binom{-j}{k}w_j\lambda^{d+j+k}\mu^{-j-k}\,.
$$
Since, by assumption, $\iota_{\mu,\lambda}A\in U((\lambda^{-1}))((\mu^{-1}))$,
we have
$$
\begin{array}{ll}
v_i\in U &\quad\text{ for }\quad -M\leq i\leq -1 \,,\\
\displaystyle{
v_i+\sum_{j=1}^i \binom{-j}{i-j}w_j
\in U 
} &\quad\text{ for }\quad 0\leq i\leq N  \,,\\
\displaystyle{
\sum_{j=1}^N \binom{-j}{i-j}w_j
\in U 
} &\quad\text{ for }\quad i>N \,.
\end{array}
$$
From the first condition above we have that $v_i$ lies in $U$ for $i<0$.
Using the third condition, we want to deduce that $w_j$ lies in $U$ for all $1\leq j\leq N$.
It then follows, from the second condition, that $v_i$ lies in $U$ for $i\geq0$ as well,
proving the claim.

For $i>N$ and $1\leq j\leq N$ we have $\binom{-j}{i-j}=(-1)^{i-j}\binom{i-1}{j-1}$.
Hence, we will be able to deduce that $w_j$ lies in $U$ for every $j$,
once we prove that the following matrix
$$
P=\left(\,(-1)^{i+j}\binom{i-1}{j-1}\,\right)_{\substack{N+1\leq i<\infty \\ 1\leq j\leq N }}\,,
$$
has rank $N$.
First, the sign $(-1)^{i+j}$ does not change the rank of the above matrix.
So it sufficies to prove that the matrices
$$
T_h=\left(\,\binom{i-1}{j-1}\,\right)_{\substack{h+1\leq i\leq h+N \\ 1\leq j\leq N }}\,,
$$
are non-degenerate for every $h\geq0$.
This is clear since the matrix $T_0$ is upper triangular with ones on the diagonal, 
and, by the Tartaglia-Pascal triangle, $T_h$ and $T_{h+1}$ have the same determinant.
\end{proof}

\subsection{Rational pseudodifferential operators}
\label{sec:2.2a}

For the rest of this section, 
let $\mc A$ be a differential algebra, i.e. a unital commutative associative algebra
with a derivation $\partial$, and assume that $\mc A$ is a domain.
For $a\in\mc A$, we denote $a'=\partial(a)$ and $a^{(n)}=\partial^n(a)$, for a non negative integer $n$.
We denote by $\mc K$ the field of fractions of $\mc A$.

Recall that a \emph{pseudodifferential operator} over $\mc A$
is an expression of the form
\begin{equation}\label{20111003:eq1}
A=A(\partial)
=\sum_{n=-\infty}^N a_n \partial^n
\,\,,\,\,\,\, a_n\in\mc A\,.
\end{equation}
If $a_N\neq0$, one says that $A$ has \emph{order} $N$.
Pseudodifferential operators form a unital associative algebra, 
denoted by $\mc A((\partial^{-1}))$,
with product $\circ$ defined by letting
\begin{equation}\label{20111130:eq1}
\partial^n\circ a=\sum_{j\in\mb Z_+}\binom nj a^{(j)}\partial^{n-j}
\,\,,\,\,\,\, n\in\mb Z,\, a\in\mc A\,.
\end{equation}
We will often omit $\circ$ if no confusion may arise.

Clearly, $\mc K((\partial^{-1}))$ is a skewfield,
and it is the skewfield of fractions of $\mc A((\partial^{-1}))$.
If $A\in\mc A((\partial^{-1}))$ is a non-zero pseudodifferential operator
of order $N$ as in \eqref{20111003:eq1},
its inverse $A^{-1}\in\mc K((\partial^{-1}))$ is computed as follows.
We write
$$
A
=a_N\Big(1+\sum_{n=-\infty}^{-1} a_N^{-1}a_{n+N} \partial^n\Big)\partial^N\,,
$$
and expanding by geometric progression, we get
\begin{equation}\label{20111130:eq2}
A^{-1}
=
\partial^{-N}\circ \sum_{k=0}^\infty\Big(-\sum_{n=-\infty}^{-1} a_N^{-1}a_{n+N} \partial^n\Big)^k\circ a_N^{-1}\,,
\end{equation}
which is well defined as a pseudodifferential operator in $\mc K((\partial^{-1}))$,
since, by formula \eqref{20111130:eq1},
the powers of $\partial$ are bounded above by $-N$,
and the coefficient of each power of $\partial$ is a finite sum.

The \emph{symbol} of the pseudodifferential operator $A(\partial)$ in \eqref{20111003:eq1}
is the formal Laurent series
$A(\lambda)=\sum_{n=-\infty}^N a_n \lambda^n\,\in\mc A((\lambda^{-1}))$,
where $\lambda$ is an indeterminate commuting with $\mc A$.
We thus get a bijective map $\mc A((\partial^{-1}))\to\mc A((\lambda^{-1}))$
(which is not an algebra homomorphism).
A closed formula for the associative product in $\mc A((\partial^{-1}))$
in terms of the corresponding symbols is the following:
\begin{equation}\label{20111003:eq2}
(A\circ B)(\lambda)=A(\lambda+\partial)B(\lambda)\,.
\end{equation}
Here and further on, we always expand an expression as $(\lambda+\partial)^{n},\,n\in\mb Z$, 
in non-negative powers of $\partial$:
\begin{equation}\label{20111004:eq1}
(\lambda+\partial)^{n}=\sum_{j=0}^\infty\binom nj \lambda^{n-j}\partial^j\,.
\end{equation}
Therefore, the RHS of \eqref{20111003:eq2} means
$\sum_{m,n=-\infty}^N\sum_{j=0}^\infty \binom{m}{j}a_m b_n^{(j)} \lambda^{m+n-j}$.
%

The algebra $\mc A((\partial^{-1}))$
contains the algebra of \emph{differential operators} $\mc A[\partial]$ as a subalgebra.
\begin{definition}\label{20110926:def}
The field $\mc K(\partial)$ of \emph{rational pseudodifferential operators}
is the smallest subskewfield of $\mc K((\partial^{-1}))$ containing $\mc A[\partial]$.
We denote $\mc A(\partial)=\mc K(\partial)\cap\mc A((\partial^{-1}))$,
the subalgebra of \emph{rational pseudodifferential operators with coefficients in} $\mc A$.
\end{definition}
The following Proposition (see \cite[Prop.3.4]{CDSK12}) describes explicitly 
the skewfield $\mc K(\partial)$ of rational pseudodifferential operators.
\begin{proposition}\label{20111003:thm2}
Let $\mc A$ be a differential domain, and let $\mc K$ be its field of fractions.
\begin{enumerate}[(a)]
\item
Every rational pseudodifferential operator $L\in\mc K(\partial)$
can be written as a right (resp. left) fraction
$L=AS^{-1}$ (resp. $L=S^{-1}A$)
for some $A,S\in\mc A[\partial]$ with $S\neq0$.
\item
Let $L=AS^{-1}$ (resp. $L=S^{-1}A$), with $A,S\in\mc A[\partial]$, 
be a decomposition of $L\in\mc K(\partial)$
such that $S$ has minimal possible order.
Then any other decomposition $L=A_1S_1^{-1}$
(resp. $L=S_1^{-1}A_1$),
with $A_1,S_1\in\mc A[\partial]$,
we have $A_1=AK$, $S_1=SK$ 
(resp. $A_1=KA$, $A_1=KS$), 
for some $K\in\mc K[\partial]$.
\end{enumerate}
\end{proposition}
%

\subsection{Rational matrix pseudodifferential operators}
\label{sec:2.3}

\begin{definition}\label{20111013:def}
A matrix pseudodifferential operator $A\in\Mat_{\ell\times\ell}\mc A((\partial^{-1}))$
is called \emph{rational with coefficients in $\mc A$}  
if its entries are rational pseudodifferential operators with coefficients in $\mc A$.
In other words, 
the algebra of rational matrix pseudodifferential operators with coefficients in $\mc A$ 
is $\Mat_{\ell\times\ell}\mc A(\partial)$.
\end{definition}
Let $M=\big(A_{ij}B_{ij}^{-1}\big)_{i,j\in I}$
be a rational matrix pseudodifferential operator with coefficients in $\mc A$, 
with $A_{ij},B_{ij}\in\mc A[\partial]$.
By the Ore condition (see e.g. \cite{CDSK12}), 
we can find a common right multiple $B\in\mc A[\partial]$
of all operators $B_{ij}$,
i.e. for every $i,j$ we can factor $B=B_{ij}C_{ij}$ for some $C_{ij}\in\mc A[\partial]$.
Hence, $A_{ij}B_{ij}^{-1}=\tilde A_{ij}B^{-1}$, where $\tilde A_{ij}=A_{ij}C_{ij}$.
Then, the matrix $M$ can be represented as a ratio of two matrices:
$M=\tilde A (B\id)^{-1}$.
Hence,
$$
\Mat{}_{\ell\times\ell}\mc A(\partial)
=
\left\{A(B\id)^{-1}\,\left|\,
\begin{array}{c}
A\in\Mat{}_{\ell\times\ell}\mc A[\partial],\,B\in\mc A[\partial],\\
A_{ij}B^{-1}\in\mc A((\partial^{-1}))\,\,\forall i,j
\end{array}
\right.\right\}\,.
$$
However, in general this is not a representation of the rational matrix $M$ in ``minimal terms''
(see Definition \ref{def:minimal-fraction} below).

We recall now some linear algebra over the skewfield $\mc K((\partial^{-1}))$
and, in particular, the notion of the Dieudonn\'e determinant
(see \cite{Art57} for an overview over an arbitrary skewfield).

An \emph{elementary row operation} of an $\ell\times\ell$ matrix pseudodifferential operator 
$A$ is either a permutation of two rows of it,
or the operation $\mc T(i,j;P)$, where $1\leq i\neq j\leq m$
and $P\in\mc K((\partial^{-1}))$,
which replaces the $j$-th row by itself minus $i$-th row
multiplied on the left by $P$.
Using the usual Gauss elimination, we can get the (well known) analogues
of standard linear algebra theorems for matrix pseudodifferential operators.
In particular, 
any matrix pseudodifferential operator $A\in\Mat_{m\times\ell}\mc K((\partial^{-1}))$
can be brought by elementary row operations to a row echelon form.

The \emph{Dieudonn\'e determinant} of a $A\in\Mat_{\ell\times\ell}\mc K((\partial^{-1}))$
has the form $\det A=c\xi^d$, where $c\in\mc A$, $\xi$ is an indeterminate, 
and $d\in\mb Z$.
It is defined by the following properties:
$\det A$ changes sign if we permute two rows of $A$,
and it is unchanged under any elementary row operation $\mc T(i,j;P)$ defined above,
for aribtrary $i\neq j$ and a pseudodifferential operator $P\in\mc K((\partial^{-1}))$;
moreover, if $A$ is upper triangular,
with diagonal entries $A_{ii}$ of order $n_i$ and leading coefficoent $a_i$,
then 
$$
\det A=\Big(\prod_i a_i\Big) \xi^{\sum_in_i}\,.
$$
It was proved in \cite{Die43} (for any skewfield) that the Dieudonn\'e determinant is well defined
and $\det(A\circ B)=(\det A)(\det B)$
for every $\ell\times\ell$ matrix pseudodifferential operators 
$A,B\in\Mat_{\ell\times\ell}\mc K((\partial^{-1}))$.

The Dieudonn\'e determinant gives a way to characterize invertible matrix pseudodifferential operators,
thanks to the following well known fact (see e.g. \cite[Prop.4.3]{CDSK12}):
\begin{proposition}\label{20111005:prop2}
Let $\mc D$ be a subskewfield of the skewfield $\mc K\!((\partial^{-1}))$,\
and let $A\in\!\Mat_{\ell\times\ell}\mc D$.
Then $A$ is invertible in $\Mat_{\ell\times\ell}\mc D$ if and only if $\det A\neq0$.
\end{proposition}
\begin{corollary}\label{20111005:prop3}
Let $A\in\Mat_{\ell\times\ell}\mc K((\partial^{-1}))$
be a matrix with $\det A\neq0$.
Then $A$ is a rational matrix if and only if $A^{-1}$ is.
\end{corollary}
\begin{proof}
It is a special case of Proposition \ref{20111005:prop2} when $\mc D$
is the subskewfield $\mc K(\partial)\subset\mc K((\partial^{-1}))$ of rational pseudodifferential operators.
\end{proof}
\begin{remark}\label{20111216:rem}
It is proved in \cite{CDSK12} that,
if $A\in\Mat_{\ell\times\ell}\mc A((\partial^{-1}))$
then we have $\det A=c\xi^d$, with $c\in\bar{\mc A}$, 
where $\bar{\mc A}$ is the integral closure of $\mc A$.
Furthermore, if $c$ is an invertible element of $\bar{\mc A}$,
then the inverse matrix $A^{-1}$
lies in $\Mat_{\ell\times\ell}\bar{\mc A}((\partial^{-1}))$.
\end{remark}

\begin{definition}[see \cite{CDSK12b}]\label{def:minimal-fraction}
Let $H\in\Mat_{\ell\times\ell}\mc K(\partial)$ be a rational matrix pseudodifferential 
operator with coefficients in the differential field $\mc K$.
A fractional decomposition $H=A B^{-1}$,
with $A,B\in\Mat_{\ell\times\ell}\mc K[\partial]$, $\det B\neq0$,
is called \emph{minimal} if $\deg_\xi \det B$ is minimal
(recall that it is a non-negative integer).
\end{definition}
\begin{proposition}[\cite{CDSK12b}]\label{prop:minimal-fraction}
\begin{enumerate}[(a)]
\item
A fractional decomposition $H\!\!=\!\!A B^{-1}$
of a rational matrix pseudodifferential operator $H\in\Mat_{\ell\times\ell}\mc K(\partial)$
is minimal if and only if
\begin{equation}\label{20120124:eq3}
\ker A\cap\ker B=0\,,
\end{equation}
in any differential field extension of $\mc K$.
\item
The minimal fractional decomposition of $H$ exists and is
unique up to multiplication of $A$ and $B$ on the right 
by a matrix differential operator $D$
which is invertible in the algebra $\Mat_{\ell\times\ell}\mc K[\partial]$.
\end{enumerate}
\end{proposition}
\begin{remark}\label{rem:minimal-fraction}
In the case $\ell=1$ the fractional decomposition $H=A B^{-1}\in\mc K(\partial)$,
is minimal if and only if the the differential operators $A,B\in\mc K[\partial]$
have no right common divisor of order greater than 0
(i.e. the right greatest common divisor of $A$ and $B$ is 1).
\end{remark}
\begin{remark}\label{rem:minimal-fraction2}
Let $\mc A$ be a differential domain, and let $\mc K$ be its field of fractions.
A fractional decomposition $H=AB^{-1}$ of $H\in\Mat_{\ell\times\ell}\mc A[\partial]$ over $\mc K$
can be turned into a fractional decomposition over $\mc A$ by clearing the denominators of $A$ and $B$.
Hence, a minimal fracitonal decomposition over $\mc A$ is also minimal over $\mc K$.
\end{remark}


\section{Non-local Poisson vertex algebras}
\label{sec:3}

\subsection{Non-local $\lambda$-brackets and non-local Lie conformal algebras}
\label{sec:3.1}

Let $R$ be a module over the algebra of polynomials $\mb F[\partial]$.
\begin{definition}\label{20110919:def1}
A \emph{non-local} $\lambda$-\emph{bracket} on $R$ is a linear map
$\{\cdot\,_\lambda\,\cdot\}:\,R\otimes R\to R((\lambda^{-1}))$
satisfying the following \emph{sesquilinearity} conditions ($a,b\in R$):
\begin{equation}\label{20110921:eq1}
\{\partial a_\lambda b\}=-\lambda\{a_\lambda b\}
\,\,,\,\,\,\,
\{a_\lambda\partial b\}=(\lambda+\partial)\{a_\lambda b\}\,.
\end{equation}
The non-local $\lambda$-bracket $\{\cdot\,_\lambda\,\cdot\}$ is said to be \emph{skewsymmetric} 
(respectively \emph{symmetric})
if ($a,b\in R$)
\begin{equation}\label{20110921:eq2}
\{b_\lambda a\}=-\{a_{-\lambda-\partial}b\}
\quad \Big(\text{ resp. } =\{a_{-\lambda-\partial}b\}\Big)
\,.
\end{equation}
\end{definition}
The RHS of the skewsymmetry condition should be interpreted as follows:
if $\{a_\lambda b\}=\sum_{n=-\infty}^Nc_n\lambda^n$, then
$$
\begin{array}{c}
\displaystyle{
\{a_{-\lambda-\partial} b\}
=
\sum_{n=-\infty}^N(-\lambda-\partial)^n c_n
=
\sum_{n=-\infty}^N\sum_{k=0}^\infty\binom{n}{k}(-1)^n(\partial^k c_n)\lambda^{n-k} 
}\\
\displaystyle{
= \sum_{m=-\infty}^N\Big(\sum_{k=0}^{N-m}\binom{m+k}{k}(-1)^{m+k}(\partial^k c_{m+k})\Big)\lambda^m
\,.
}
\end{array}
$$
In other words, we move $-\lambda-\partial$ to the left and
we expand in non negative powers of $\partial$ as in \eqref{20111004:eq1}.

In general we have $\{a_\lambda\{b_\mu c\}\}\in R((\lambda^{-1}))((\mu^{-1}))$
for an arbitrary $\lambda$-bracket $\{\cdot\,_\lambda\,\cdot\}$.
Recall from Section \ref{sec:2.1} that $R_{\lambda,\mu}$ can be considered
as a subspace of $R((\lambda^{-1}))((\mu^{-1}))$ via the embedding $\iota_{\mu,\lambda}$.
\begin{definition}\label{20110919:def2}
The non-local $\lambda$-bracket $\{\cdot\,_\lambda\,\cdot\}$ on $R$
is called \emph{admissible} if
$$
\{a_\lambda\{b_\mu c\}\}\in R_{\lambda,\mu}
\qquad\forall a,b,c\in R\,.
$$
\end{definition}
\begin{remark}\label{20110919:rem}
If $\{\cdot\,_\lambda\,\cdot\}$ is a skewsymmetric admissible $\lambda$-bracket on $R$,
then $\{b_\mu\{a_\lambda c\}\}\in R_{\lambda,\mu}$ and $\{\{a_\lambda b\}_{\lambda+\mu} c\}\in R_{\lambda,\mu}$
for all $a,b,c\in R$.
Indeed, the first claim is obvious since $R_{\lambda,\mu}=R_{\mu,\lambda}$.
For the second claim, by skewsymmetry 
$\{\{a_\lambda b\}_{\lambda+\mu} c\}=-\{c_{-\lambda-\mu-\partial}\{a_\lambda b\}\}$,
and by the admissibility assumption $\{c_\nu\{a_\lambda b\}\}\in R_{\lambda,\nu}$.
To conclude it suffices to note that when replacing $\nu$ by $-\lambda-\mu-\partial$
in an element of $R_{\lambda,\nu}=R[[\lambda^{-1},\nu^{-1},(\lambda+\nu)^{-1}]][\lambda,\nu]$,
we have that $\nu^{-1}$ is expanded in negative powers of $\lambda+\mu$
and $(\lambda+\nu)^{-1}$ is expanded in negative powers of $\mu$.
As a result, we get an element of $R[[\lambda^{-1},\mu^{-1},(\lambda+\mu)^{-1}]][\lambda,\mu]=R_{\lambda,\mu}$.
\end{remark}
\begin{definition}\label{20110921:def1}
A \emph{non-local Lie conformal algebra} is an $\mb F[\partial]$-module $R$
endowed with an admissible skewsymmetric $\lambda$-bracket
$\{\cdot\,_\lambda\,\cdot\}:\,R\otimes R\to R((\lambda^{-1}))$
satisfying the Jacobi identity (in $R_{\lambda,\mu}$):
\begin{equation}\label{20110922:eq3}
\{a_\lambda\{b_\mu c\}\}-\{b_\mu\{a_\lambda c\}\}=\{\{a_\lambda b\}_{\lambda+\mu} c\}
\,\,\,\,\text{ for every } a,b,c\in R\,.
\end{equation}
\end{definition}
\begin{example}\label{20110921:ex1}
Let $R=\big(\mb F[\partial]\otimes V\big)\oplus\mb FC$,
where $V$ is a vector space with a symmetric bilinear form $(\cdot\,|\,\cdot)$.
Define the (non-local) $\lambda$-bracket on $R$ by
letting $C$ be a central element, defining
$$
\{a_\lambda b\}=(a|b)C\lambda^{-1}
\,\,\,\,\text{ for } a,b\in V\,,
$$
and extending it to a $\lambda$-bracket on $R$ by sesquilinearity.
Skewsymmetry for this $\lambda$-bracket holds since, by assumption, $(\cdot\,|\,\cdot)$ is symmetric.
Moreover, since any triple $\lambda$-bracket is zero,
the $\lambda$-bracket is obviously admissible and it satisfies the Jacobi identity.
Hence, we have a non-local Lie conformal algebra.
\end{example}

\subsection{Non-local Poisson vertex algebras}
\label{sec:3.2}

Let $\mc V$ be a differential algebra, i.e. a unital commutative associative algebra
with a derivation $\partial:\,\mc V\to\mc V$.
As before, we assume that $\mc V$ is a domain and denote by $\mc K$ its field of fractions.
\begin{definition}\label{20110921:def2}
\begin{enumerate}[(a)]
\item 
A \emph{non-local} $\lambda$-\emph{bracket} on the differential algebra $\mc V$ is a linear map
$\{\cdot\,_\lambda\,\cdot\}:\,\mc V\otimes \mc V\to \mc V((\lambda^{-1}))$
satisfying the sesquilinearity conditions \eqref{20110921:eq1}
and the following left and right \emph{Leibniz rules}:
\begin{equation}\label{20110921:eq3}
\begin{array}{l}
\{a_\lambda bc\}=b\{a_\lambda c\}+c\{a_\lambda b\}\,, \\
\{ab_\lambda c\}=\{a_{\lambda+\partial}c\}_\to b+\{b_{\lambda+\partial} c\}_\to a\,.
\end{array}
\end{equation}
Here and further an expression $\{a_{\lambda+\partial}b\}_\to c$ is interpreted as follows:
if $\{a_{\lambda}b\}=\sum_{n=-\infty}^Nc_n\lambda^n$, 
then $\{a_{\lambda+\partial}b\}_\to c=\sum_{n=-\infty}^Nc_n(\lambda+\partial)^nc$,
where we expand $(\lambda+\partial)^nc$ in non-negative powers of $\partial$ as in \eqref{20111004:eq1}.
\item
The conditions of (\emph{skew})\emph{symmetry}, \emph{admissibility} and \emph{Jacobi identity} 
for a non-local $\lambda$-bracket $\{\cdot\,_\lambda\,\cdot\}$ on $\mc V$
are the same as in Definitions \ref{20110919:def1}, \ref{20110919:def2} and \ref{20110921:def1} respectively.
\item
A \emph{non-local Poisson vertex algebra} is a differential algebra $\mc V$
endowed with a \emph{non-local Poisson} $\lambda$-\emph{bracket},
i.e. a skewsymetric admissible non-local $\lambda$-bracket,
satisfying the Jacobi identity.
\end{enumerate}
\end{definition}
\begin{example}[cf. Example \ref{20110921:ex1}]\label{20110921:ex2}
Let $\mc V=\mb F[u_i^{(n)}\,|\,i=1,\dots,\ell,n\in\mb Z_+]$ 
be the algebra of diffenertial polynoamials in $\ell$ differential variables $u_1,\dots,u_\ell$.
Let $C=\big(c_{ij}\big)_{i,j=1}^\ell$ be an $\ell\times\ell$ symmetric matrix 
with coefficients in $\mb F$.
The following formula defines a structure of a non-local Poisson vertex algebra on $\mc V$:
$$
\{P_\lambda Q\}
=
\sum_{m,n\in\mb Z_+}\sum_{i,j\in\mb Z_+} c_{ij}
\frac{\partial Q}{\partial u_j^{(n)}} (-1)^m 
(\lambda+\partial)^{m+n-1} 
\frac{\partial P}{\partial u_i^{(m)}}\,.
$$
For example, $\{{u_i}_\lambda {u_j}\}=c_{ij}\lambda^{-1}$ but, for 
$P,Q\in\mb F[u_1,\dots,u_\ell]\subset\mc V$,
we get an infinite formal Laurent series in $\lambda^{-1}$:
$$
\begin{array}{l}
\displaystyle{
\{P_\lambda Q\}
=
\sum_{i,j=1}^\ell c_{ij} \frac{\partial Q}{\partial u_j} 
(\lambda+\partial)^{-1} \frac{\partial P}{\partial u_i}
} \\
\displaystyle{
=\sum_{i,j=1}^\ell \sum_{n=0}^\infty (-1)^n  \frac{\partial Q}{\partial u_j}
\Big(\partial^n \frac{\partial P}{\partial u_i}\Big) \lambda^{-n-1}
\in\mc V((\lambda^{-1}))\,.
}
\end{array}
$$
We will prove that this is indeed a non-local Poisson $\lambda$-bracket 
in the next section,
where we will discuss a general construction of non-local Poisson vertex algebras,
which will include this example as a special case
(see Theorem \ref{20110923:prop}).
\end{example}

\begin{proposition}\label{20111219:prop}
Let $\{\cdot\,_\lambda\,\cdot\}$ be a non-local Poisson vertex algebra structure on the differential domain $\mc V$.
Then there is a unique way to extend it to a non-local Poisson vertex algebra structure
on the differential field of fractions $\mc K$,
and it can be computed using the following formulas ($a,b\in\mc K\backslash\{0\}$):
\begin{equation}\label{20111219:eq1}
\{a_\lambda b^{-1}\}=-b^{-2}\{a_\lambda b\}
\,\,,\,\,\,\,
\{a^{-1}_\lambda b\}=-\{a_{\lambda+\partial} b\}_\to a^{-2}\,.
\end{equation}
\end{proposition}
\begin{proof}
It is straightforward to check that formulas \eqref{20111219:eq1} define
a non-local $\lambda$-bracket on the field of fraction $\mc K$,
satysfying all the axioms of non-local Poisson vertex algebra.
In particular, admissibility of the $\lambda$-bracket can be derived from Lemma \ref{20111006:lem}.
The details of the proof are left to the reader.
\end{proof}
%

Thanks to Proposition \ref{20111219:prop}
we can extend, uniquely, a non-local Poisson vertex algebra $\lambda$-bracket
on $\mc V$ to its field of fractions $\mc K$.
The following results are useful to prove admissibility of a non-local $\lambda$-bracket.
\begin{lemma}\label{20111012:lem}
Let $\mc V$ be a differential algebra, endowed 
with a non-local $\lambda$-bracket $\{\cdot\,_\lambda\,\}$.
Assume that $\mc V$ is a domain, and let $\mc K$ be its field of fractions.
Let $S=\big(S_{ij}\big)_{i,j\in I}\in\Mat_{\ell\times\ell}\big(\mc K((\partial^{-1}))\big)$
be an invertible $\ell\times\ell$ matrix pseudodifferential operator with coefficients in $\mc K$.
Letting $S_{ij}=\sum_{n=-\infty}^N s_{ij;n}\partial^n$,
the following identities hold for every $a\in\mc K$ and $i,j\in I$:
\begin{equation}\label{20111012:eq2a}
\begin{array}{r}
\displaystyle{
\big\{a_\lambda (S^{-1})_{ij}(\mu)\big\}
=
-\sum_{r,t=1}^\ell\sum_{n=-\infty}^N
\iota_{\mu,\lambda}(S^{-1})_{ir}(\lambda+\mu+\partial)
} \\
\displaystyle{
\{a_\lambda s_{rt;n}\} (\mu+\partial)^n (S^{-1})_{tj}(\mu)
\,\in\mc K((\lambda^{-1}))((\mu^{-1}))
\,,
}
\end{array}
\end{equation}
and
\begin{equation}\label{20111012:eq2b}
\begin{array}{l}
\displaystyle{
\big\{(S^{-1})_{ij}(\lambda) _{\lambda+\mu} a\big\}
=
-\sum_{r,t=1}^\ell\sum_{n=-\infty}^N
\{{s_{rt;n}}_{\lambda+\mu+\partial}a\}_\to
} \\
\displaystyle{
\Big((\lambda+\partial)^n (S^{-1})_{tj}(\lambda)\Big)
\iota_{\lambda,\lambda+\mu}({S^*}^{-1})_{ri}(\mu) 
\,\in\mc K(((\lambda+\mu)^{-1}))((\lambda^{-1}))
\,,
}
\end{array}
\end{equation}
where $\iota_{\mu,\lambda}$ and $\iota_{\lambda,\lambda+\mu}$ 
are as in \eqref{20110919:eq1b}.
In equation \eqref{20111012:eq2b} $S^*$ denotes the adjoint 
of the matrix differential operator $S$
(its inverse being $(S^{-1})^*$).
\end{lemma}
\begin{proof}
The identity $S\circ S^{-1}=1$ becomes, in terms of symbols,
$$
\sum_{t=1}^\ell S_{r,t}(\mu+\partial)(S^{-1})_{tj}(\mu)=\delta_{rj}\,.
$$
Taking $\lambda$-bracket with $a$, we have, by sesquilinearity and the (left) Leibniz rule,
$$
\begin{array}{l}
\displaystyle{
0=\sum_{t=1}^\ell \big\{a_\lambda S_{rt}(\mu+\partial)(S^{-1})_{tj}(\mu)\big\}
} \\
\displaystyle{
=
\sum_{t=1}^\ell\sum_{n=-\infty}^N 
\big\{a_\lambda s_{rt;n} (\mu+\partial)^n (S^{-1})_{tj}(\mu)\big\}
} \\
\displaystyle{
=
\sum_{t=1}^\ell\sum_{n=-\infty}^N \{a_\lambda s_{rt;n}\} (\mu+\partial)^n (S^{-1})_{tj}(\mu)
} \\
\displaystyle{
+\sum_{t=1}^\ell \iota_{\mu,\lambda} S_{rt}(\lambda+\mu+\partial) \big\{a_\lambda (S^{-1})_{tj}(\mu)\big\}
\,.
}
\end{array}
$$
Note that $\iota_{\mu,\lambda} S(\lambda+\mu+\partial)$
is invertible in
$\Mat_{\ell\times\ell}\big(\mc K[\partial]((\lambda^{-1}))((\mu^{-1}))\big)$,
its inverse being $\iota_{\mu,\lambda} S^{-1}(\lambda+\mu+\partial)$.
We then apply $\iota_{\mu,\lambda}(S^{-1})_{ir}(\lambda+\mu+\partial)$ on the left to both sides of the above equation
and we sum over $r=1,\dots,\ell$, to get
$$
\begin{array}{l}
\displaystyle{
\sum_{t=1}^\ell \delta_{it} \big\{a_\lambda (S^{-1})_{tj}(\mu)\big\}
} \\
\displaystyle{
=-
\sum_{r=1}^\ell
\sum_{t=1}^\ell\sum_{n=-\infty}^N \iota_{\mu,\lambda}(S^{-1})_{ir}(\lambda+\mu+\partial) \{a_\lambda s_{rt;n}\} (\mu+\partial)^n (S^{-1})_{tj}(\mu)\,,
}
\end{array}$$
proving equation \eqref{20111012:eq2a}.

Similarly, for the second equation we have, by the right Leibniz rule,
$$
\begin{array}{l}
\displaystyle{
0=\sum_{t=1}^\ell \big\{S_{rt}(\lambda+\partial)(S^{-1})_{tj}(\lambda)\,_{\lambda+\mu}a\big\}
} \\
\displaystyle{
=
\sum_{t=1}^\ell \sum_{n=-\infty}^N \big\{s_{rt;n}(\lambda+\partial)^n(S^{-1})_{tj}(\lambda)_{\lambda+\mu}a\big\}
} \\
\displaystyle{
=
\sum_{t=1}^\ell \sum_{n=-\infty}^N \big\{{s_{rt;n}}_{\lambda+\mu+\partial}a\big\}_\to (\lambda+\partial)^n(S^{-1})_{tj}(\lambda)
} \\
\displaystyle{
+ \sum_{t=1}^\ell \sum_{n=-\infty}^N \big\{(S^{-1})_{tj}(\lambda)_{\lambda+\mu+\partial}a\big\}_\to 
\iota_{\lambda,\lambda+\mu}(\lambda-\lambda-\mu-\partial)^ns_{rt;n}
} \\
\displaystyle{
=
\sum_{t=1}^\ell \sum_{n=-\infty}^N \big\{{s_{rt;n}}_{\lambda+\mu+\partial}a\big\}_\to (\lambda+\partial)^n(S^{-1})_{tj}(\lambda)
} \\
\displaystyle{
+ \sum_{t=1}^\ell \big\{(S^{-1})_{tj}(\lambda)_{\lambda+\mu+\partial}a\big\}_\to \iota_{\lambda,\lambda+\mu} S^*_{tr}(\mu)
\,.
}
\end{array}
$$
We next replace in the above equation $\mu$ (placed at the right) by $\mu+\partial$,
and we apply the resulting differential operator to $\iota_{\lambda,\lambda+\mu}({S^*}^{-1})_{ri}(\mu)$.
As a result we get, after summing over $r=1,\dots,\ell$,
$$
\begin{array}{l}
\displaystyle{
\sum_{t=1}^\ell 
\big\{(S^{-1})_{tj}(\lambda)_{\lambda+\mu+\partial}a\big\}_\to \delta_{ti}
} \\
\displaystyle{
=
-\sum_{r=1}^\ell
\sum_{t=1}^\ell 
\sum_{n=0}^N 
\big\{{s_{rt;n}}_{\lambda+\mu+\partial}a\big\}_\to 
\Big((\lambda+\partial)^n(S^{-1})_{tj}(\lambda)\Big)
\iota_{\lambda,\lambda+\mu}({S^*}^{-1})_{ri}(\mu)
\,,
}
\end{array}
$$
proving equation \eqref{20111012:eq2b}.
\end{proof}
\begin{corollary}\label{20111014:cor}
Let $\mc V$ be a differential algebra, endowed 
with a non-local $\lambda$-bracket $\{\cdot\,_\lambda\,\cdot\}$.
Assume that $\mc V$ is a domain, and let $\mc K$ be its field of fractions.
Let $S=\big(S_{ij}\big)_{i,j\in I}\in\Mat_{\ell\times\ell}\big(\mc K[\partial]\big)$
have non-zero Dieudonn\`e determinant.
Then the following identities hold for every $a\in\mc K$ and $i,j\in I$:
\begin{equation}\label{20111012:eq2c}
\begin{array}{l}
\displaystyle{
\big\{a_\lambda (S^{-1})_{ij}(\mu)\big\}
} \\
\displaystyle{
=
-\sum_{r,t=1}^\ell\sum_{n=0}^N
(S^{-1})_{ir}(\lambda+\mu+\partial) \{a_\lambda s_{rt;n}\} (\mu+\partial)^n (S^{-1})_{tj}(\mu)
\,\in\mc K_{\lambda,\mu}
\,,
}
\end{array}
\end{equation}
and
\begin{equation}\label{20111012:eq2d}
\begin{array}{l}
\displaystyle{
\big\{(S^{-1})_{ij}(\lambda) _{\lambda+\mu} a\big\}
} \\
\displaystyle{
=
-\sum_{r,t=1}^\ell\sum_{n=0}^N
\{{s_{rt;n}}_{\lambda+\mu+\partial}a\}_\to
\Big((\lambda+\partial)^n (S^{-1})_{tj}(\lambda)\Big)
({S^*}^{-1})_{ri}(\mu) 
\,\in\mc K_{\lambda,\mu}
\,,
}
\end{array}
\end{equation}
where $S_{ij}=\sum_{n=0}^N s_{ij;n}\partial^n$.
\end{corollary}
\begin{proof}
It is immediate from equations \eqref{20111012:eq2a} and \eqref{20111012:eq2b}.
\end{proof}
\begin{corollary}\label{20111007:prop}
Let $\mc V$ be a differential algebra, endowed 
with a non-local $\lambda$-bracket $\{\cdot\,_\lambda\,\}$.
Assume that $\mc V$ is a domain, and let $\mc K$ be its field of fractions.
Let $A\in\mc V(\partial)=\mc K(\partial)\cap\mc V((\partial^{-1}))$ 
be a rational pseudodifferential operator with coefficients in $\mc V$.
Then 
$\{a_\lambda A(\mu)\}$
and $\{A(\lambda)_{\lambda+\mu} a\}$
lie in $\mc V_{\lambda,\mu}$
for every $a\in\mc V$.
\end{corollary}
\begin{proof}
First, note that if the pseudodifferential operators $A,B\in\mc K((\partial^{-1}))$ 
satisfy the conditions
$$
\{a_\lambda A(\mu)\}
\,,\{A(\lambda)_{\lambda+\mu} a\}
\{a_\lambda B(\mu)\}
\,,\{B(\lambda)_{\lambda+\mu} a\}
\,\,
\in\mc K_{\lambda,\mu}\,,
$$
for every $a\in\mc K$,
so does $A\circ B$.
Indeed, by the Leibniz rule,
$$
\begin{array}{l}
\displaystyle{
\{a_\lambda (A\circ B)(\mu)\}
=
\{a_\lambda A(\mu+\partial) B(\mu)\}
} \\
\displaystyle{
=
\{a_\lambda A(\mu+\partial)\}_{\to} B(\mu)
+ A(\lambda+\mu+\partial) \{a_\lambda B(\mu)\}\,,
}
\end{array}
$$
and both terms in the RHS lie in $\mc K_{\lambda,\mu}$ by the assumption on $A$ and $B$,
thanks to Lemma \ref{20111006:lem}.
Similarly, by the right Leibniz rule,
$$
\begin{array}{l}
\displaystyle{
\{ (A\circ B)(\lambda) _{\lambda+\mu} a \}
=
\{ A(\lambda+\partial) B(\lambda) _{\lambda+\mu} a \}
} \\
\displaystyle{
=
\{ B(\lambda) _{\lambda+\mu+\partial} a \}_\to 
\iota_{\lambda,\lambda+\mu}A^*(\mu)
+
\{ A(\lambda+\partial) _{\lambda+\mu+\partial} a \}_\to B(\lambda)
\,,
}
\end{array}
$$
and both terms in the RHS lie in $\mc K_{\lambda,\mu}$ 
(rather in the image of $\mc K_{\lambda,\mu}$ in $\mc K(((\lambda+\mu)^{-1}))((\lambda^{-1}))$ 
via $\iota_{\lambda,\lambda+\mu}$) by Lemma \ref{20111006:lem}.
By Corollary \ref{20111014:cor}
we have that, if $S\in\mc V[\partial]$, then 
$\{a_\lambda S^{-1}(\mu)\}$ and $\{S^{-1}(\lambda)_{\lambda+\mu} a\}$
lie in $\mc K_{\lambda,\mu}$ for all $a\in\mc K$.
Hence, 
by Definition \ref{20110926:def} and the above observations,
we get that, if $A\in\mc V(\partial)=\mc K(\partial)\cap\mc V((\partial^{-1}))$,
then 
$\{a_\lambda A(\mu)\}$ and $\{A(\lambda)_{\lambda+\mu} a\}$
lie in $\mc K_{\lambda,\mu}$
for all $a\in\mc K$.
On the other hand, if $a\in\mc V$, we clearly have
$\{a_\lambda A(\mu)\}\in\mc V((\lambda^{-1}))((\mu^{-1}))$ 
and $\{A(\lambda)_{\lambda+\mu} a\}\in\mc V(((\lambda+\mu)^{-1}))((\lambda^{-1}))$.
The claim follows from Lemma \ref{20120131:lem1}
applied to $V=\mc K$ and $U=\mc V$.
\end{proof}
\begin{remark}\label{20111104:rem}
In the case when $S\in\mc V(\partial)$ 
is a rational pseudodifferential operator with coefficients in $\mc V$,
thanks to Corollary \ref{20111007:prop},
we can drop $\iota_{\mu,\lambda}$ and $\iota_{\lambda,\lambda+\mu}$ respectively from
equations \eqref{20111012:eq2a} and \eqref{20111012:eq2b},
which hold in the space $\mc V_{\lambda,\mu}$.
\end{remark}


\section{Non-local Hamiltonian structures}
\label{sec:4}

\subsection{Algebras of differential functions}
\label{sec:4.1}

Let $R_\ell=\mb F [u_i^{(n)}\, |\, i \in I,n \in \mb Z_+]$ be the algebra of differential polynomials
in the $\ell$ variables $u_i,\,i\in I=\{1,\dots,\ell\}$,
with the derivation $\partial$ defined by $\partial (u_i^{(n)}) = u^{(n+1)}_i$.
The partial derivatives $\frac{\partial}{\partial u_i^{(n)}}$
are commuting derivations of $R_\ell$, and they satisfy the following 
commutation relations with $\partial$:
\begin{equation}\label{eq:0.4}
\left[  \frac{\partial}{\partial u_i^{(n)}}, \partial \right] = \frac{\partial}{\partial u_i^{(n-1)}}
\,\,\,\,
\text{ (the RHS is $0$ if $n=0$) }\,.
\end{equation}

Recall from \cite{BDSK09} that an \emph{algebra of differential functions} 
is a differential algebra extension $\mc V$ of $R_\ell$,
endowed with commuting derivations
$$
\frac{\partial}{\partial u_i^{(n)}}:\,\mc V\to\mc V
\,\,,\,\,\,\,
i \in I ,\,n \in \mb Z_+\,,
$$
extending the usual partial derivatives on $R_\ell$,
such that only a finite number of
$\frac{\partial f}{\partial u_i^{(n)}}$ are non-zero for each $f\in \mc V$, 
and such that the commutation rules \eqref{eq:0.4} hold on $\mc V$.
Examples other than $R_\ell$ itself are any localization of $R_\ell$
by a multiplicative subset or any algebraic extension of $R_\ell$.

For $f\in\mc V$, as usual we denote by $\tint f$
the image of $f$ in the quotient space $\mc V/\partial\mc V$.
Recall that, by \eqref{eq:0.4}, we have a well-defined variational derivative
$\frac{\delta}{\delta u}:\,\mc V/\partial\mc V\to\mc V^{\oplus\ell}$,
given by
$$
\frac{\delta\tint f}{\delta u_i}=\sum_{n\in\mb Z_+}(-\partial)^n\frac{\partial f}{\partial u_i^{(n)}},\,i\in I\,.
$$

Note that if the algebra of differential functions $\mc V$ is a domain,
then its field of fractions $\mc K$ is again an algebra of differential functions
in the same variables $u_1,\dots,u_\ell$,
with the maps $\frac{\partial}{\partial u_i^{(n)}}:\,\mc K\to\mc K$ defined in the obvious way.
When $\mc V=\mc K$ we call it a \emph{field of differential functions}.

We denote by $\mc C=\big\{c\in\mc V\,\big|\,\partial c=0\big\}\subset\mc V$ the subalgebra of \emph{constants},
and by
$$
\mc F=\Big\{f\in\mc V\,\Big|\,\frac{\partial f}{\partial u_i^{(n)}}=0
\,\,\text{ for all }\, i\in I,n\in\mb Z_+\Big\}\subset\mc V
$$
the subalgebra of \emph{quasiconstants}. It is easy to see that $\mc C\subset\mc F$.
Given $f\in\mc V$ which is not a quasiconstant, we say that is has \emph{differential order} $N$ if
$\frac{\partial f}{\partial u_i^{(N)}}\neq0$ for some $i\in I$,
and $\frac{\partial f}{\partial u_j^{(n)}}=0$ for every $j\in I$ and $n>N$.
We also set the differential order of a quasiconstant element equal to $-\infty$.
We let $\mc V_N$ be the subalgebra of elements of differential order at most $N$.
This gives an increasing sequence of subalgebras 
$$
\mc C\subset\mc F=\mc V_{-\infty}\subset\mc V_0\subset\mc V_1\subset\dots\subset\mc V\,,
$$
such that $\partial\mc V_N\subset\mc V_{N+1}$.

\subsection{Construction of non-local Hamiltonian structures}
\label{sec:4.2}

Let $\mc V$ be an algebra of differential functions in the variables $u_1,\dots,u_\ell$.
Assume that $\mc V$ is a domain, and let $\mc K$ be the corresponding
field of differential functions of fractions.
Let $H=\big(H_{ij}\big)_{i,j\in I}\in\Mat_{\ell\times\ell}\mc V((\partial^{-1}))$
be an $\ell\times\ell$ matrix pseudodifferential operator over $\mc V$,
namely
$$
H_{ij}=\sum_{n=-\infty}^NH_{ij;n}\partial^n
\,\,\in\mc V((\partial^{-1}))
\,\,,\,\,\,\,
i,j\in I\,.
$$
We associate to this matrix $H$ a non-local $\lambda$-bracket on $\mc V$
given by the following \emph{Master Formula} (cf. \cite{DSK06})
\begin{equation}\label{20110922:eq1}
\{f_\lambda g\}_H
=
\sum_{\substack{i,j\in I \\ m,n\in\mb Z_+}} 
\frac{\partial g}{\partial u_j^{(n)}}
(\lambda+\partial)^n
H_{ji}(\lambda+\partial)
(-\lambda-\partial)^m
\frac{\partial f}{\partial u_i^{(m)}}
\,\in\mc V((\lambda^{-1}))
\,.
\end{equation}
In particular
\begin{equation}\label{20110923:eq1}
\{{u_i}_\lambda{u_j}\}_H
=
H_{ji}(\lambda)
\,\,,\,\,\,i,j\in I\,.
\end{equation}

The following result gives a way to check if a matrix pseudodifferential operator 
$H\in\Mat_{\ell\times\ell}\mc V((\partial^{-1}))$
defines a structure of non-local Poisson vertex algebra on $\mc V$.
The analogous statement in the local case was proved in \cite{BDSK09}.
\begin{theorem}\label{20110923:prop}
Let $\mc V$ be an algebra of differential functions, which is a domain,
ane let $\mc K$ be its field of fractions.
Let $H\in\Mat_{\ell\times\ell}\mc V((\partial^{-1}))$.
Then:
\begin{enumerate}[(a)]
\item
Formula \eqref{20110922:eq1} gives a well-defined non-local $\lambda$-bracket on $\mc V$.
\item
This non-local $\lambda$-bracket is skewsymmetric if and only if $H$
is a skew-adjoint matrix pseudodifferential operator.
\item
If $H=\big(H_{ij}\big)_{i,j\in I}\in\Mat_{\ell\times\ell}\mc V(\partial)$
is a rational matrix pseudodifferential operator with coefficients in $\mc V$,
then the corresponding non-local $\lambda$-bracket 
$\{\cdot\,_\lambda\,\cdot\}_H:\,\mc V\times\mc V\to\mc V((\lambda^{-1}))$
(given by equation \eqref{20110922:eq1}) is admissible.
\item
Let $H=\big(H_{ij}\big)_{i,j\in I}\in\Mat_{\ell\times\ell}\mc V(\partial)$
be a skewadjoint rational matrix pseudodifferential operator with coefficients in $\mc V$.
Then the non-local $\lambda$-bracket $\{\cdot\,_\lambda\,\cdot\}_H$ defined by \eqref{20110922:eq1}
is a Poisson non-local $\lambda$-bracket, i.e. it satisfies the Jacobi identity \eqref{20110922:eq3},
if and only if the Jacobi identity holds on generators ($i,j,k\in I$):
\begin{equation}\label{20110922:eq4}
\{{u_i}_\lambda\{{u_j}_\mu {u_k}\}_H\}_H-\{{u_j}_\mu\{{u_i}_\lambda {u_k}\}_H\}_H
-\{{\{{u_i}_\lambda {u_j}\}_H}_{\lambda+\mu} {u_k}\}_H=0\,,
\end{equation}
where the equality holds in the space $\mc V_{\lambda,\mu}$.
\end{enumerate}
\end{theorem}
\begin{proof}
For the proofs of (a), (b) and (d) one does the same computations 
as in the proof of \cite[Thm.1.15]{BDSK09} for the local case.
So, we only prove part (c).
Let $a,f,g\in\mc V$. By the Master Formula \eqref{20110922:eq1}
and the left Leibniz rule, we have
$$
\begin{array}{l}
\displaystyle{
\{a_\lambda\{f_\mu g\}_H\}_H
=
\sum_{\substack{i,j\in I \\ m,n\in\mb Z_+}} 
\big\{
a_\lambda
\frac{\partial g}{\partial u_j^{(n)}}
(\mu+\partial)^n
H_{ji}(\mu+\partial)
(-\mu-\partial)^m
\frac{\partial f}{\partial u_i^{(m)}}
\big\}_H
} \\
\displaystyle{
=
\sum_{\substack{i,j\in I \\ m,n\in\mb Z_+}} 
\big\{
a_\lambda
\frac{\partial g}{\partial u_j^{(n)}}
\big\}_H
(\mu+\partial)^n
H_{ji}(\mu+\partial)
(-\mu-\partial)^m
\frac{\partial f}{\partial u_i^{(m)}}
} \\
\displaystyle{
+
\sum_{\substack{i,j\in I \\ m,n\in\mb Z_+}} 
\frac{\partial g}{\partial u_j^{(n)}}
(\lambda+\mu+\partial)^n
{\{
a_\lambda
H_{ji}(\mu+\partial)
\}_H}_\to
(-\mu-\partial)^m
\frac{\partial f}{\partial u_i^{(m)}}
} \\
\displaystyle{
+
\sum_{\substack{i,j\in I \\ m,n\in\mb Z_+}} 
\frac{\partial g}{\partial u_j^{(n)}}
(\lambda+\mu+\partial)^n
H_{ji}(\lambda+\mu+\partial)
(-\lambda-\mu-\partial)^m
\big\{
a_\lambda
\frac{\partial f}{\partial u_i^{(m)}}
\big\}_H\,.
}
\end{array}
$$
All sums in the above equations are finite.
Therefore, all three terms in the RHS lie in $\mc V_{\lambda,\mu}$,
thanks to Corollary \ref{20111007:prop} and Lemma \ref{20111006:lem}.
\end{proof}
%
%
\begin{definition}\label{20111007:def}
Let $\mc V$ be an algebra of differential functions, which is a domain.
A \emph{non-local Hamiltonian structure} on $\mc V$
is  a skewadjoint rational matrix pseudodifferential operator with coefficients in $\mc V$,
$H=\big(H_{ij}\big)_{i,j\in I}\in\Mat_{\ell\times\ell}\mc V(\partial)$,
satisfying equation \eqref{20110922:eq4} for every $i,j,k\in I$.
\end{definition}
\begin{remark}\label{20111007:rem2}
It is easy to show that, if $L\in\mc K(\partial)$ is a rational pseudodifferential operator,
then it can be expanded as
\begin{equation}\label{20110922:eq2}
L=\sum_{s=1}^\infty
\sum_{n=0}^N
\sum_{\substack{p_1,\dots,p_s\in\mc V_M \\ (\text{finite sum})}}
p_1\partial^{-1}\circ p_2\partial^{-1}\circ\dots\partial^{-1}\circ p_s\partial^n
\,,
\end{equation}
for some fixed $M,N\in\mb Z_+$.
To see this, write $L=AS^{-1}$, where $A,S\in\mc V[\partial]$
and $S=\sum_{n=0}^Ns_n\partial^n$ has non-zero leading coefficient $s_N$,
and expand $S^{-1}$ using geometric progression:
\begin{equation}\label{20111007:eq1}
S^{-1}=
\partial^{-N}\sum_{i=0}^\infty\Big(-s_N^{-1}s_{N-1}\partial^{-1}-\dots-s_N^{-1}s_0\partial^{-N}\Big)^is_N^{-1}\,.
\end{equation}
On the other hand, it is not hard to see that if $L$ admits an expansion as in \eqref{20110922:eq2},
then $\{a_\lambda L(\mu)\}_H\in K_{\lambda,\mu}$ for every $a\in\mc K$
and every matrix pseudodifferential operator $H$.
As a consequence, if all the entries of a matrix pseudodifferential operator $H$
admit an expansion as in \eqref{20110922:eq2}, 
then the corresponding $\lambda$-bracket $\{\cdot\,_\lambda\,\cdot\}_H$ on $\mc K$
is admissible.
\end{remark}
\begin{remark}\label{20111019:def}
It is claimed in the literature (without a proof) \cite{DN89}
that, in order to show that a skewadjoint operator $H$ defines a (local) Hamiltonian structure,
it suffices to check the Jacobi identity 
for the Lie bracket $\{\cdot\,,\,\cdot\}_H=\{\cdot\,_\lambda\,\cdot\}_H\big|_{\lambda=0}$
in $\mc V/\partial\mc V$ on triples of elements of the form $\tint f u_i$, where $f\in\mc F$ is a quasiconstant.
This is indeed true, provided that the algebra of quasiconstants $\mc F$ is ``big enough'', by the following argument.
By a straightforward computation, using the Master Formula, we get
$$
\begin{array}{l}
\displaystyle{
\{\tint fu_i,\{\tint gu_j,\tint hu_k\}_H\}_H
-\{\tint gu_j,\{\tint fu_i,\tint hu_k\}_H\}_H
} \\
\displaystyle{
-\{\{\tint fu_i,\tint gu_j\}_H,\tint hu_k\}_H
=
\tint h
\Big(
\{{u_i}_\lambda\{{u_j}_\mu{u_k}\}_H\}_H
} \\
\displaystyle{
-\{{u_j}_\mu\{{u_i}_\lambda{u_k}\}_H\}_H
-\{{\{{u_i}_\lambda{u_j}\}_H}_{\lambda+\mu}{u_k}\}_H
\Big)
\big(|_{\lambda=\partial}f\big)
\big(|_{\mu=\partial}g\big)\,.
}
\end{array}
$$
Clearly, this is zero for all $f,g,h\in\mc F$ and all $i,j,k\in I$
if and only if $H$ is a Hamiltonian structure,
provided that the algebra $\mc F$ satisfies the following non-degeneracy conditions:
\begin{enumerate}[(i)]
\item 
if $\tint ha=0$ for some $a\in\mc V$ and all $h\in\mc F$, then $a=0$,
\item
if $P(\partial)f=0$ for some differential operator $P\in\mc V[\partial]$ 
and for all $f\in\mc F$, then $P=0$.
\end{enumerate}
Obviously, $\mc F$ fulfills these conditions if it contains the algebra of polynomials $\mb F[x]$.
Often in the literature this criterion is used also for non-local Hamiltonian structures,
which does not seem to have much sense, since in the non-local case $\mc V/\partial\mc V$
does not have a Lie algebra structure.
\end{remark}

\subsection{Examples}
\label{sec:4.3}

\begin{example}\label{20111010:ex1}
Let $\mc V$ be any algebra of differential functions in $\ell$ differential variables, 
with subalgebra of quasiconstants $\mc F\subset\mc V$.
Any skewadjoint rational matrix pseudodifferential operator 
with quasiconstant coefficients,
$H=\big(H_{ij}(\partial)\big)_{ij\in I}\in\Mat_{\ell\times\ell}\mc F(\partial)$,
is a Hamiltonian structure.
Indeed, by askewadjointness of $H$ the $\lambda$-bracket $\{\cdot\,_\lambda\,\cdot\}_H$
is skewsymmetric, and by the Master Formula, all triple $\lambda$-brackets are zero.
Note that, if $H\in\Mat_{\ell\times\ell}\mc F((\partial^{-1}))$ is skewadjoint, 
even if it is not a rational matrix,
the corresponding $\lambda$-bracket $\{\cdot\,_\lambda\,\cdot\}_H$ 
is still admissible,
hence it defines a non-local Poisson vertex algebra on $\mc V$.

In the special case when $H_{ij}(\lambda)=c_{ij}\lambda^{-1}$,
and $C=(c_{ij})_{i,j=1}^\ell$ is a symmetric matrix with constant coefficients,
we recover the non-local Poisson vertex algebras from Example \ref{20110921:ex2}.
When $C$ if a symmetrized Cartan matrix or extended Cartan matrix of a simple Lie algebra,
we get the Hamiltonian structure for a Toda lattice (see \cite{Fr98}).
\end{example}
\begin{example}\label{20110922:ex1}
The following three operators form a compatible family of non-local Hamiltonian structures
(i.e. any their linear combination is a non-local Hamiltonian structure)
on the algebra $R_1=\mb F[u,u',u'',\dots]$
of differential polynomials in one variable:
\begin{enumerate}[(i)]
\item
$K_{1}=\partial$ (GFZ Hamiltonina structure),
\item
$K_{-1}=\partial^{-1}$ (Toda non-local Hamiltonian structure),
\item
$H=u'\partial^{-1}\circ u'$ (Sokolov non-local Hamitonian structure),
\end{enumerate}
First, any linear combination over $\mc C$ of $K_1$ and $K_{-1}$
is a non-local Hamiltonian structure, as discussed in Example \ref{20111010:ex1}.
Next, it is easy to show (cf. \cite[Example 3.14]{BDSK09})
that $H^{-1}$ is a symplectic structure on the field of fractions $\mc K_1=\text{Frac} R_1$,
known as the Sokolov symplectic structure, \cite{Sok84}.
Hence, by Theorem \ref{20111012:thm} below,
we deduce that $H$ is a non-local Hamiltonian structure.
To conclude that $K_1,K_{-1},H$ form a compatible family,
it suffices to check that
$$
\{u_\lambda H(\mu)\}_{K_{\pm1}}
-\{u_\mu H(\lambda)\}_{K_{\pm1}}
=\{H(\lambda)_{\lambda+\mu}u\}_{K_{\pm1}}\,,
$$
where $H(\lambda)=u'(\partial+\lambda)^{-1}u'\in\mc V((\lambda^{-1}))$.
This is straightforward.
\end{example}
\begin{example}\label{20110922:ex2}
Dorfman non-local Hamiltonian structure 
on the algebra of differential polynomials
$R_1=\mb F[u,u',u'',\dots]$ is:
$$
H=\partial^{-1}\circ u'\partial^{-1}\circ u'\partial^{-1}\,.
$$
One easily shows (cf. \cite[Example 3.14]{BDSK09})
that $H^{-1}$ is a symplectic structure on the field of fractions $\mc K_1=\text{Frac} R_1$, 
known as Dorfman symplectic structre, \cite{Dor93},
hence $H$ is indeed a non-local Hamiltonian structure.
Furthermore, one can show, by a lengthy calculation,
that Sokolov's and Dorfman's non-local Hamiltonian structures
are compatible.
\end{example}
\begin{example}[cf. \cite{Dor93}]\label{20110922:ex3}
Another triple of compatible non-local Hamiltonian structures
on $R_1=\mb F[u,u',u'',\dots]$ is:
\begin{enumerate}[(i)]
\item
$K_{1}=\partial$ (GFZ Hamiltonian structure),
\item
$K_{-1}=\partial^{-1}$ (Toda non-local Hamiltonian structure),
\item
$H=\partial^{-1}\circ u'+u'\partial^{-1}$ (potential Virasoro-Magri non-local Hamiltonian structure).
\end{enumerate}
\end{example}
\begin{example}[cf. \cite{Mag80}]\label{20110922:ex4}
There is yet another triple of compatible non-local Hamiltonian structures
on $R_1=\mb F[u,u',u'',\dots]$:
\begin{enumerate}[(i)]
\item
$K_1=\partial$ (GFZ Hamiltonian structure),
\item
$K_3=\partial^3$,
\item
$H=\partial\circ u\partial^{-1}\circ u\partial$ 
(modified Virasoro-Magri non-local Hamiltonian structure).
\end{enumerate}
\end{example}
\begin{example}[cf. \cite{Mag78,Mag80}]\label{20110922:ex5}
The following is a triple of compatible non-local Hamiltonian structures
on $R_2=\mb F[u,v,u',v',\dots]$:
\begin{enumerate}[(i)]
\item
$K_1=\partial\id$ (GFZ Hamiltonian structure),
\item
$K=\left(\begin{array}{cc} 0 & -1 \\ 1 & 0 \end{array}\right)$,
\item
$H=\left(\begin{array}{cc} 
v\partial^{-1}\circ v & -v\partial^{-1}\circ u \\
-u\partial^{-1}\circ v & u\partial^{-1}\circ u
\end{array}\right)$ (NLS non-local Hamiltonian structure).
\end{enumerate}
\end{example}


\section{Constructing families of compatible non-local \\ Hamiltonian structures}
\label{sec:6}

As in the previous sections, let $\mc V$ be an algebra of differential functions
in the variables $u_1,\dots,u_\ell$,
we assume that $\mc V$ is a domain, and we let $\mc K$ be its field of fractions.
As in the local case, two non-local Poisson vertex algebra $\lambda$-brackets on $\mc V$
(respectively two non-local Hamiltonian structures) 
are said to be \emph{compatible} if any their linear combination
is again a non-local Poisson vertex algebra structure
(resp. a non-local Hamiltonian structure).
Such a pair is called a \emph{bi-Hamiltonian structure}.
More generally, a collection of non-local 
Hamiltonian structures $\{H^\alpha\}_{\alpha\in\mc A}$ on $\mc V$,
is called \emph{compatible} if any their (finite) linear combination
is a non-local Hamiltonian structure on $\mc V$.

Recalling the Jacobi identity \eqref{20110922:eq4}, we introduce the following notation.
Given rational $\ell\times\ell$-matrix pseudodifferential operators 
$K,H\in\Mat_{\ell\times\ell}\mc V(\partial)$,
we let $J(H,K)=J^1(H,K)-J^2(H,K)-J^3(H,K)$,
where $J^\alpha(H,K)=\big(J^\alpha_{ijk}(H,K)(\lambda,\mu)\big)_{i,j,k\in I}$, for $\alpha=1,2,3$,
are the arrays with the following entries in $\mc V_{\lambda,\mu}$:
\begin{equation}\label{20111118:eq1}
\begin{array}{rcl}
J^1(H,K)_{ijk}(\lambda,\mu)&=&\{{u_i}_\lambda\{{u_j}_\mu{u_k}\}_H\}_K\,,\\
J^2(H,K)_{ijk}(\lambda,\mu)&=&\{{u_j}_\mu\{{u_i}_\lambda{u_k}\}_H\}_K\,,\\
J^3(H,K)_{ijk}(\lambda,\mu)&=&\{{\{{u_i}_\lambda{u_j}\}_H}_{\lambda+\mu}{u_k}\}_K\,.
\end{array}
\end{equation}
Consider a collection $\{H^\alpha\}_{\alpha\in\mc A}$
of skewadjoint rational non-local matrix pseudodifferential operators.
By definition, $H^\alpha$ is a Hamiltonian structure
if and only if 
$J(H^\alpha,H^\alpha)=0$.
It is easy to see that the $H^\alpha$'s form a compatible family of Hamiltonian structures
if and only if
\begin{equation}\label{20111116:eq3}
J(H^\alpha,H^\beta)+J(H^\beta,H^\alpha)=0
\,\,,\,\,\,\,\forall \alpha,\beta\in\mc A\,.
\end{equation}

\begin{theorem}\label{20111021:thm}
Let $H,\,K\in\Mat_{\ell\times\ell}\mc V(\partial)$
be compatible non-local Hamiltonian structures on 
the the algebra of differential functions $\mc V$, which is a domain.
Assume that $K$ is an invertible element of the algebra $\Mat_{\ell\times\ell}\mc V(\partial)$.
Then the following sequence of rational matrix pseudodifferential operators
with coefficients in $\mc V$:
$$
H^{[0]}=K
\,\,,\,\,\,\,
H^{[n]} :=
\big(H\circ K^{-1}\big)^{n-1}\circ H
\,\in\Mat{}_{\ell\times\ell}\mc V(\partial)
\,\,,\,\,\,\,n\geq1\,,
$$
form a compatible family of non-local Hamiltonian structures on $\mc V$.
\end{theorem}
\begin{remark}\label{20111114:rem}
It is stated in \cite{FF81} that $H^{[n]},\,n\geq0$, are non-local Hamiltonian structures,
but the prove there is given only under the additional
assupmtion that $H$ is invertible as well.
In this case the proof becomes much easier since $H^{[n]}$ is invertible, 
therefore one needs to prove that $(H^{[n]})^{-1}$ is a symplectic structure.
\end{remark}
Following the idea in \cite{TT11}, we will reduce the proof of Theorem \ref{20111021:thm}
to the following special case of it:
\begin{lemma}\label{20111116:lem1}
Let $\tilde H,\,K\in\Mat_{\ell\times\ell}\mc V(\partial)$
be compatible non-local Hamiltonian structures on $\mc V$,
and assume that $K$ is an invertible element 
of the algebra $\Mat_{\ell\times\ell}\mc V(\partial)$.
Then the rational matrix pseudodifferential operator with coefficients in $\mc V$
$$
\tilde H(\partial)\circ K^{-1}(\partial)\circ\tilde H(\partial)\,\in\Mat{}_{\ell\times\ell}\mc V(\partial)\,,
$$
is a non-local Hamiltonian structure on $\mc V$.
\end{lemma}
\begin{proof}
To simplify notation, in this proof we denote $\tilde H$ by $H$,
and we let $R=H\circ K^{-1}$
so that $R^*=K^{-1}\circ H$.
Let $H^{[2]}=H\circ K^{-1}\circ H
\Big(=R\circ H=H\circ R^*\Big)$,
and let $\{\cdot\,_\lambda\,\cdot\}_2=\{\cdot\,_\lambda\,\cdot\}_{H^{[2]}}$
be the non-local $\lambda$-bracket on $\mc V$ associated
to $H^{[2]}\in\Mat_{\ell\times\ell}\mc V(\partial)$
via \eqref{20110922:eq1}.
We need to prove the Jacobi identity, i.e. using the notation in \eqref{20111118:eq1},
that $J(H^{[2]},H^{[2]})=0$.

We need to compute all three terms $J^\alpha=J^\alpha(H^{[2]},H^{[2]})_{ijk}(\lambda,\mu)$, for $\alpha=1,2,3$,
of the Jacobi identity.
First, if $f\in\mc V$ and $i\in I$, we have, in $\mc V((\lambda^{-1}))$,
\begin{eqnarray}
\label{20111103:eq1a}
\{{u_i}_\lambda f\}_2
&=&
\displaystyle{
\sum_{s\in I}{\{{u_s}_{\lambda+\partial}f\}_H}_\to R^*_{si}(\lambda)\,,
} \\
\label{20111103:eq1b}
\{{u_j}_\mu f\}_2
&=&
\displaystyle{
\sum_{t\in I}{\{{u_t}_{\mu+\partial}f\}_H}_\to R^*_{tj}(\mu)\,,
} \\
\label{20111103:eq2}
\{f_{\lambda+\mu}{u_k}\}_2
&=&
\displaystyle{
\sum_{r\in I}R_{kr}(\lambda+\mu+\partial)\, \{f_{\lambda+\mu}u_r\}_H\,.
}
\end{eqnarray}
Both the above equations follow immediately from the Master formula \eqref{20110922:eq1}
and the definition of $H^{[2]}$.
The following identities are proved in a similar way, using that $K\circ K^{-1}=\id$,
\begin{eqnarray}
\label{20111121:eq1a}
\{{u_i}_\lambda f\}_H
&=&
\displaystyle{
\sum_{s\in I}{\{{u_s}_{\lambda+\partial}f\}_K}_\to R^*_{si}(\lambda)\,,
} \\
\label{20111121:eq1b}
\{{u_j}_\mu f\}_H
&=&
\displaystyle{
\sum_{t\in I}{\{{u_t}_{\mu+\partial}f\}_K}_\to R^*_{tj}(\mu)\,,
} \\
\label{20111121:eq2}
\{f_{\lambda+\mu}{u_k}\}_H
&=&
\displaystyle{
\sum_{r\in I}R_{kr}(\lambda+\mu+\partial)\, \{f_{\lambda+\mu}u_r\}_K\,.
}
\end{eqnarray}
Next, it is not hard to chek, using the left and right Leibniz rules
and Lemma \ref{20111012:lem}, that,
given an admissible non-local $\lambda$-bracket $\{\cdot\,_\lambda\,\cdot\}$ on $\mc V$,
the following identities hold in $\mc V_{\lambda,\mu}$, for every $i,j,k\in I$:
\begin{eqnarray}
\label{20111103:eq3a}
&&\{{u_i}_\lambda H^{[2]}_{kj}(\mu)\} 
=
\displaystyle{
\sum_{t\in I}
\{{u_i}_\lambda \{{u_t}_y{u_k}\}_H\} \Big(\Big|_{y=\mu+\partial}R^*_{tj}(\mu)\Big)
} \\
&&\displaystyle{
- \sum_{r,t\in I}
R_{kr}(\lambda\!+\!\mu\!+\!\partial)\,
\{{u_i}_\lambda \{{u_t}_y{u_r}\}_K\} \Big(\Big|_{y=\mu+\partial}\!\!\!R^*_{tj}(\mu)\Big)
}\nonumber \\
&&\displaystyle{
+ \sum_{r\in I}
R_{kr}(\lambda+\mu+\partial)\,
\{{u_i}_\lambda \{{u_j}_\mu{u_r}\}_H\}
\,,
}\nonumber \\
\label{20111103:eq3b}
&&\{{u_j}_\mu H^{[2]}_{ki}(\lambda)\} 
=
\displaystyle{
\sum_{s\in I}
\{{u_j}_\mu \{{u_s}_x{u_k}\}_H\} \Big(\Big|_{x=\lambda+\partial}R^*_{si}(\lambda)\Big)
} \\
&&\displaystyle{
- \sum_{r,s\in I}
R_{kr}(\lambda\!+\!\mu\!+\!\partial)\,
\{{u_j}_\mu \{{u_s}_x{u_r}\}_K\} \Big(\Big|_{x=\lambda+\partial}\!\!\!R^*_{si}(\lambda)\Big)
}\nonumber \\
&&\displaystyle{
+ \sum_{r\in I}
R_{kr}(\lambda+\mu+\partial)\,
\{{u_j}_\mu \{{u_i}_\lambda{u_r}\}_H\}
\,,
}\nonumber \\
\label{20111103:eq3c}
&&\{H^{[2]}_{ji}(\lambda)_{\lambda+\mu}{u_k}\} 
=
\displaystyle{
\sum_{s\in I}
\{{\{{u_s}_x{u_j}\}_H}_{\lambda+\mu+\partial}{u_k}\}_\to 
\Big(\Big|_{x=\lambda+\partial}\!\!\!R^*_{si}(\lambda)\Big)
} \\
&&\displaystyle{
- \sum_{s,t\in I}
\{{\{{u_s}_x{u_t}\}_K}_{\lambda+\mu+\partial}{u_k}\}_\to 
\Big(\Big|_{x=\lambda+\partial}\!\!\!\!R^*_{si}(\lambda)\!\Big)
\Big(\Big|_{y=\mu+\partial}\!\!\!\!R^*_{tj}(\mu)\!\Big)
}\nonumber \\
&&\displaystyle{
+ \sum_{t\in I}
\{{\{{u_i}_\lambda{u_t}\}_H}_{\lambda+\mu+\partial}{u_k}\}_\to 
R^*_{tj}(\mu)
\,.
}\nonumber 
\end{eqnarray}
Here and further we use the following notation:
given an element 
$$
P(\lambda,\mu)=\sum_{m,n,p=-\infty}^N p_{m,n,p}\lambda^m\mu^n(\lambda+\mu)^p
\in\mc V_{\lambda,\mu}\,,
$$ 
and $f,g\in\mc V$, we let
\begin{equation}\label{20111018:eq5}
\begin{array}{l}
\displaystyle{
P(x,y)\Big(\Big|_{x=\lambda+\partial}f\Big)
\Big(\Big|_{y=\mu+\partial}g\Big)
} \\
\displaystyle{
=
\sum_{m,n,p=-\infty}^N
p_{m,n,p}(\lambda+\mu+\partial)^p
\big((\lambda+\partial)^mf\big)\big((\mu+\partial)^ng\big)
\,\,\in\mc V_{\lambda,\mu}
\,.
}
\end{array}
\end{equation}
In equation \eqref{20111103:eq3c} we used the assumption that $H$ and $K$ are skewadjoint.
Combining equations \eqref{20111103:eq1a} and \eqref{20111103:eq3a},
equations \eqref{20111103:eq1b} and \eqref{20111103:eq3b},
and equations \eqref{20111103:eq2} and \eqref{20111103:eq3c},
we get, respectively,
\begin{eqnarray}
\label{20111121:eq3a}
&&\displaystyle{
J^1 = \{{u_i}_\lambda \{{u_j}_\mu{u_k}\}_2\}_2
} \\
&&\displaystyle{
=\sum_{s,t\in I}
\{{u_s}_x \{{u_t}_y{u_k}\}_H\}_H 
\Big(\Big|_{x=\lambda+\partial}R^*_{si}(\lambda)\Big)
\Big(\Big|_{y=\mu+\partial}R^*_{tj}(\mu)\Big)
}\nonumber \\
&&\displaystyle{
- \sum_{r,s,t\in I}
R_{kr}(\lambda\!+\!\mu\!+\!\partial)\,
\{{u_s}_x \{{u_t}_y{u_r}\}_K\}_H 
\Big(\Big|_{x=\lambda+\partial}R^*_{si}(\lambda)\Big)
\Big(\Big|_{y=\mu+\partial}R^*_{tj}(\mu)\Big)
}\nonumber \\
&&\displaystyle{
+ \sum_{r,s\in I}
R_{kr}(\lambda+\mu+\partial)\,
\{{u_s}_x \{{u_j}_\mu{u_r}\}_H\}_H
\Big(\Big|_{x=\lambda+\partial}R^*_{si}(\lambda)\Big)
\,,
}\nonumber \\
\label{20111121:eq3b}
&&\displaystyle{
J^2 = \{{u_j}_\mu \{{u_i}_\lambda{u_k}\}_2\}_2
} \\
&&\displaystyle{
=\sum_{s,t\in I}
\{{u_j}_y \{{u_s}_x{u_k}\}_H\}_H 
\Big(\Big|_{x=\lambda+\partial}R^*_{si}(\lambda)\Big)
\Big(\Big|_{y=\mu+\partial}R^*_{tj}(\mu)\Big)
}\nonumber \\
&&\displaystyle{
- \sum_{r,s,t\in I}
R_{kr}(\lambda\!+\!\mu\!+\!\partial)\,
\{{u_t}_y \{{u_s}_x{u_r}\}_K\}_H 
\Big(\Big|_{x=\lambda+\partial}R^*_{si}(\lambda)\Big)
\Big(\Big|_{y=\mu+\partial}R^*_{tj}(\mu)\Big)
}\nonumber \\
&&\displaystyle{
+ \sum_{r,t\in I}
R_{kr}(\lambda+\mu+\partial)\,
\{{u_t}_y \{{u_i}_\lambda{u_r}\}_H\}_H
\Big(\Big|_{y=\mu+\partial}R^*_{tj}(\mu)\Big)
\,,
}\nonumber \\
\label{20111121:eq3c}
&&\displaystyle{
J^3 = \{{\{{u_i}_\lambda{u_j}\}_2}_{\lambda+\mu}{u_k}\}_2
} \\
&&\displaystyle{
=\sum_{r,s\in I}
R_{kr}(\lambda+\mu+\partial)\,
{\{{\{{u_s}_x{u_j}\}_H}_{\lambda+\mu+\partial}{u_r}\}_H}_\to 
\Big(\Big|_{x=\lambda+\partial}\!\!\!R^*_{si}(\lambda)\Big)
}\nonumber \\
&&\displaystyle{
- \sum_{r,s,t\in I}
R_{kr}(\lambda+\mu+\partial)\,
{\{{\{{u_s}_x{u_t}\}_K}_{\lambda+\mu+\partial}{u_r}\}_H}_\to 
\Big(\Big|_{x=\lambda+\partial}R^*_{si}(\lambda)\!\Big)
R^*_{tj}(\mu)
}\nonumber \\
&&\displaystyle{
+ \sum_{r,t\in I}
R_{kr}(\lambda+\mu+\partial)\,
{\{{\{{u_i}_\lambda{u_t}\}_H}_{\lambda+\mu+\partial}{u_r}\}_H}_\to 
R^*_{tj}(\mu)
\,.
}\nonumber 
\end{eqnarray}
We need to prove that $J^1-J^2-J^3=0$.
The first term of the RHS of \eqref{20111121:eq3a} combined with the first term of the RHS of \eqref{20111121:eq3b}
gives, by the Jacobi identity for $H$ and by equation \eqref{20111121:eq2},
\begin{equation}\label{20111121:eq4a}
\begin{array}{l}
\displaystyle{
\sum_{s,t\in I}
\!\!
\Big(\{{u_s}_x \{{u_t}_y{u_k}\}_H\}_H - \{{u_j}_y \{{u_s}_x{u_k}\}_H\}_H\Big)\!\!
\Big(\Big|_{x=\lambda+\partial}R^*_{si}(\lambda)\Big)\!\!
\Big(\Big|_{y=\mu+\partial}R^*_{tj}(\mu)\Big)
} \\
\displaystyle{
=
\sum_{s,t\in I}
\{{\{{u_s}_x{u_t}\}_H}_{x+y}{u_k}\}_H\}
\Big(\Big|_{x=\lambda+\partial}R^*_{si}(\lambda)\Big)
\Big(\Big|_{y=\mu+\partial}R^*_{tj}(\mu)\Big)
} \\
\displaystyle{
=
\sum_{r,s,t\in I}
\!\!
R_{kr}(\lambda+\mu+\partial)
\{{\{{u_s}_x{u_t}\}_H}_{x+y}{u_r}\}_K
\Big(\Big|_{x=\lambda+\partial}R^*_{si}(\lambda)\Big)
\!\!\Big(\Big|_{y=\mu+\partial}R^*_{tj}(\mu)\Big).
}
\end{array}
\end{equation}
Similarly,
the third term of the RHS of \eqref{20111121:eq3a} combined with the first term of the RHS of \eqref{20111121:eq3c}
gives, by the Jacobi identity for $H$ and by equation \eqref{20111121:eq1b},
\begin{equation}\label{20111121:eq4b}
\begin{array}{l}
\displaystyle{
\sum_{r,s\in I}
R_{kr}(\lambda+\mu+\partial)\,
\Big(\{{u_s}_x \{{u_j}_\mu{u_r}\}_H\}_H
-{\{{\{{u_s}_x{u_j}\}_H}_{\lambda+\mu+\partial}{u_r}\}_H}_\to\Big)
} \\
\displaystyle{
\Big(\Big|_{x=\lambda+\partial}R^*_{si}(\lambda)\Big)
=
\sum_{r,s\in I}
R_{kr}(\lambda+\mu+\partial)\,
\{{u_j}_\mu \{{u_s}_x{u_r}\}_H\}_H
\Big(\Big|_{x=\lambda+\partial}R^*_{si}(\lambda)\Big)
} \\
\displaystyle{
=
\sum_{r,s,t\in I}
R_{kr}(\lambda+\mu+\partial)\,
\{{u_t}_y \{{u_s}_x{u_r}\}_H\}_K
\Big(\Big|_{x=\lambda+\partial}R^*_{si}(\lambda)\Big)
\Big(\Big|_{y=\mu+\partial}R^*_{tj}(\mu)\Big)\,.
}
\end{array}
\end{equation}
In the same way,
the third term of the RHS of \eqref{20111121:eq3b} combined with the third term of the RHS of \eqref{20111121:eq3c}
gives, by the Jacobi identity for $H$ and by equation \eqref{20111121:eq1a},
\begin{equation}\label{20111121:eq4c}
\begin{array}{c}
\displaystyle{
-\sum_{r,t\in I}
R_{kr}(\lambda+\mu+\partial)\,
\Big(
\{{u_t}_y \{{u_i}_\lambda{u_r}\}_H\}_H
+{\{{\{{u_i}_\lambda{u_t}\}_H}_{\lambda+y}{u_r}\}_H} 
\Big)
} \\
\displaystyle{
\Big(\Big|_{y=\mu+\partial}R^*_{tj}(\mu)\Big)
=
-\sum_{r,t\in I}
R_{kr}(\lambda+\mu+\partial)\,
\{{u_i}_\lambda \{{u_t}_y{u_r}\}_H\}_H
} \\
\displaystyle{
\Big(\Big|_{y=\mu+\partial}R^*_{tj}(\mu)\Big)
=
-\sum_{r,s,t\in I}
R_{kr}(\lambda+\mu+\partial)\,
\{{u_s}_\lambda \{{u_t}_y{u_r}\}_H\}_K
} \\
\displaystyle{
\Big(\Big|_{x=\lambda+\partial}R^*_{si}(\lambda)\Big)
\Big(\Big|_{y=\mu+\partial}R^*_{tj}(\mu)\Big)\,.
}
\end{array}
\end{equation}
Finally, combining the second term in the RHS of \eqref{20111121:eq3a}, \eqref{20111121:eq3b} and \eqref{20111121:eq3c},
together with the RHS of equations \eqref{20111121:eq4a}, \eqref{20111121:eq4b} and \eqref{20111121:eq4c},
we get
$$
\begin{array}{c}
\displaystyle{
J^1-J^2-J^3=
\sum_{r,s,t\in I}
R_{kr}(\lambda+\mu+\partial)\,
\Big(
-\{{u_s}_x \{{u_t}_y{u_r}\}_K\}_H 
} \\
\displaystyle{
+\{{u_t}_y \{{u_s}_x{u_r}\}_K\}_H 
+{\{{\{{u_s}_x{u_t}\}_K}_{x+y}{u_r}\}_H}_\to 
+\{{\{{u_s}_x{u_t}\}_H}_{x+y}{u_r}\}_K
} \\
\displaystyle{
+\{{u_t}_y \{{u_s}_x{u_r}\}_H\}_K
-\{{u_s}_\lambda \{{u_t}_y{u_r}\}_H\}_K
\Big)
\Big(\Big|_{x=\lambda+\partial}R^*_{si}(\lambda)\Big)
\Big(\Big|_{y=\mu+\partial}R^*_{tj}(\mu)\Big)\,,
}
\end{array}
$$
which is zero since, by assumption, $H$ and $K$ are compatible.
\end{proof}
\begin{remark}\label{20111122:rem}
The proof of Lemma \ref{20111116:lem1} does not use the assumption that $K$
is a Hamiltonian structure.
\end{remark}
\begin{proof}[Proof of Theorem \ref{20111021:thm}]
We prove, by induction on $n\geq1$, that the rational matrix pseudodifferential operators
$$
H^{[0]}=K,\,H^{[1]}=H,\dots,H^{[n]}
\,\in\Mat{}_{\ell\times\ell}\mc V(\partial)\,,
$$
form a compatible family of non-local Hamiltonian structures on $\mc V$.
For $n=1$, this holds by assumption.
Assuming by induction that the statement holds for $n\geq1$,
we will prove that it holds for $n+1$.
Namely, 
thanks to the observations at the beginning of the section,
we need to prove that
\begin{enumerate}[(i)]
\item
$J(H^{[n+1]},H^{[n+1]})=0$,
\item
$J(H^{[m]},H^{[n+1]})+J(H^{[n+1]},H^{[m]})=0$
for every $m=0,\dots,n$.
\end{enumerate}
By the inductive assumption, $\tilde H=\sum_{i=0}^nx_iH^{[i]}$
is a Hamiltonian structure for every $x_0,\dots,x_n\in\mb F$.
Hence, by Lemma \ref{20111116:lem1},
we get the following Hamiltonian structure for every point $(x_0,\dots,x_n)\in\mb F^{n+1}$:
$$
\tilde H\circ K^{-1}\circ\tilde H
=
\sum_{i,j=0}^nx_ix_j H^{[i]}\circ K^{-1}\circ H^{[j]}
=
\sum_{p=0}^{2n}Q_p(x_0,\dots,x_n) H^{[p]}\,,
$$
where, for $p=0,\dots,2n$,
\begin{equation}\label{20111116:eq4}
Q_p(x_0,\dots,x_n)=
\sum_{\substack{i,j=0 \\ (i+j=p)}}^n x_ix_j\,.
\end{equation}
We thus get
$$
\begin{array}{l}
\displaystyle{
0=J(\tilde H K^{-1}\tilde H,\tilde H K^{-1}\tilde H)
=
\sum_{p=0}^{2n}Q_p^2(x_0,\dots,x_n) J(H^{[p]},H^{[p]})
} \\
\displaystyle{
+\sum_{\substack{p,q=0 \\ (p<q)}}^{2n}Q_p(x_0,\dots,x_n) Q_q(x_0,\dots,x_n) 
\big(J(H^{[p]},H^{[q]})+J(H^{[q]},H^{[p]})\big)\,,
}
\end{array}
$$
for every $(x_0,\dots,x_n)\in\mb F^{n+1}$.
Note that, by the inductive assumption, $J(H^{[p]},H^{[p]})=0$
for every $0\leq p\leq n$
and $J(H^{[p]},H^{[q]})+J(H^{[q]},H^{[p]})=0$ for every $0\leq p<q\leq n$.
Hence the above equation gives
\begin{equation}\label{20111116:eq5}
\begin{array}{l}
\displaystyle{
\sum_{p=n+1}^{2n} Q_p^2(x_0,\dots,x_n) J(H^{[p]},H^{[p]})
} \\
\displaystyle{
+\sum_{p=0}^n\sum_{q=n+1}^{2n} Q_p(x_0,\dots,x_n) Q_q(x_0,\dots,x_n) 
\big(J(H^{[p]},H^{[q]})+J(H^{[q]},H^{[p]})\big)
} \\
\displaystyle{
+\!\sum_{\substack{p,q=n+1 \\ (p<q)}}^{2n} Q_p(x_0,\dots,x_n) Q_q(x_0,\dots,x_n) 
\big(J(H^{[p]},H^{[q]})+J(H^{[q]},H^{[p]})\big)=0
}
\end{array}
\end{equation}
for every $(x_0,\dots,x_n)\in\mb F^{n+1}$.
Next, we introduce a grading in the algebra of polynomials in $x_0,\dots,x_n$,
letting $\deg(x_i)=i$.
Then $Q_p(x_0,\dots,x_n)$ is homogeneous of degree $p$.
By looking at the terms of degree $d=2n+2$ in equation \eqref{20111116:eq5},
we get
\begin{equation}\label{20111116:eq6}
\begin{array}{l}
\displaystyle{
Q_{n+1}^2(x_0,\dots,x_n) J(H^{[n+1]},H^{[n+1]})
+\sum_{p=2}^n
Q_p(x_0,\dots,x_n) 
} \\
\displaystyle{
Q_{2n+2-p}(x_0,\dots,x_n) 
\big(J(H^{[p]},H^{[2n+2-p]})+J(H^{[2n+2-p]},H^{[p]})\big)=0\,,
}
\end{array}
\end{equation}
while, 
by looking at the terms of degree $d=m+n+1$ with $m\in\{0,\dots,n\}$
in equation \eqref{20111116:eq5}, we get
\begin{equation}\label{20111116:eq7}
\begin{array}{l}
\displaystyle{
\sum_{p=0}^m
Q_p(x_0,\dots,x_n) Q_{m+n+1-p}(x_0,\dots,x_n) 
} \\
\displaystyle{
\big(J(H^{[p]},H^{[m+n+1-p]})+J(H^{[m+n+1-p]},H^{[p]})\big)=0\,,
}
\end{array}
\end{equation}
for every $(x_0,\dots,x_n)\in\mb F^{n+1}$.
To conclude the proof, we only need to show that equations \eqref{20111116:eq6} and \eqref{20111116:eq7}
imply respectively relations (i) and (ii) above.
This is a consequence of the following lemma.
\begin{lemma}\label{20111116:lem2}
\begin{enumerate}[(a)]
\item
For every $n\geq1$,
\begin{equation}\label{20111117:eq1}
Q_{n+1}^2(x_0,\dots,x_n)\not\in
\Span{}_{\mb F}\Big\{Q_p(x_0,\dots,x_n) Q_{2n+2-p}(x_0,\dots,x_n)\Big\}_{2\leq p\leq n}\,.
\end{equation}
\item
For every $n\geq1$ and $m\in\{0,\dots,n\}$,
\begin{equation}\label{20111117:eq2}
Q_m Q_{n+1}
\not\in\Span{}_{\mb F}\Big\{Q_p Q_{m+n+1-p}\Big\}_{0\leq p\leq m-1}\,.
\end{equation}
\end{enumerate}
\end{lemma}
\begin{proof}
Note that,
$$
\begin{array}{l}
Q_p(0,\dots,0,x_k,\dots,x_n) \\
\displaystyle{
=
\sum_{\substack{i,j=k \\ (i+j=p)}}^nx_ix_j=
\left\{\begin{array}{ll}
0 & \text{ if } p<2k\,,\\
x_k^2 & \text{ if } p=2k\,,\\
\displaystyle{
2x_kx_{p-k}+\dots } & \text{ if } p>2k\,.
\end{array}\right.
}
\end{array}
$$

We prove part (a) separately in the cases when $n$ is even and odd.
If $n=2k-1$ is odd, letting $x_0=\dots=x_{k-1}=0$ we have
$Q_{n+1}=x_k^2\neq0$, and $Q_p=0$ for all $p=2,\dots,n=2k-1$.
This implies \eqref{20111117:eq1} for odd $n$.
If $n=2$, we have $Q_2=2x_0x_2+x_1^2,\,Q_3=2x_1x_2,\,Q_4=x_2^2$,
hence $Q_3^2\not\in\mb FQ_2Q_4$.
If $n=2k$ with $k\geq2$, letting $x_0=\dots=x_{k-1}=0$ we have
$Q_p=0$ for all $p=2,\dots,n-1$,
$Q_n=x_k^2$,
$Q_{n+1}=2x_kx_{k+1}$, 
$Q_{n+2}=2x_kx_{k+2}+x_{k+1}^2$.
Since $Q_{n+1}^2=4x_k^2x_{k+1}^2$ is not a multiple of $Q_nQ_{n+2}=2x_k^3x_{k+2}+x_k^2x_{k+1}^2$,
\eqref{20111117:eq1} holds for even $n$.

Similarly, we prove part (b) separately in the cases when $m$ is even and odd.
If $m=2k$ is even, letting $x_0=\dots=x_{k-1}=0$ we have
$Q_{m}Q_{n+1}=x_k^2Q_{n+1}\neq0$, and $Q_p=0$ for all $p=2,\dots,m-1$.
Hence \eqref{20111117:eq2} holds for even $m$.
For $m=1\leq n$, we have $Q_0=x_0^2,\,Q_1=2x_0x_1,\,Q_{n+1}=2x_1x_n+\dots$.
Therefore, $Q_0Q_{n+2}$ is divisible by $x_0^2$, while $Q_1Q_{n+1}=2x_0x_1(2x_1x_n+\dots)$ is not.
%
Finally, if $m=2k+1$ is odd, with $k\geq1$, letting $x_0=\dots=x_{k-1}=0$ we have
$Q_p=0$ for all $p=2,\dots,m-2$,
$Q_{m-1}=x_k^2$,
$Q_{m}=2x_kx_{k+1}$, 
$Q_{n+1}=2x_kx_{n+1-k}+2x_{k+1}x_{n-k}+\dots$.
Hence, $Q_{m-1}Q_{n+2}=x_k^2Q_{n+2}$ is divisible by $x_k^2$,
while $Q_mQ_{n+1}=4x_k^2x_{k+1}x_{n+1-k}+4x_kx_{k+1}^2x_{n-k}+\dots$ is not,
proving \eqref{20111117:eq2} for odd $m$.
\end{proof}
\end{proof}
\begin{example}\label{20120224:vic}
Let $K=\partial^3,\,H=\partial^2\circ\frac1u\partial\circ\frac1u\partial^2$. These are compatible Hamiltonian
structures (see \cite{DSKW10}).
Hence, by Theorem \ref{20111021:thm},
$$
H^{[n]}=(H\circ K^{-1})^{n-1}\circ H
=\partial^2\circ (\frac1u\circ \partial)^{2n}\circ\partial
\,\,,\,\,\,\,n\in\mb Z_+\,,
$$
are compatible Hamiltonian structures.
This was proved in \cite{DSKW10} by direct verification,
and deduced from Theorem \ref{20111021:thm} in \cite{TT11}.
\end{example}


\section{Symplectic structures and Dirac structures in terms of non-local Hamiltonian structures}
\label{sec:5}

\subsection{Simplectic structure as inverse of a non-local Hamiltonian structure}
\label{sec:5.1}

As in the previous sections,
let $\mc V$ be an algebra of differential functions in the variables $u_1,\dots,u_\ell$,
which is a domain, and let $\mc K$ be its field of fractions.

Recall that (see e.g. \cite{BDSK09}) a \emph{symplectic structure} on $\mc V$ 
is an $\ell\times\ell$ matrix differential operator
$S=\big(S_{ij}(\partial)\big)_{i,j\in I}\in\Mat_{\ell\times\ell}\mc V[\partial]$
which is skewadjoint and satisfies the
following \emph{symplectic identity}:
\begin{equation}\label{20111012:eq1}
\sum_{n\in\mb Z_+}\Big(
\frac{\partial S_{ki}(\mu)}{\partial u_j^{(n)}} \lambda^n
-\frac{\partial S_{kj}(\lambda)}{\partial u_i^{(n)}} \mu^n
+(-\lambda-\mu-\partial)^n \frac{\partial S_{ij}(\lambda)}{\partial u_k^{(n)}}
\Big)=0\,.
\end{equation}

We can write the symplectic identity \eqref{20111012:eq1} in terms 
of the \emph{Beltrami} $\lambda$-\emph{bracket}
$\langle\cdot\,_\lambda\,\cdot\rangle:\,\mc V\times\mc V\to\mc V[\lambda]$,
introduced in \cite{BDSK09}.
It is defined as the symmetric $\lambda$-bracket such that $\langle {u_i}_\lambda u_j\rangle=\delta_{ij}$,
and extended by the Master Formula \eqref{20110922:eq1}:
$$
\langle f_\lambda g\rangle
=
\sum_{\substack{ i\in I \\ m,n\in\mb Z_+}}(-1)^m
\frac{\partial g}{\partial u_i^{(n)}}(\lambda+\partial)^{m+n}\frac{\partial f}{u_i^{(m)}}\,.
$$
Then, the symplectic identity \eqref{20111012:eq1} becomes
\begin{equation}\label{20111017:eq1}
\langle {u_j}_\lambda\{{u_i}_\mu{u_k}\}_S\rangle
-\langle {u_i}_\mu\{{u_j}_\lambda{u_k}\}_S\rangle
+\langle {\{{u_j}_\lambda {u_i}\}_S}_{\lambda+\mu}{u_k}\rangle=0
\,,
\end{equation}
where, recalling \eqref{20110922:eq1}, we let $\{{u_j}_\lambda{u_i}\}_S=S_{ij}(\lambda)$.

Note that, 
if $S\in\Mat_{\ell\times\ell}\mc V(\partial)$ 
is a rational matrix pseudodifferential operator with coefficients in $\mc V$,
then, by Corollary \ref{20111007:prop}, all three terms in the LHS of equation \eqref{20111017:eq1}
lie in $\mc V_{\lambda,\mu}$. Hence, equation \eqref{20111017:eq1} still makes sense
(as an equation in $\mc V_{\lambda,\mu}$).
\begin{definition}\label{20111017:def}
A \emph{non-local symplectic structure} on $\mc V$
is a skewadjoint rational matrix pseudodifferential operator 
$S=\big(S_{ij}(\partial)\big)_{i,j\in I}\in\Mat_{\ell\times\ell}\mc V(\partial)$
with coefficients in $\mc V$,
satisftying equation \eqref{20111017:eq1} in $\mc V_{\lambda,\mu}$ for all $i,j,k\in I$.
\end{definition}
\begin{theorem}\label{20111012:thm}
Let $S\in\Mat_{\ell\times\ell}\mc V(\partial)$ be a skewadjoint
rational matrix pseudodifferential operator with coefficients in 
the algebra  of differential functions $\mc V$.
Assume that $S$ is an invertible element of the algebra $\Mat_{\ell\times\ell}\mc V(\partial)$.
Then, $S$ is a non-local symplectic structure on $\mc V$ if and only if 
$S^{-1}$ is a non-local Hamiltonian structure on $\mc V$.
\end{theorem}
\begin{proof}
Clearly, $S$ is skewadjoint if and only if $S^{-1}$ is skewadjoint.
Hence, recalling the Definition \ref{20111007:def} of non-local Hamiltonian structure,
we only need to show that equation \eqref{20111012:eq1} in $\mc V_{\lambda,\mu}$ 
is equivalent to the Jacobi identity \eqref{20110922:eq4}, again in $\mc V_{\lambda,\mu}$, 
for $H=S^{-1}$.
By equation \eqref{20111012:eq2a}, Remark \ref{20111104:rem},
and the Master Formula \eqref{20110922:eq1}, we have,
letting $S_{ij}(\partial)=\sum_{p=-\infty}^Ns_{ij;p}\partial^p$,
\begin{equation}\label{20111018:eq1}
\begin{array}{l}
\displaystyle{
\{{u_i}_\lambda\{{u_j}_\mu{u_k}\}_H\}_H
=
\big\{{u_i}_\lambda (S^{-1})_{kj}(\mu)\big\}_{S^{-1}}
=
}\\
\displaystyle{
-\sum_{r,t=1}^\ell\sum_{p=-\infty}^N
(S^{-1})_{kr}(\lambda+\mu+\partial) 
\{{u_i}_\lambda s_{rt;p}\}_{S^{-1}} (\mu+\partial)^p (S^{-1})_{tj}(\mu)
}\\
\displaystyle{
=
-\sum_{r,s,t\in I,\,n\in\mb Z_+}
(S^{-1})_{kr}(\lambda+\mu+\partial) 
\Big(
(\lambda+\partial)^n
(S^{-1})_{si}(\lambda)
\Big)
}\\
\displaystyle{
\,\,\,\,\,\,\,\,\,\,\,\,\,\,\,\,\,\,\,\,\,\,\,\,\,\,\,\,\,\,\,\,\,\,\,\,\,\,\,\,\,\,\,\,\,\,\,\,\,\,\,\,\,\,
\,\,\,\,\,\,\,\,\,\,\,\,\,\,\,\,\,\,\,\,\,\,\,\,\,\,\,\,\,\,\,\,\,\,\,\,\,\,\,\,\,\,\,\,\,\,\,\,\,\,\,\,\,\,
\Big(
\frac{\partial S_{rt}(\mu+\partial)}{\partial u_s^{(n)}}
(S^{-1})_{tj}(\mu)
\Big)
\,.
}
\end{array}
\end{equation}
Exchanging $i$ with $j$ and $\lambda$ with $\mu$, we get
\begin{equation}\label{20111018:eq2}
\begin{array}{c}
\displaystyle{
\{{u_j}_\mu\{{u_i}_\lambda{u_k}\}_H\}_H
=
-\sum_{r,s,t\in I,\,n\in\mb Z_+}
(S^{-1})_{kr}(\lambda+\mu+\partial) 
}\\
\displaystyle{
\Big(
(\mu+\partial)^n
(S^{-1})_{tj}(\mu)
\Big)
\Big(
\frac{\partial S_{rs}(\lambda+\partial)}{\partial u_t^{(n)}}
(S^{-1})_{si}(\lambda)
\Big)
\,.
}
\end{array}
\end{equation}
Similarly, by equation \eqref{20111012:eq2b} and Remark \ref{20111104:rem}, we have,
using the assumption that $S$ is skewadjoint,
\begin{equation}\label{20111018:eq3}
\begin{array}{l}
\displaystyle{
\{{\{{u_i}_\lambda{u_j}\}_H}_{\lambda+\mu}{u_k}\}_H
=
\big\{(S^{-1})_{ji}(\lambda)_{\lambda+\mu}{u_k}\big\}_{S^{-1}}
}\\
\displaystyle{
=
\sum_{s,t=1}^\ell\sum_{p=-\infty}^N
{\{{s_{ts;p}}_{\lambda+\mu+\partial}{u_k}\}_{S^{-1}}}_\to
(S^{-1})_{tj}(\mu) (\lambda+\partial)^p (S^{-1})_{si}(\lambda)
}\\
\displaystyle{
=
\sum_{r,s,t\in I,\,m\in\mb Z_+}
(S^{-1})_{kr}(\lambda+\mu+\partial)
(-\lambda-\mu-\partial)^m
}\\
\displaystyle{
\,\,\,\,\,\,\,\,\,\,\,\,\,\,\,\,\,\,\,\,\,\,\,\,\,\,\,\,\,\,\,\,\,\,\,\,\,\,\,\,\,\,\,\,\,\,\,\,\,\,\,\,\,\,
\,\,\,\,\,\,\,\,\,\,\,\,\,
\Big(
(S^{-1})_{tj}(\mu)
\frac{\partial S_{ts}(\lambda+\partial)}{\partial u_r^{(m)}} (S^{-1})_{si}(\lambda)
\Big)
\,.
}
\end{array}
\end{equation}
Combining equations \eqref{20111018:eq1}, \eqref{20111018:eq2} and \eqref{20111018:eq3},
we get that the LHS of the Jacobi identity \eqref{20110922:eq4} is
\begin{equation}\label{20111018:eq4}
\begin{array}{l}
\displaystyle{
\{{u_i}_\lambda\{{u_j}_\mu {u_k}\}_H\}_H-\{{u_j}_\mu\{{u_i}_\lambda {u_k}\}_H\}_H
-\{{\{{u_i}_\lambda {u_j}\}_H}_{\lambda+\mu} {u_k}\}_H
} \\
\displaystyle{
=\sum_{r,s,t\in I,\,n\in\mb Z_+}
(S^{-1})_{kr}(\lambda+\mu+\partial) 
\Bigg(
-\frac{\partial S_{rt}(y)}{\partial u_s^{(n)}}x^n
+\frac{\partial S_{rs}(x)}{\partial u_t^{(n)}}y^n
} \\
\displaystyle{
-(-x-y-\partial)^n
\frac{\partial S_{ts}(x)}{\partial u_r^{(n)}}
\Bigg)
\Big(\Big|_{x=\lambda+\partial}(S^{-1})_{si}(\lambda)\Big)
\Big(\Big|_{y=\mu+\partial}(S^{-1})_{tj}(\mu)\Big)
\,,
}
\end{array}
\end{equation}
where we used the notation introduced in \eqref{20111018:eq5}.
Clearly, the RHS of \eqref{20111018:eq4} is zero, provided that the symplectic identity \eqref{20111012:eq1} holds.
For the opposite implication,
we have, by \eqref{20111018:eq4},
$$
\begin{array}{l}
\displaystyle{
\sum_{i,j,k\in I}
S_{\gamma k}(x+y+\partial)
\Bigg(\{{u_i}_x\{{u_j}_y {u_k}\}_H\}_H-\{{u_j}_y\{{u_i}_x {u_k}\}_H\}_H
} \\
\displaystyle{
-\{{\{{u_i}_x {u_j}\}_H}_{x+y} {u_k}\}_H
\Bigg)
\Big(\Big|_{x=\lambda+\partial} S_{i\alpha}(\lambda)\Big)
\Big(\Big|_{y=\mu+\partial}S_{j\beta}(\mu)\Big)
} \\
\displaystyle{
=
\sum_{n\in\mb Z_+}
\Bigg(
-\frac{\partial S_{\gamma\beta}(\mu)}{\partial u_\alpha^{(n)}}\lambda^n
+\frac{\partial S_{\gamma\alpha}(\lambda)}{\partial u_\beta^{(n)}}\mu^n
-(-\lambda-\mu-\partial)^n
\frac{\partial S_{\beta\alpha}(\lambda)}{\partial u_\gamma^{(n)}}
\Bigg)
\,.
}
\end{array}
$$
Hence, equation \eqref{20110922:eq4} implies equation \eqref{20111012:eq1}.
\end{proof}

\subsection{Dirac structure in terms of non-local Hamiltonian structure}
\label{sec:5.2}

Let $\mc V$ be an algebra of differential functions, which is a domain,
and let $\mc K$ be its field of fractions.
Given a set $J$ and an element $X\in\mc V^J$,
we define the \emph{Frechet derivative} of $X$ as the differential operator
$D_X(\partial):\,\mc V^\ell\to\mc V^J$ given by
\begin{equation}\label{20111020:eq1}
\big(D_X(\partial)P\big)_i
=\sum_{n\in\mb Z_+}\sum_{j\in I}\frac{\partial X_i}{\partial u_j^{(n)}} \partial^n P_j\,.
\end{equation}
Its adjoint operator is the map $D_X^*(\partial):\,\mc V^{\oplus J}\to\mc V^{\oplus\ell}$ given by:
\begin{equation}\label{20111020:eq2}
\big(D_X^*(\partial)Y\big)_i
=\sum_{n\in\mb Z_+}\sum_{j\in I}(-\partial)^n\Big(\frac{\partial X_j}{\partial u_i^{(n)}} Y_j\Big)\,.
\end{equation}
Here and further, for a possibly infinite set $J$, $\mc V^{\oplus J}$
denotes the space of column vectors in $\mc V^J$ with only finitely many non-zero entries.
(Though in this paper we do not consider any example with infinite $\ell$,
we still distinguish $\mc V^\ell$ and $\mc V^{\oplus\ell}$ as a book-keeping device).

The following identity can be checked directly and it will be useful later:
\begin{equation}\label{20120405:eq1}
\tint X\cdot D_Y(\partial)P+\tint Y\cdot D_X(\partial)P
=\tint P\cdot \frac{\delta}{\delta u}(X\cdot Y)\,,
\end{equation}
for all $X\in\mc V^J,\,Y\in\mc V^{\oplus J},\,P\in\mc V^\ell$.

We have the usual pairing $\mc V^{\oplus\ell}\times\mc V^\ell\to\mc V/\partial\mc V$
given by $(F|P)=\tint F\cdot P$.
This pairing is non-degenerate (see e.g. \cite[Prop.1.3(a)]{BDSK09}).
We extend it to a non-degenerate  symmetric bilinear form
\begin{equation}\label{20111020:eq3}
\langle\cdot\,|\,\cdot\rangle\,:\,\,
\big(\mc V^{\oplus\ell}\oplus\mc V^\ell\big)\times\big(\mc V^{\oplus\ell}\oplus\mc V^\ell\big)\to\mc V/\partial\mc V\,,
\end{equation}
given by
$\langle F\oplus P|G\oplus Q\rangle=\tint (F\cdot Q+G\cdot P)$.


The \emph{Courant-Dorfman product} is the following product 
on the space $\mc V^{\oplus\ell}\oplus\mc V^{\ell}$:
\begin{equation}\label{20111020:eq5}
(F\oplus P)\circ(G\oplus Q)
=
\big(
D_G(\partial)P+D_P^*(\partial)G-D_F(\partial)Q+D_F^*(\partial)Q
\big)
\oplus
[P,Q]\,,
\end{equation}
where, for $P,Q\in\mc V^\ell$, we let
\begin{equation}\label{20120126:eq1}
[P,Q]
=
D_Q(\partial)P-D_P(\partial)Q\,.
\end{equation}
The above formula takes a simpler form when $F$ and $G$ are variational derivatives
(see \cite[Rem.4.6]{BDSK09}):
\begin{equation}\label{20120127:eq1}
\Big(\frac{\delta \tint f}{\delta u}\oplus P\Big)\circ\Big(\frac{\delta \tint g}{\delta u}\oplus Q\Big)
=
\frac{\delta}{\delta u}\Big(\int P\cdot\frac{\delta g}{\delta u}\Big)\oplus[P,Q]
\,.
\end{equation}

\begin{remark}\label{20111020:rem1}
All the above notions have a natural interpretation from the point of view of variational calculus.
Indeed, the space $\mc V^\ell$ is naturally identified with the Lie algebra of evolutionary vector fields $\mf g^\partial$,
and the space $\mc V^{\oplus\ell}$ is naturally identified with the space 
of variational 1-forms $\Omega^1$.
Then the contraction of variational 1-forms by evolutionary vector fields gives the inner product \eqref{20111020:eq3};
the Courant-Dorfman product corresponds to the derived bracket
$[\cdot\,,\,\cdot]_d$, where $[\cdot\,,\,\cdot]$ is the Lie superalgebra bracket 
on the space of endomorphisms of the space of all de Rham forms over $\mc V$,
and $d=\ad(\delta)$, where $\delta$ is the de Rham differential, \cite[Prop.4.2]{BDSK09}.
\end{remark}

\begin{definition}[\cite{Dor93,BDSK09}]\label{20111020:def}
A \emph{Dirac structure} is a subspace $\mc L\subset\mc V^{\oplus\ell}\oplus\mc V^{\ell}$,
which is maximal isotropic with respect to the inner product \eqref{20111020:eq3},
and which is closed under the Courant-Dorfman product \eqref{20111020:eq5}.
\end{definition}
Given two $\ell\times\ell$ 
matrix differential operators $A,B\in\Mat_{\ell\times\ell}\mc V[\partial]$
consider the following subspace of $\mc V^{\oplus\ell}\oplus\mc V^{\ell}$:
\begin{equation}\label{20120109:eq1}
\mc L_{A,B}=\big\{B(\partial)X\oplus A(\partial)X\,\big|\,X\in\mc V^{\oplus\ell}\big\} \,.
\end{equation}
\begin{proposition}\label{20120103:propa}
The subspace $\mc L_{A,B}\subset\mc V^{\oplus\ell}\oplus\mc V^{\ell}$ is isotropic
with respect to the inner product \eqref{20111020:eq3}
if and only if 
\begin{equation}\label{20120107:eq1}
A^*\circ B+B^*\circ A=0\,.
\end{equation}
If, moreover, $\det B\neq0$ (i.e. $B$ is invertible 
in the algebra $\Mat_{\ell\times\ell} \mc K((\partial^{-1}))$),
then \eqref{20120107:eq1} holds if and only if
$A\circ B^{-1}\in\Mat_{\ell\times\ell} \mc K((\partial^{-1}))$ 
is skewadjiont,
while if $\det A\neq0$, 
then \eqref{20120107:eq1} holds if and only if
$B\circ A^{-1}\in\Mat_{\ell\times\ell} \mc K((\partial^{-1}))$
is skewadjiont.
\end{proposition}
\begin{proof}
For $X,Y\in\mc V^{\oplus\ell}$ we have
$$
\langle B(\partial)X\oplus A(\partial)X\,|\,B(\partial)Y\oplus A(\partial)Y\rangle
=\tint Y\cdot\big(A^*(\partial)B(\partial)+B^*(\partial)A(\partial)\big)X\,.
$$
Hence, due to non-degeeracy of the pairing $(F|P)=\tint F\cdot P$,
the space $\mc L_{A,B}$ is isotropic if and only if \eqref{20120107:eq1} holds.
The remaining statements are straightforward.
\end{proof}
\begin{example}\label{20120124:ex0}
Letting $A\in\Mat_{\ell\times\ell}\mc V$ and $B=\id_\ell\,\partial$,
condition \eqref{20120107:eq1} holds
if and only if $A$ is a symmetric matrix with entries in $\mc C\subset\mc V$ 
(the subring of constant functions).
In this case $AB^{-1}$ is a skewadjont matrix pseudodifferential operator
and 
$\mc L_{A,B}=\big\{AX\oplus\partial X\,\big|\,X\in\mc V^{\oplus\ell}\big\}$ 
is an isotropic subspace of $\mc V^{\oplus\ell}\oplus\mc V^\ell$.
It is not hard to show directly that
$\mc L_{A,B}$ is maximal isotropic if and only if 
the matrix $A$ is non-degenerate.
When $\mc V=\mc K$ is a differential field, this is a corollary of the following general result:
\end{example}
\begin{proposition}[\cite{CDSK12b}]\label{prop:max-isotrop}
Let $\mc K$ be a differential field, 
and let $H=AB^{-1}$ be a minimal fractional decomposition
of the skewadjoint rational matrix pseudodifferential operator $H\in\Mat_{\ell\times\ell}\mc K(\partial)$.
Then the subspace $\mc L_{A,B}\subset\mc K^{\oplus\ell}\oplus\mc K^\ell$
is maximal isotropic with respect to the inner product \eqref{20111020:eq3}.
\end{proposition}
\begin{proposition}\label{20120103:propb}
Suppose that $A,B\in\Mat_{\ell\times\ell}\mc V[\partial]$
satisfy equation \eqref{20120107:eq1}.
Then the following conditions are equivalent:
\begin{enumerate}[(i)]
\item
$\langle X\circ Y,Z\rangle=0$ for all $X,Y,Z\in\mc L_{A,B}$.
\item 
for every $F,G\in\mc V^{\ell}$ one has:
\begin{equation}\label{20120103:eq1}
\begin{array}{l}
A^*(\partial)D_{B(\partial)G}(\partial)A(\partial)F
+A^*(\partial)D_{A(\partial)F}^*(\partial)B(\partial)G \\
-A^*(\partial)D_{B(\partial)F}(\partial)A(\partial)G
+A^*(\partial)D_{B(\partial)F}^*(\partial)A(\partial)G \\
+B^*(\partial)D_{A(\partial)G}(\partial)A(\partial)F
-B^*(\partial)D_{A(\partial)F}(\partial)A(\partial)G\,=\,0\,.
\end{array}
\end{equation}
\item
for every $i,j,k\in I$, one has in the space $\mc V[\lambda,\mu]$:
\end{enumerate}
\begin{equation}\label{20120103:eq2}
\begin{array}{l}
\displaystyle{
\sum_{\substack{s,t\in I \\ n\in\mb Z_+}}
\!\Bigg(\!
A^*_{ks}(\lambda+\mu+\partial) \bigg(
\frac{\partial B_{sj}(\mu)}{\partial u_t^{(n)}} (\lambda+\partial)^nA_{ti}(\lambda)
-
\frac{\partial B_{si}(\lambda)}{\partial u_t^{(n)}} (\mu+\partial)^nA_{tj}(\mu)
\!\bigg)
} \\
\displaystyle{
+B^*_{ks}(\lambda+\mu+\partial) \bigg(
\frac{\partial A_{sj}(\mu)}{\partial u_t^{(n)}} 
(\lambda+\partial)^nA_{ti}(\lambda)
-
\frac{\partial A_{si}(\lambda)}{\partial u_t^{(n)}} 
(\mu+\partial)^nA_{tj}(\mu)
\bigg)
} \\
\displaystyle{
+A^*_{ks}(\lambda+\mu+\partial) 
(-\lambda-\mu-\partial)^n
\bigg(
\!
\frac{\partial A_{ti}(\lambda)}{\partial u_s^{(n)}} 
B_{tj}(\mu)
+
\frac{\partial B_{ti}(\lambda)}{\partial u_s^{(n)}} 
A_{tj}(\mu)
\bigg)
\!\Bigg)
=\,0\,.
}
\end{array}
\end{equation}
\end{proposition}
\begin{proof}
Letting 
$X=B(\partial)F\oplus A(\partial)F,\,Y=B(\partial)G\oplus A(\partial)G,\,Z=B(\partial)E\oplus A(\partial)E$,
condition (i) reads
$$
\begin{array}{l}
\tint
\big(A(\partial)E\big)
\cdot \big(D_{B(\partial)G}(\partial)A(\partial)F
+D_{A(\partial)F}^*(\partial)B(\partial)G
-D_{B(\partial)F}(\partial) A(\partial)G 
\\
+D_{B(\partial)F}^*(\partial)A(\partial)G\big)
\!+\!
\big(B(\partial)E\big)
\!\!\cdot\!\!
\big(D_{A(\partial)G}(\partial)A(\partial)F
\!-\!
D_{A(\partial)F}(\partial)A(\partial)G\big)
\!=\!0
.
\end{array}
$$
Since the above equation holds for every $E\in\mc V^{\oplus\ell}$
it reduces, integrating by parts, to equation \eqref{20120103:eq1}.
For this we use the non-degeneracy of the pairing \eqref{20111020:eq3}.
This proves that conditions (i) and (ii) are equivalent.

We next prove that conditions (ii) and (iii) are equivalent,
provided that \eqref{20120107:eq1} holds.
For $\alpha=1,\dots,6$, let \eqref{20120103:eq1}$_\alpha$ 
be the $k$-entry of the $\alpha$-th term of the LHS of \eqref{20120103:eq1}:
for example
\eqref{20120103:eq1}$_1=\big(A^*(\partial)D_{B(\partial)G}(\partial)A(\partial)F\big)_k$.
We have, by the definition of the Frechet derivative and some algebraic manipulations
(similar to those used in the proof of \cite[Prop.1.16]{BDSK09}),
$$
\begin{array}{rcl}
\eqref{20120103:eq1}_1&=&
\displaystyle{
\sum_{i,j,s,t\in I}\sum_{n\in\mb Z_+}
\bigg(
A^*_{ks} (\partial)\Big(\frac{\partial B_{sj}(\partial)}{\partial u_t^{(n)}} G_j\Big) \partial^n A_{ti}(\partial)F_i
}\\&&
\displaystyle{
+
A^*_{ks}(\partial)B_{sj}(\partial) \frac{\partial G_j}{\partial u_t^{(n)}} \partial^n A_{ti}(\partial)F_i
\bigg)
\,,} \\
\eqref{20120103:eq1}_2&=&
\displaystyle{
\sum_{i,j,s,t\in I}\sum_{n\in\mb Z_+}
\bigg(
A^*_{ks}(\partial) (-\partial)^n 
\Big(\frac{\partial A_{ti}(\partial)}{\partial u_s^{(n)}} F_i\Big)  B_{tj}(\partial)G_j
}\\&&
\displaystyle{
+
A^*_{ks}(\partial) (-\partial)^n
\frac{\partial F_i}{\partial u_s^{(n)}} A^*_{it}(\partial) B_{tj}(\partial)G_j
\bigg)
\,,} \\
\eqref{20120103:eq1}_3&=&
\displaystyle{
-\sum_{i,j,s,t\in I}\sum_{n\in\mb Z_+}
\bigg(
A^*_{ks} (\partial)\Big(\frac{\partial B_{si}(\partial)}{\partial u_t^{(n)}} F_i\Big) \partial^n A_{tj}(\partial)G_j
}\\&&
\displaystyle{
+
A^*_{ks}(\partial)B_{si}(\partial) \frac{\partial F_i}{\partial u_t^{(n)}} \partial^n A_{tj}(\partial)G_j
\bigg)
\,,} 
\end{array}
$$
$$
\begin{array}{rcl}
\eqref{20120103:eq1}_4&=&
\displaystyle{
\sum_{i,j,s,t\in I}\sum_{n\in\mb Z_+}
\bigg(
A^*_{ks}(\partial) (-\partial)^n 
\Big(\frac{\partial B_{ti}(\partial)}{\partial u_s^{(n)}} F_i\Big)  A_{tj}(\partial)G_j
}\\&&
\displaystyle{
+
A^*_{ks}(\partial) (-\partial)^n
\frac{\partial F_i}{\partial u_s^{(n)}} B^*_{it}(\partial) A_{tj}(\partial)G_j
\bigg)
\,,} \\
\eqref{20120103:eq1}_5&=&
\displaystyle{
\sum_{i,j,s,t\in I}\sum_{n\in\mb Z_+}
\bigg(
B^*_{ks} (\partial)\Big(\frac{\partial A_{sj}(\partial)}{\partial u_t^{(n)}} G_j\Big) \partial^n A_{ti}(\partial)F_i
}\\&&
\displaystyle{
+
B^*_{ks}(\partial)A_{sj}(\partial) \frac{\partial G_j}{\partial u_t^{(n)}} \partial^n A_{ti}(\partial)F_i
\bigg)
\,,}
\end{array}
$$
$$
\begin{array}{rcl}
\eqref{20120103:eq1}_6&=&
\displaystyle{
-\sum_{i,j,s,t\in I}\sum_{n\in\mb Z_+}
\bigg(
B^*_{ks} (\partial)\Big(\frac{\partial A_{si}(\partial)}{\partial u_t^{(n)}} F_i\Big) \partial^n A_{tj}(\partial)G_j
}\\&&
\displaystyle{
+
B^*_{ks}(\partial)A_{si}(\partial) \frac{\partial F_i}{\partial u_t^{(n)}} \partial^n A_{tj}(\partial)G_j
\bigg)
\,.}
\end{array}
$$
Combining the second terms in the RHS 
of \eqref{20120103:eq1}$_1$ and \eqref{20120103:eq1}$_5$
we get zero, thanks to equation \eqref{20120107:eq1}.
Similarly, we get zero if we combine the second terms in the RHS 
of \eqref{20120103:eq1}$_2$ and \eqref{20120103:eq1}$_4$,
and if we combine the second terms in the RHS 
of \eqref{20120103:eq1}$_3$ and \eqref{20120103:eq1}$_6$.
Equation \eqref{20120103:eq2} thus follows from equation \eqref{20120103:eq1}
once we replace $\partial$ acting on $F_i$ by $\lambda$,
and $\partial$ acting on $G_j$ by $\mu$.
\end{proof}
\begin{remark}
It follows by Definition \ref{20111020:def} and Proposition \ref{20120103:propb}
that a Dirac structure is a maximal isotropic 
subspace $\mc L$ of $\mc V^{\oplus\ell}\oplus\mc V^\ell$
satisfying one of the equivalent conditions (i)--(iii) above.
\end{remark}
\begin{proposition}\label{20120103:propc}
Suppose that $A,B\in\Mat_{\ell\times\ell}\mc V[\partial]$
satisfy equation \eqref{20120107:eq1}
and assume that $B$ has non-zero Dieudonn\`e determinant.
Suppose, moreover, that
the (skewadjoint) rational matrix pseudodifferential operator 
$H=A\circ B^{-1}$ has coefficients in $\mc V$, i.e. $H\in\Mat_{\ell\times\ell}\mc V(\partial)$.
Consider the corresponding non-local $\lambda$-bracket $\{\cdot\,_\lambda\,\cdot\}_H$
given by the Master Formula \eqref{20110922:eq1}.
Then the Jacobi identity \eqref{20110922:eq4} on $\{\cdot\,_\lambda\,\cdot\}_H$
is equivalent to equation \eqref{20120103:eq2} on the entries of matrices 
$A$ and $B$.
\end{proposition}
\begin{proof}
Letting $A_{st}(\partial)=\sum_{m=0}^Ma_{st;m}\partial^m$
and $B_{st}(\partial)=\sum_{n=0}^M b_{ij;n}\partial^n$.
By formula \eqref{20110923:eq1} and the left Leibniz rule \eqref{20110921:eq3}
we have,
$$
\begin{array}{c}
\displaystyle{
\{{u_i}_\lambda\{{u_j}_\mu{u_k}\}_H\}_H
=
\sum_{r\in I}
\{{u_i}_\lambda A_{kr}(\mu+\partial)B^{-1}_{rj}(\mu)\}_H
} \\
\displaystyle{
=
\sum_{r\in I}\sum_{m=0}^M
\{{u_i}_\lambda a_{kr;m}\}_H (\mu+\partial)^m B^{-1}_{rj}(\mu)
} \\
\displaystyle{
+\sum_{r\in I}
A_{kr}(\lambda+\mu+\partial) \{{u_i}_\lambda B^{-1}_{rj}(\mu)\}_H 
\,.} 
\end{array}
$$
By equation \eqref{20111012:eq2c} we have
$$
\begin{array}{l}
\displaystyle{
\{{u_i}_\lambda B^{-1}_{rj}(\mu)\}_H
} \\
\displaystyle{
=
-\sum_{s,t\in I} \sum_{m=0}^M
(B^{-1})_{rs}(\lambda+\mu+\partial) \{{u_i}_\lambda b_{st;m}\}_H (\mu+\partial)^m (B^{-1})_{tj}(\mu)\,.
}
\end{array}
$$
Combining the above two equations we then get, using the Master Formula \eqref{20110922:eq1},
$$
\begin{array}{l}
\displaystyle{
\{{u_i}_\lambda\{{u_j}_\mu{u_k}\}_H\}_H
} \\
\displaystyle{
=
\sum_{r,s,t\in I}\sum_{m=0}^M\sum_{n\in\mb Z_+}
\frac{\partial a_{kr;m}}{\partial u_s^{(n)}} 
\Big(
(\mu+\partial)^m B^{-1}_{rj}(\mu)
\Big)
\Big(
(\lambda+\partial)^n
A_{st}(\lambda+\partial)
B^{-1}_{ti}(\lambda)
\Big)
} \\
\displaystyle{
-\sum_{p,q,r,s,t\in I} \sum_{m=0}^M \sum_{n\in\mb Z_+}
A_{kr}(\lambda+\mu+\partial) 
(B^{-1})_{rs}(\lambda+\mu+\partial) 
} \\
\displaystyle{
\,\,\,\,\,\,\,\,\,\,\,\,\,\,\,\,\,\,\,\,\,\,\,\,\,\,\,\,\,\,\,\,\,
\frac{\partial b_{st;m}}{\partial u_p^{(n)}}
\Big((\mu+\partial)^m (B^{-1})_{tj}(\mu)\Big)
\Big((\lambda+\partial)^n
A_{pq}(\lambda+\partial) B^{-1}_{qi}(\lambda)\Big)
} \\
\displaystyle{
=
\sum_{r,s,t\in I}\sum_{n\in\mb Z_+}
\Big(
\frac{\partial A_{kr}(\mu+\partial)}{\partial u_s^{(n)}} 
B^{-1}_{rj}(\mu)
\Big)
\Big(
(\lambda+\partial)^n
A_{st}(\lambda+\partial)
B^{-1}_{ti}(\lambda)
\Big)
} \\
\displaystyle{
-\sum_{p,q,r,s,t\in I} \sum_{n\in\mb Z_+}
A_{kr}(\lambda+\mu+\partial) 
(B^{-1})_{rs}(\lambda+\mu+\partial) 
} \\
\displaystyle{
\,\,\,\,\,\,\,\,\,\,\,\,\,\,\,\,\,\,\,\,\,\,\,\,\,\,\,\,\,\,\,\,\,
\Big(
\frac{\partial B_{st}(\mu+\partial)}{\partial u_p^{(n)}}
(B^{-1})_{tj}(\mu)\Big)
\Big((\lambda+\partial)^n
A_{pq}(\lambda+\partial) B^{-1}_{qi}(\lambda)\Big)
\,.}
\end{array}
$$
Next, we apply $B^*_{k'k}(\lambda+\mu+\partial)$ to both sides of the above equation (on the left),
replace $\lambda$ by $\lambda+\partial$ acting on $B_{ii'}(\lambda)$,
replace $\mu$ by $\mu+\partial$ acting on $B_{jj'}(\mu)$,
and sum over $i,j,k\in I$.
As a result we get, using the assumption \eqref{20120107:eq1}
(see notation \eqref{20111018:eq5}),
\begin{equation}\label{20120109:eq2}
\begin{array}{c}
\displaystyle{
\sum_{i,j,k\in I}B^*_{k'k}(\lambda+\mu+\partial)
\{{u_i}_x\{{u_j}_y{u_k}\}_H\}_H
\Big(\Big|_{x=\lambda+\partial}B_{ii'}(\lambda)\Big)
\Big(\Big|_{y=\mu+\partial}B_{jj'}(\mu)\Big)
} \\
\displaystyle{
=
\sum_{k,i\in I}\sum_{n\in\mb Z_+}
\bigg(
B^*_{k'k}(\lambda+\mu+\partial)
\frac{\partial A_{kj'}(\mu)}{\partial u_i^{(n)}} 
(\lambda+\partial)^n
A_{ii'}(\lambda)
} \\
\displaystyle{
+
A^*_{k'k}(\lambda+\mu+\partial)
\frac{\partial B_{kj'}(\mu)}{\partial u_i^{(n)}}
(\lambda+\partial)^n
A_{ii'}(\lambda) 
\bigg)
\,.}
\end{array}
\end{equation}
Exchanging $i'$ with $j'$ and $\lambda$ with $\mu$ in \eqref{20120109:eq2}, we get
\begin{equation}\label{20120109:eq3}
\begin{array}{c}
\displaystyle{
\sum_{i,j,k\in I}B^*_{k'k}(\lambda+\mu+\partial)
\{{u_j}_y\{{u_i}_x{u_k}\}_H\}_H
\Big(\Big|_{x=\lambda+\partial}B_{ii'}(\lambda)\Big)
\Big(\Big|_{y=\mu+\partial}B_{jj'}(\mu)\Big)
} \\
\displaystyle{
=
\!\sum_{i,j,k\in I}B^*_{k'k}(\lambda+\mu+\partial)
\{{u_i}_x\{{u_j}_y{u_k}\}_H\}_H
\Big(\Big|_{x=\mu+\partial}B_{ij'}(\mu)\Big)
\Big(\Big|_{y=\lambda+\partial}B_{ji'}(\lambda)\Big)
} \\
\displaystyle{
=
\sum_{k,j\in I}\sum_{n\in\mb Z_+}
\bigg(
B^*_{k'k}(\lambda+\mu+\partial)
\frac{\partial A_{ki'}(\lambda)}{\partial u_j^{(n)}} 
(\mu+\partial)^n
A_{jj'}(\mu)
} \\
\displaystyle{
+
A^*_{k'k}(\lambda+\mu+\partial)
\frac{\partial B_{ki'}(\lambda)}{\partial u_j^{(n)}}
(\mu+\partial)^n
A_{jj'}(\mu) 
\bigg)
\,.}
\end{array}
\end{equation}
We are left to study the third term in the Jacobi identity.
By the right Leibniz rule \eqref{20110921:eq3} we get,
$$
\begin{array}{l}
\displaystyle{
\{{\{{u_i}_\lambda{u_j}\}_H}_{\lambda+\mu}{u_k}\}_H
=
\sum_{r\in I}
\{A_{jr}(\lambda+\partial)B^{-1}_{ri}(\lambda)_{\lambda+\mu}{u_k}\}_H
} \\
\displaystyle{
=
\sum_{r\in I}\sum_{m=0}^M
{\{{a_{jr;m}}_{\lambda+\mu+\partial}{u_k}\}_H}_\to
(\lambda+\partial)^m B^{-1}_{ri}(\lambda)
} \\
\displaystyle{
+
\sum_{r\in I}\sum_{m=0}^M
{\{B^{-1}_{ri}(\lambda)_{\lambda+\mu+\partial}{u_k}\}_H}_\to
(-\mu-\partial)^m a_{jr;m}
\,.} 
\end{array}
$$
By equation \eqref{20111012:eq2d} we have
$$
\begin{array}{l}
\displaystyle{
{\{B^{-1}_{ri}(\lambda)_{\lambda+\mu+\partial}{u_k}\}_H}_\to
} \\
\displaystyle{
=
-\sum_{s,t=1}^\ell\sum_{m=0}^M
\{{b_{st;m}}_{\lambda+\mu+\partial}{u_k}\}_\to
\circ
\Big((\lambda+\partial)^m (B^{-1})_{ti}(\lambda)\Big)
({B^*}^{-1})_{sr}(\mu+\partial) 
\,.}
\end{array}
$$
Combining the above two equations and using the Master Formula \eqref{20110922:eq1} we then get
$$
\begin{array}{l}
\displaystyle{
\{{\{{u_i}_\lambda{u_j}\}_H}_{\lambda+\mu}{u_k}\}_H
} \\
\displaystyle{
=
\sum_{\substack{r,s,t\in I \\ n\in\mb Z_+}}
A_{kt}(\lambda+\mu+\partial)B^{-1}_{ts}(\lambda+\mu+\partial)
(-\lambda-\mu-\partial)^n
\frac{\partial A_{jr}(\lambda+\partial)}{\partial u_s^{(n)}} B^{-1}_{ri}(\lambda)
} \\
\displaystyle{
-
\sum_{\substack{r,p,q,s,t\in I \\ n\in\mb Z_+}}
A_{kq}(\lambda+\mu+\partial)B^{-1}_{qp}(\lambda+\mu+\partial)
(-\lambda-\mu-\partial)^n
} \\
\displaystyle{
\,\,\,\,\,\,\,\,\,\,\,\,\,\,\,\,\,\,\,\,\,\,\,\,\,\,\,\,\,\,\,\,\,\,\,\,
\Big(
\frac{\partial B_{st}(\lambda+\partial)}{\partial u_p^{(n)}}
(B^{-1})_{ti}(\lambda)\Big)
\Big(({B^*}^{-1})_{sr}(\mu+\partial) A^*_{rj}(\mu)\Big)
\,.}
\end{array}
$$
Hence, if we apply, as before, $B^*_{k'k}(\lambda+\mu+\partial)$ on the left,
replace $\lambda$ by $\lambda+\partial$ acting on $B_{ii'}(\lambda)$,
replace $\mu$ by $\mu+\partial$ acting on $B_{jj'}(\mu)$,
and sum over $i,j,k\in I$, we get, using \eqref{20120107:eq1},
\begin{equation}\label{20120109:eq4}
\begin{array}{l}
\displaystyle{
\sum_{i,j,k\in I}B^*_{k'k}(\lambda+\mu+\partial)
\{{\{{u_i}_x{u_j}\}_H}_{x+y}{u_k}\}_H
\Big(\Big|_{x=\lambda+\partial}B_{ii'}(\lambda)\Big)
} \\
\displaystyle{
\Big(\Big|_{y=\mu+\partial}B_{jj'}(\mu)\Big)
=
-\sum_{j,k\in I}\sum_{n\in\mb Z_+}
A^*_{k'k}(\lambda+\mu+\partial)
(-\lambda-\mu-\partial)^n
} \\
\displaystyle{
\,\,\,\,\,\,\,\,\,\,\,\,\,\,\,\,\,\,\,\,\,\,\,\,\,\,\,\,\,\,\,\,\,\,\,\,
\,\,\,\,\,\,\,\,\,\,\,\,\,\,\,\,\,\,
\bigg(
\frac{\partial A_{ji'}(\lambda)}{\partial u_k^{(n)}} B_{jj'}(\mu)
+
\frac{\partial B_{ji'}(\lambda)}{\partial u_k^{(n)}}
A_{jj'}(\mu)
\bigg)
\,.}
\end{array}
\end{equation}
Combining \eqref{20120109:eq2}, \eqref{20120109:eq3} and \eqref{20120109:eq4},
we get that the expression
$$
\begin{array}{l}
\displaystyle{
\sum_{i,j,k\in I}B^*_{k'k}(\lambda+\mu+\partial)
\Big(
\{{u_i}_x\{{u_j}_y{u_k}\}_H\}_H
-\{{u_j}_y\{{u_i}_x{u_k}\}_H\}_H
} \\
\displaystyle{
-\{{\{{u_i}_x{u_j}\}_H}_{x+y}{u_k}\}_H
\Big)
\Big(\Big|_{x=\lambda+\partial}B_{ii'}(\lambda)\Big)
\Big(\Big|_{y=\mu+\partial}B_{jj'}(\mu)\Big)
}
\end{array}
$$
is the same as the LHS of \eqref{20120103:eq2}.
The claim follows from the assumption that 
the matrix $B\in\Mat_{\ell\times\ell}\mc V[\partial]$
has non-zero Dieudonn\`e determinant.
\end{proof}
\begin{theorem}\label{20111020:thm}
Let $\mc V$ be an algebra of differential functions in the variables $u_1,\dots,u_\ell$,
which is a domain, and let $\mc K$ be its field of fractions.
Let $H=A B^{-1}$, with $A,B\in\Mat_{\ell\times\ell}\mc V[\partial]$, $\det B\neq0$,
be a minimal fractional decomposition (cf. Definition \ref{def:minimal-fraction} 
and Remark \ref{rem:minimal-fraction2}) of the rational matrix pseudodifferential 
operator $H\in\Mat_{\ell\times\ell}\mc V(\partial)$.
Then the subspace 
\begin{equation}\label{eq:dirac}
\mc L_{A,B}(\mc K)=\big\{B(\partial)X\oplus A(\partial)A\,\big|\,X\in\mc K^{\oplus\ell}\big\}
\,\subset\mc K^{\oplus\ell}\oplus\mc K^\ell\,,
\end{equation}
is a Dirac structure if and only if
$H$ is a non-local Hamiltonian structure on $\mc V$.
\end{theorem}
\begin{proof}
It immediately follows from Remark \ref{rem:minimal-fraction2} 
and Propositions \ref{20120103:propa}, \ref{prop:max-isotrop}, 
\ref{20120103:propb} and \ref{20120103:propc}.
\end{proof}
\begin{remark}
We may define a ``generalized'' Dirac structure
as a subspace $\mc L$ of $\mc V^{\oplus\ell}\oplus\mc V^\ell$,
such that $\mc L\subset\mc L^\perp$ (i.e. $\mc L$ is isotropic),
and $\mc L\circ\mc L\subset\mc L^\perp$ (i.e. condition (i) in Proposition \ref{20120103:propb} holds),
where $\mc L^\perp$ is the orthogonal complement to $\mc L$ 
with respect to the inner product \eqref{20111020:eq3}.
Note that a Dirac structure is a special case of this when $\mc L$ is maximal isotropic.
If $A,B\in\Mat_{\ell\times\ell}\mc V[\partial]$, $\det B\neq0$,
then $\mc L_{A,B}$ is a generalized Dirac structure
if and only if $H=AB^{-1}$ is a non-local Hamiltonian structure on $\mc V$
(not necessarily in its minimal fractional decomposition).
Note also that any subspace of a generalized Dirac structure is a generalized Dirac structure.
\end{remark}

\subsection{Compatible pairs of Dirac structures}
\label{sec:6.3}

The notion of compatibility of Dirac structures was introduced by Gelfand and Dorfman 
\cite{GD80}, \cite{Dor93} (see also \cite{BDSK09}).
In this paper we introduce a weaker, but more natural, 
notion of compatibility, which still can be used to implement successfully 
the Lenard-Magri scheme of integrability, and which is more closely related 
to the notion of compatibility of the corresponding non-local Hamiltonian structures.

Given two Dirac structures $\mc L$ and $\mc L^\prime\subset\mc V^{\oplus\ell}\oplus\mc V^\ell$,
we define the relations 
\begin{equation}\label{eq:20090322_1}
\begin{array}{l}
\mc N_{\mc L,\mc L^\prime}
=
\big\{P\oplus P^\prime
\,\big|\,
F\oplus P\in\mc L,\,F\oplus P^\prime\in\mc L^\prime\,\text{ for some } F\in\mc V^{\oplus\ell}\big\}
\subset\mc V^\ell\oplus\mc V^\ell\,,
\\
\mc N_{\mc L,\mc L^\prime}\wcheck
=
\big\{F\oplus F^\prime
\,\big|\,
F\oplus P\in\mc L,\,F^\prime\oplus P\in\mc L^\prime\,\text{ for some } P\in\mc V^{\ell}\big\}
\subset\mc V^{\oplus\ell}\!\oplus\!\mc V^{\oplus\ell}.
\end{array}
\end{equation}
\begin{definition}\label{2006_NRel}
Two Dirac structures $\mc{L},\,\mc{L}^\prime\,\subset \mc V^{\oplus\ell}\oplus\mc V^\ell$ 
are said to be \emph{compatible} if 
for all $P,P^\prime,Q,Q^\prime\in\mc V^\ell,\,
F,F^\prime,F^{\prime\prime}\in\mc V^{\oplus\ell}$
such that
$$
P\oplus P^\prime,\,Q\oplus Q^\prime\,\in\mc N_{\mc L,\mc L^\prime}
\,\,\,\,\text{ and }\,\,\,\,
F\oplus F^\prime,\,F^\prime\oplus F^{\prime\prime}
\in \mc N_{\mc L,\mc L^\prime}\wcheck\,,
$$ 
we have
\begin{equation}\label{eq:20090320_1}
(F|[P,Q])-(F^\prime|[P,Q^\prime])-(F^\prime|[P^\prime,Q])
+(F^{\prime\prime}|[P^\prime,Q^\prime])\,=\,0\,,
\end{equation}
where, as before, $(F|P)=\tint F\cdot P$,
and, for $P,Q\in\mc V^\ell$, $[P,Q]$ is given by \eqref{20120126:eq1}.
\end{definition}
\begin{remark}\label{oldcompatibility}
The original notion of compatibility, introduced by Dorfman in \cite{Dor93},
is similar, except that $\mc N_{\mc L,\mc L^\prime}\wcheck$
is replaced by the ``dual'' relation
$$
\mc N^*_{\mc L,\mc L^\prime}
\,=\,
\big\{F\oplus F^\prime\in\mc V^{\oplus\ell}\oplus\mc V^{\oplus\ell}
\,\big|\,
\tint F\cdot P=\tint F^\prime\cdot P^\prime\,\text{ for all } P\oplus P^\prime
\in\mc N_{\mc L,\mc L^\prime}\big\}\,.
$$
Since $\mc L$ and $\mc L^\prime$ are isotropic, we have, for 
$F\oplus F^\prime\in\mc N_{\mc L,\mc L^\prime}\wcheck$,
and for $Q\oplus Q^\prime\in\mc N_{\mc L,\mc L^\prime}$,
$\tint F\cdot Q=-\tint G\cdot P=\tint F^\prime\cdot Q^\prime$,
where $P\in\mc V^\ell$ and $G\in\mc V^{\oplus\ell}$ are such that 
$F\oplus P,G\oplus Q\in\mc L,\,F^\prime\oplus P,G\oplus Q^\prime\in\mc L^\prime$.
Hence, $\mc N_{\mc L,\mc L^\prime}\wcheck\subset\mc N^*_{\mc L,\mc L^\prime}$.
\end{remark}

Even with the weaker notion of compatibility,
the following important theorem still holds (cf. \cite[Thm.4.13]{BDSK09}).
\begin{theorem}\label{mtst}
Let $(\mc L,\mc L^\prime)$ be a pair of compatible Dirac structures.
Let $F_0,F_1,F_2\in\mc V^{\oplus\ell}$ be such that:
\begin{enumerate}[(i)]
\item
$D_{F_n}^*(\partial)=D_{F_n}(\partial)$, for $n=0,1$;
\item
$F_0\oplus F_1,\,F_1\oplus F_2\,\in\mc N_{\mc L,\mc L^\prime}\wcheck$.
\end{enumerate}
Then, for all $P\oplus P^\prime,Q\oplus Q^\prime\in\mc N_{\mc L,\mc L^\prime}$, we have
\begin{equation}\label{20120405:eq2}
\tint Q^\prime\cdot \big(D_{F_2}(\partial)-D_{F_2}^*(\partial)\big)P^\prime
=0\,.
\end{equation}
\end{theorem}
\begin{proof}
By the assumption \eqref{eq:20090320_1}, we have
$$
\begin{array}{l}
\displaystyle{
0=(F_0|[P,Q])-(F_1|[P,Q^\prime])-(F_1|[P^\prime,Q])
+(F_2|[P^\prime,Q^\prime])
} \\
\displaystyle{
=\int \Big(
F_0\cdot D_Q(\partial)P -F_0\cdot D_P(\partial)Q
-F_1\cdot D_{Q^\prime}(\partial)P +F_1\cdot D_P(\partial)Q^\prime
} \\
\displaystyle{
-F_1\cdot D_Q(\partial)P^\prime +F_1\cdot D_{P^\prime}(\partial)Q
+F_2\cdot D_{Q^\prime}(\partial)P^\prime -F_2\cdot D_{P^\prime}(\partial)Q^\prime
\Big)
} \\
\displaystyle{
=\int P\cdot \frac{\delta}{\delta u} \big((F_0|Q)-(F_1|Q^\prime)\big)
-\int Q\cdot\frac{\delta}{\delta u} \big((F_0|P)-(F_1|P^\prime)\big)
} \\
\displaystyle{
-\int P^\prime\cdot \frac{\delta}{\delta u}\big((F_1|Q)- (F_2|Q^\prime)\big)
+\int Q^\prime\cdot \frac{\delta}{\delta u} \big((F_1|P)-(F_2|P^\prime)\big)
} \\
\displaystyle{
-\int Q\cdot D_{F_0}(\partial)P 
+\int P\cdot D_{F_0}(\partial)Q
+\int {Q^\prime}\cdot D_{F_1}(\partial)P 
-\int P\cdot D_{F_1}(\partial)Q^\prime
} \\
\displaystyle{
+\int Q\cdot D_{F_1}(\partial)P^\prime 
-\int {P^\prime}\cdot D_{F_1}(\partial)Q
-\int {Q^\prime}\cdot D_{F_2}(\partial)P^\prime 
+\int {P^\prime}\cdot D_{F_2}(\partial)Q^\prime
\,.
}\end{array}
$$
In the second identity we used the definition \eqref{20120126:eq1} of the Lie bracket on $\mc V^{\ell}$,
and in the last identity we used equation \eqref{20120405:eq1}.
Since, by assumption, $F_0\oplus F_1\in\mc N_{\mc L,\mc L^\prime}\wcheck$ 
and $Q\oplus Q^\prime\in\mc N_{\mc L,\mc L^\prime}$, we have (by Remark \ref{oldcompatibility})
that $(F_0|Q)=(F_1|Q^\prime)$. 
Hence the first term in the RHS above is zero,
and, by the same argument, the first four terms are zero.
The following six terms are also zero since, by assumption, $D_{F_0}(\partial)$
and $D_{F_1}(\partial)$ are selfadjoint.
In conclusion, equation \eqref{20120405:eq2} holds.
\end{proof}

\subsection{Compatible non-local Hamiltonian structures 
and corresponding compatible pairs of Dirac structures}
\label{sec:6.4}

In Theorem \ref{20111020:thm} we proved that to a non-local Hamiltonian 
structure $H\in\Mat_{\ell\times\ell}\mc V(\partial)$
in its minimal fractional decomposition $H=AB^{-1}$, 
with $A,B\in\Mat_{\ell\times\ell}\mc V[\partial]$,
there corresponds a Dirac structure $\mc L_{A,B}(\mc K)$ 
on the field of frations $\mc K$.
In this section we prove that to a compatible pair 
of non-local Hamiltonian structures $H=AB^{-1},\,K=CD^{-1}$,
in their minimal fractional decompositions,
there corresponds a compatible pair 
of Dirac structures $\mc L_{A,B}(\mc K),\,\mc L_{C,D}(\mc K)$ on $\mc K$.
This is stated in the following:
\begin{theorem}\label{20120126:prop2}
Let $\mc V$ be an algebra of differential functions in $u_1,\dots,u_\ell$,
which is a domain, and let $\mc K$ be its field of fractions.
Let $H,\,K\in\Mat_{\ell\times\ell}\mc V(\partial)$ be compatible 
non-local Hamiltonian structures on $\mc V$.
Let $H=AB^{-1},\,K=CD^{-1}$ be their minimal fractional decompositions
(cf. Definition \ref{def:minimal-fraction}).
Then $\mc L_{A,B}(\mc K)$ and $\mc L_{C,D}(\mc K)$
are compatible Dirac structures on $\mc K$.
\end{theorem}

By Theorem \ref{20110923:prop}, 
the Hamiltonian structures $H$ and $K$ on $\mc V$ are compatible
if and only if we have the following ``mixed'' Jacobi identity on generators
($i,j,k\in I$):
\begin{equation}\label{20120405:eq3}
\begin{array}{l}
\{{u_i}_\lambda\{{u_j}_\mu {u_k}\}_H\}_K-\{{u_j}_\mu\{{u_i}_\lambda {u_k}\}_H\}_K
-\{{\{{u_i}_\lambda {u_j}\}_H}_{\lambda+\mu} {u_k}\}_K \\
+\{{u_i}_\lambda\{{u_j}_\mu {u_k}\}_K\}_H-\{{u_j}_\mu\{{u_i}_\lambda {u_k}\}_K\}_H
-\{{\{{u_i}_\lambda {u_j}\}_K}_{\lambda+\mu} {u_k}\}_H
=0\,,
\end{array}\end{equation}

In order to relate the above condition to the compatibility 
of the corresponding Dirac structures $\mc L_{A,B}$ and $\mc L_{C,D}$,
we need to compute explicitly each term of the above equation.
This is done in the following:
\begin{lemma}\label{20120405:lem1}
Suppose that the pairs $(A,B)$ and $(C,D)$, with $A,B,C,D\in\Mat_{\ell\times\ell}\mc V[\partial]$,
satisfy equation \eqref{20120107:eq1}:
\begin{equation}\label{20120405:skew}
A^*\circ B+B^*\circ A=0
\,\,,\,\,\,\,
C^*\circ D+D^*\circ C=0\,.
\end{equation}
Assume that $B$ and $D$ have non-zero Dieudonn\`e determinant,
and that the (skewadjoint) rational matrix pseudodifferential operators 
$H=A B^{-1}$ and $K=CD^{-1}$
have coefficients in $\mc V$, i.e. $H,K\in\Mat_{\ell\times\ell}\mc V(\partial)$.
Consider the corresponding non-local $\lambda$-brackets 
$\{\cdot\,_\lambda\,\cdot\}_H$ and $\{\cdot\,_\lambda\,\cdot\}_K$
given by the Master Formula \eqref{20110922:eq1}.
Then, in terms of notation \eqref{20111018:eq5}, 
we have the following identities for every $i',j',k'\in I$:
\begin{equation}\label{20120405:A1}
\begin{array}{c}
\displaystyle{
\sum_{i,j,k\in I}
B^*_{k'k}(\lambda+\mu+\partial)
\{{u_i}_x\{{u_j}_y {u_k}\}_H\}_K
\big(\big|_{x=\lambda+\partial}D_{ii'}(\lambda)\big)
\big(\big|_{y=\mu+\partial}B_{jj'}(\mu)\big)
} \\
\displaystyle{
=
\sum_{i,k\in I}\sum_{n\in\mb Z_+}
B^*_{k'k}(\lambda+\mu+\partial)
\frac{\partial A_{kj'}(\mu)}{\partial u_i^{(n)}}
(\lambda+\partial)^n C_{ii'}(\lambda)
} \\
\displaystyle{
+\sum_{i,k\in I}\sum_{n\in\mb Z_+}
A^*_{k'k}(\lambda+\mu+\partial)
\frac{\partial B_{kj'}(\mu)}{\partial u_i^{(n)}}
(\lambda+\partial)^n C_{ii'}(\lambda)
\,,} 
\end{array}
\end{equation}
\begin{equation}\label{20120405:A2}
\begin{array}{c}
\displaystyle{
\sum_{i,j,k\in I}
D^*_{k'k}(\lambda+\mu+\partial)
\{{u_i}_x\{{u_j}_y {u_k}\}_K\}_H
\big(\big|_{x=\lambda+\partial}B_{ii'}(\lambda)\big)
\big(\big|_{y=\mu+\partial}D_{jj'}(\mu)\big)
} \\
\displaystyle{
=
\sum_{i,k\in I}\sum_{n\in\mb Z_+}
D^*_{k'k}(\lambda+\mu+\partial)
\frac{\partial C_{kj'}(\mu)}{\partial u_i^{(n)}}
(\lambda+\partial)^n A_{ii'}(\lambda)
} \\
\displaystyle{
+\sum_{i,k\in I}\sum_{n\in\mb Z_+}
C^*_{k'k}(\lambda+\mu+\partial)
\frac{\partial D_{kj'}(\mu)}{\partial u_i^{(n)}}
(\lambda+\partial)^n A_{ii'}(\lambda)
\,,} 
\end{array}
\end{equation}
\begin{equation}\label{20120405:B1}
\begin{array}{c}
\displaystyle{
\sum_{i,j,k\in I}
B^*_{k'k}(\lambda+\mu+\partial)
\{{u_j}_y\{{u_i}_x {u_k}\}_H\}_K
\big(\big|_{x=\lambda+\partial}B_{ii'}(\lambda)\big)
\big(\big|_{y=\mu+\partial}D_{jj'}(\mu)\big)
} \\
\displaystyle{
=
\sum_{j,k\in I}\sum_{n\in\mb Z_+}
B^*_{k'k}(\lambda+\mu+\partial)
\frac{\partial A_{ki'}(\lambda)}{\partial u_j^{(n)}}
(\mu+\partial)^n C_{jj'}(\mu)
} \\
\displaystyle{
+\sum_{j,k\in I}\sum_{n\in\mb Z_+}
A^*_{k'k}(\lambda+\mu+\partial)
\frac{\partial B_{ki'}(\lambda)}{\partial u_j^{(n)}}
(\mu+\partial)^n C_{jj'}(\mu)
\,,} 
\end{array}
\end{equation}
\begin{equation}\label{20120405:B2}
\begin{array}{c}
\displaystyle{
\sum_{i,j,k\in I}
D^*_{k'k}(\lambda+\mu+\partial)
\{{u_j}_y\{{u_i}_x {u_k}\}_K\}_H
\big(\big|_{x=\lambda+\partial}D_{ii'}(\lambda)\big)
\big(\big|_{y=\mu+\partial}B_{jj'}(\mu)\big)
} \\
\displaystyle{
=
\sum_{j,k\in I}\sum_{n\in\mb Z_+}
D^*_{k'k}(\lambda+\mu+\partial)
\frac{\partial C_{ki'}(\lambda)}{\partial u_j^{(n)}}
(\mu+\partial)^n A_{jj'}(\mu)
} \\
\displaystyle{
+\sum_{j,k\in I}\sum_{n\in\mb Z_+}
C^*_{k'k}(\lambda+\mu+\partial)
\frac{\partial D_{ki'}(\lambda)}{\partial u_j^{(n)}}
(\mu+\partial)^n A_{jj'}(\mu)
\,,} 
\end{array}
\end{equation}
\begin{equation}\label{20120405:C1}
\begin{array}{c}
\displaystyle{
\sum_{i,j,k\in I}
D^*_{k'k}(\lambda+\mu+\partial)
\{{\{{u_i}_x{u_j}\}_H}_{x+y}{u_k}\}_K
\big(\big|_{x=\lambda+\partial}B_{ii'}(\lambda)\big)
\big(\big|_{y=\mu+\partial}B_{jj'}(\mu)\big)
} \\
\displaystyle{
=
-\sum_{j,k\in I}\sum_{n\in\mb Z_+}
C^*_{k'k}(\lambda+\mu+\partial)
(-\lambda-\mu-\partial)^n
\frac{\partial A_{ji'}(\lambda)}{\partial u_k^{(n)}}
B_{jj'}(\mu)
} \\
\displaystyle{
-\sum_{j,k\in I}\sum_{n\in\mb Z_+}
C^*_{k'k}(\lambda+\mu+\partial)
(-\lambda-\mu-\partial)^n
\frac{\partial B_{ji'}(\lambda)}{\partial u_k^{(n)}}
A_{jj'}(\mu)
\,,} 
\end{array}
\end{equation}
\begin{equation}\label{20120405:C2}
\begin{array}{c}
\displaystyle{
\sum_{i,j,k\in I}
B^*_{k'k}(\lambda+\mu+\partial)
\{{\{{u_i}_x{u_j}\}_K}_{x+y}{u_k}\}_H
\big(\big|_{x=\lambda+\partial}D_{ii'}(\lambda)\big)
\big(\big|_{y=\mu+\partial}D_{jj'}(\mu)\big)
} \\
\displaystyle{
=
-\sum_{j,k\in I}\sum_{n\in\mb Z_+}
A^*_{k'k}(\lambda+\mu+\partial)
(-\lambda-\mu-\partial)^n
\frac{\partial C_{ji'}(\lambda)}{\partial u_k^{(n)}}
D_{jj'}(\mu)
} \\
\displaystyle{
-\sum_{j,k\in I}\sum_{n\in\mb Z_+}
A^*_{k'k}(\lambda+\mu+\partial)
(-\lambda-\mu-\partial)^n
\frac{\partial D_{ji'}(\lambda)}{\partial u_k^{(n)}}
C_{jj'}(\mu)
\,.} 
\end{array}
\end{equation}
\end{lemma}
\begin{proof}
For equation \eqref{20120405:A1}, 
we can use the Leibniz rule and equation \eqref{20111012:eq2c} to get
$$
\begin{array}{l}
\displaystyle{
\{{u_i}_x\{{u_j}_y {u_k}\}_H\}_K
=\sum_{r\in I}\{{u_i}_x A_{kr}(y+\partial)(B^{-1})_{rj}(y)\}_K
} \\
\displaystyle{
=\sum_{r\in I}\sum_{m\in\mb Z_+}\{{u_i}_x a_{kr;m}\}_K
(y+\partial)^m(B^{-1})_{rj}(y)
} \\
\displaystyle{
-\sum_{\substack{r,p,q\in I \\ m\in\mb Z_+}}\!\!\!
A_{kr}(x+y+\partial)
(B^{-1})_{rp}(x+y+\partial)
\{{u_i}_\lambda b_{pq;m}\}_K
(y+\partial)^m(B^{-1})_{qj}(y)\,.
}
\end{array}
$$
We can then use the Master Formula \eqref{20110922:eq1} to get
$$
\begin{array}{l}
\displaystyle{
\{{u_i}_x\{{u_j}_y {u_k}\}_H\}_K
=\sum_{r,s\in I}\sum_{n\in\mb Z_+}
\Big(\frac{\partial A_{kr}(y+\partial)}{\partial u_s^{(n)}} (B^{-1})_{rj}(y)\Big)
(x+\partial)^n K_{si}(x)
} \\
\displaystyle{
-\sum_{\substack{r,p,q,s\in I \\ n\in\mb Z_+}}\!\!\!
A_{kr}(x+y+\partial)
(B^{-1})_{rp}(x+y+\partial)
\Big(\frac{\partial B_{pq}(y+\partial)}{\partial u_s^{(n)}} (B^{-1})_{qj}(y)\Big)
(x+\partial)^n K_{si}(x)
\,.
}
\end{array}
$$
If we now replace $x$ with $\lambda+\partial$ acting on $D_{ii'}(\lambda)$
and $y$ by $\mu+\partial$ acting on $B_{jj'}(\mu)$,
and we apply $B^*_{k'k}(\lambda+\mu+\partial)$, acting from the left,
to both sides of the above equation, 
we get, after using the assuption \eqref{20120405:skew},
that equation \eqref{20120405:A1} holds.
Equation \eqref{20120405:A2} is obtained from \eqref{20120405:A1}
by exchanging the roles of $H$ and $K$.
Equation \eqref{20120405:B1} is obtained from \eqref{20120405:A1}
by exchanging $\lambda$ with $\mu$ and $i$ and $i'$ with $j$ and $j'$ respectively,
and equation \eqref{20120405:B2} is obtained
from \eqref{20120405:B1} by exchaing the roles of $H$ and $K$.
Finally, equations \eqref{20120405:C1} and \eqref{20120405:C2} 
can be derived with a similar computation,
which involves the right Leibniz rule (instead of the left)
and equation \eqref{20111012:eq2d} (instead of \eqref{20111012:eq2c}).
\end{proof}

Let us next describe the relations \eqref{eq:20090322_1}
associated to Dirac structures $\mc L_{A,B}(\mc K)$ and $\mc L_{C,D}(\mc K)$
defined in \eqref{eq:dirac}.
We have
\begin{equation}\label{20120405:eq5}
\begin{array}{l}
\mc N_{\mc L_{A,B}(\mc K),\mc L_{C,D}(\mc K)}
=
\big\{
A(\partial)X\oplus C(\partial)X^\prime
\,\big|\,
X,X^\prime\in\mc K^{\oplus\ell}\,,\,\,B(\partial)X=D(\partial)X^\prime
\big\}
\,,
\\
\mc N_{\mc L_{A,B}(\mc K),\mc L_{C,D}(\mc K)}\wwcheck
=
\big\{
B(\partial)Z\oplus D(\partial)Z^\prime
\,\big|\,
Z,Z^\prime\in\mc K^{\oplus\ell}\,,\,\,A(\partial)Z=C(\partial)Z^\prime
\big\}
\,.
\end{array}
\end{equation}
Hence, by Definition \ref{2006_NRel},
the Dirac structures $\mc L_{A,B}$ and $\mc L_{C,D}$
are compatible if and only if,
for every $X,X^\prime,Y,Y^\prime,Z,Z^\prime,W,W^\prime\in\mc V^{\oplus\ell}$ 
such that
\begin{equation}\label{20120405:eq6}
\begin{array}{c}
\displaystyle{
\vphantom{\Big(}
B(\partial)X=D(\partial)X^\prime
\,\,,\,\,\,\,
B(\partial)Y=D(\partial)Y^\prime
\,\,,\,\,\,\,
B(\partial)W=D(\partial)Z^\prime
\,,}\\
\displaystyle{
\vphantom{\Big(}
A(\partial)Z=C(\partial)Z^\prime
\,\,,\,\,\,\,
A(\partial)W=C(\partial)W^\prime
\,,}
\end{array}
\end{equation}
we have the following identity:
\begin{equation}\label{20120405:eq7}
\begin{array}{l}
\displaystyle{
\vphantom{\Big(}
\big(B(\partial)Z\big|[A(\partial)X,A(\partial)Y]\big)
-\big(D(\partial)Z^\prime\big|[A(\partial)X,C(\partial)Y^\prime]\big)
}\\
\displaystyle{
\vphantom{\Big(}
-\big(B(\partial)W\big|[C(\partial)X^\prime,A(\partial)Y]\big)
+\big(D(\partial)W^\prime\big|[C(\partial)X^\prime,C(\partial)Y^\prime]\big)
=0
\,.}
\end{array}\end{equation}

\begin{lemma}\label{20120405:lem2}
Suppose that $H=AB^{-1}$ and $K=CD^{-1}$ are non-local Hamiltonian structures,
and that conditions \eqref{20120405:eq6} hold. 
Then equation \eqref{20120405:eq7} is equivalent to the following equation:
\begin{equation}\label{20120405:eq16}
\begin{array}{l}
\displaystyle{
\vphantom{\Big(}
-\tint (A(\partial)Y)\cdot D_{B(\partial)X}(\partial)C(\partial)Z^\prime
+\tint (A(\partial)Y)\cdot D^*_{B(\partial)X}(\partial) C(\partial)Z^\prime
} \\
\displaystyle{
\vphantom{\Big(}
+\tint (B(\partial)Y)\cdot D_{C(\partial)Z^\prime}(\partial)A(\partial)X
-\tint (B(\partial)Y)\cdot D_{A(\partial)X}(\partial) C(\partial)Z^\prime
} \\
\displaystyle{
\vphantom{\Big(}
+\tint (C(\partial)Y^\prime)\cdot D_{D(\partial)Z^\prime}(\partial)A(\partial)X
+\tint (C(\partial)Y^\prime)\cdot D^*_{A(\partial)X}(\partial) D(\partial)Z^\prime
} \\
\displaystyle{
\vphantom{\Big(}
+\tint (A(\partial)Y)\cdot D_{B(\partial)W}(\partial)C(\partial)X^\prime
+\tint (A(\partial)Y)\cdot D^*_{C(\partial)X^\prime}(\partial) B(\partial)W
} \\
\displaystyle{
\vphantom{\Big(}
-\tint (C(\partial)Y^\prime)\cdot D_{D(\partial)X^\prime}(\partial)A(\partial)W
+\tint (C(\partial)Y^\prime)\cdot D^*_{D(\partial)X^\prime}(\partial) A(\partial)W
} \\
\displaystyle{
\vphantom{\Big(}
+\tint (D(\partial)Y^\prime)\cdot D_{A(\partial)W}(\partial)C(\partial)X^\prime
-\tint (D(\partial)Y^\prime)\cdot D_{C(\partial)X^\prime}(\partial) A(\partial)W
=0\,.
}\end{array}
\end{equation}
\end{lemma}
\begin{proof}
By \eqref{20120126:eq1} and \eqref{20120405:eq1}, we have 
\begin{equation}\label{20120405:eq8}
\begin{array}{l}
\displaystyle{
\vphantom{\Big(}
\big(B(\partial)Z\big|[A(\partial)X,A(\partial)Y]\big)
=\tint (B(\partial)Z)\cdot D_{A(\partial)Y}(\partial)A(\partial)X
} \\
\displaystyle{
\vphantom{\Big(}
-\tint (B(\partial)Z)\cdot D_{A(\partial)X}(\partial)A(\partial)Y
=\tint (A(\partial)X)\cdot\frac{\delta}{\delta u}(B(\partial)Z|A(\partial)Y)
} \\
\displaystyle{
\vphantom{\Big(}
-\tint (A(\partial)Y)\cdot D_{B(\partial)Z}(\partial)A(\partial)X
-\tint (A(\partial)Y)\cdot D^*_{A(\partial)X}(\partial) B(\partial)Z
\,.
}\end{array}
\end{equation}
Similarly, we have
\begin{equation}\label{20120405:eq9}
\begin{array}{l}
\displaystyle{
\vphantom{\Big(}
\big(D(\partial)Z^\prime\big|[A(\partial)X,C(\partial)Y^\prime]\big)
=\tint (A(\partial)X)\cdot\frac{\delta}{\delta u}(D(\partial)Z^\prime|C(\partial)Y^\prime)
} \\
\displaystyle{
\vphantom{\Big(}
-\tint (C(\partial)Y^\prime)\cdot D_{D(\partial)Z^\prime}(\partial)A(\partial)X
-\tint (C(\partial)Y^\prime)\cdot D^*_{A(\partial)X}(\partial) D(\partial)Z^\prime
\,,
}\end{array}
\end{equation}
\begin{equation}\label{20120405:eq10}
\begin{array}{l}
\displaystyle{
\vphantom{\Big(}
\big(B(\partial)W\big|[C(\partial)X^\prime,A(\partial)Y]\big)
=\tint (C(\partial)X^\prime)\cdot\frac{\delta}{\delta u}(B(\partial)W|A(\partial)Y)
} \\
\displaystyle{
\vphantom{\Big(}
-\tint (A(\partial)Y)\cdot D_{B(\partial)W}(\partial)C(\partial)X^\prime
-\tint (A(\partial)Y)\cdot D^*_{C(\partial)X^\prime}(\partial) B(\partial)W
\,,
}\end{array}
\end{equation}
and
\begin{equation}\label{20120405:eq11}
\begin{array}{l}
\displaystyle{
\vphantom{\Big(}
\big(D(\partial)W^\prime\big|[C(\partial)X^\prime,C(\partial)Y^\prime]\big)
=\tint (C(\partial)X^\prime)\cdot\frac{\delta}{\delta u}(D(\partial)W^\prime|C(\partial)Y^\prime)
} \\
\displaystyle{
\vphantom{\Big(}
-\tint (C(\partial)Y^\prime)\cdot D_{D(\partial)W^\prime}(\partial)C(\partial)X^\prime
-\tint (C(\partial)Y^\prime)\cdot D^*_{C(\partial)X^\prime}(\partial) D(\partial)W^\prime
\,.
}\end{array}
\end{equation}
By the skewadnointness of $H$ and $K$, which translates to \eqref{20120405:skew}, 
and by conditions \eqref{20120405:eq6}, we have
$$
\begin{array}{l}
(B(\partial)Z|A(\partial)Y)
=-(A(\partial)Z|B(\partial)Y) \\
=-(C(\partial)Z^\prime|D(\partial)Y^\prime)
=(D(\partial)Z^\prime|C(\partial)Y^\prime)
\,,
\end{array}
$$
hence the first terms in the RHS of \eqref{20120405:eq8} and \eqref{20120405:eq9}
cancel.
Similarly for the first terms in the RHS of \eqref{20120405:eq10} and \eqref{20120405:eq11}.
Therefore, combining equations \eqref{20120405:eq8}--\eqref{20120405:eq11}, we get that
equation \eqref{20120405:eq7} is equivalent to
\begin{equation}\label{20120405:eq12}
\begin{array}{l}
\displaystyle{
\vphantom{\Big(}
-\tint (A(\partial)Y)\cdot D_{B(\partial)Z}(\partial)A(\partial)X
-\tint (A(\partial)Y)\cdot D^*_{A(\partial)X}(\partial) B(\partial)Z
} \\
\displaystyle{
\vphantom{\Big(}
+\tint (C(\partial)Y^\prime)\cdot D_{D(\partial)Z^\prime}(\partial)A(\partial)X
+\tint (C(\partial)Y^\prime)\cdot D^*_{A(\partial)X}(\partial) D(\partial)Z^\prime
} \\
\displaystyle{
\vphantom{\Big(}
+\tint (A(\partial)Y)\cdot D_{B(\partial)W}(\partial)C(\partial)X^\prime
+\tint (A(\partial)Y)\cdot D^*_{C(\partial)X^\prime}(\partial) B(\partial)W
} \\
\displaystyle{
\vphantom{\Big(}
-\tint (C(\partial)Y^\prime)\cdot D_{D(\partial)W^\prime}(\partial)C(\partial)X^\prime
-\tint (C(\partial)Y^\prime)\cdot D^*_{C(\partial)X^\prime}(\partial) D(\partial)W^\prime
=0\,.
}\end{array}
\end{equation}
Next, since by assumption $H=AB^{-1}$ is a non-local Hamiltonian structure,
it follows by Propositions \ref{20120103:propc} and \ref{20120103:propb}
that equation \eqref{20120103:eq1} holds.
In particular,
\begin{equation}\label{20120405:eq14}
\begin{array}{l}
\displaystyle{
\vphantom{\Big(}
-\tint (A(\partial)Y)\cdot D_{B(\partial)Z}(\partial)A(\partial)X
-\tint (A(\partial)Y)\cdot D^*_{A(\partial)X}(\partial) B(\partial)Z
} \\
\displaystyle{
\vphantom{\Big(}
=-\tint (A(\partial)Y)\cdot D_{B(\partial)X}(\partial)A(\partial)Z
+\tint (A(\partial)Y)\cdot D^*_{B(\partial)X}(\partial) A(\partial)Z
} \\
\displaystyle{
\vphantom{\Big(}
+\tint (B(\partial)Y)\cdot D_{A(\partial)Z}(\partial)A(\partial)X
-\tint (B(\partial)Y)\cdot D_{A(\partial)X}(\partial) A(\partial)Z
\,.
}\end{array}
\end{equation}
Similarly, using the assumption that $K=CD^{-1}$ is a non-local Hamiltonian structure,
we get
\begin{equation}\label{20120405:eq15}
\begin{array}{l}
\displaystyle{
\vphantom{\Big(}
-\tint (C(\partial)Y^\prime)\cdot D_{D(\partial)W^\prime}(\partial)C(\partial)X^\prime
-\tint (C(\partial)Y^\prime)\cdot D^*_{C(\partial)X^\prime}(\partial) D(\partial)W^\prime
} \\
\displaystyle{
\vphantom{\Big(}
=-\tint (C(\partial)Y^\prime)\cdot D_{D(\partial)X^\prime}(\partial)C(\partial)W^\prime
+\tint (C(\partial)Y^\prime)\cdot D^*_{D(\partial)X^\prime}(\partial) C(\partial)W^\prime
} \\
\displaystyle{
\vphantom{\Big(}
+\tint (D(\partial)Y^\prime)\cdot D_{C(\partial)W^\prime}(\partial)C(\partial)X^\prime
-\tint (D(\partial)Y^\prime)\cdot D_{C(\partial)X^\prime}(\partial) C(\partial)W^\prime
\,.
}\end{array}
\end{equation}
Combining equations \eqref{20120405:eq12}, \eqref{20120405:eq14} and \eqref{20120405:eq15},
we get \eqref{20120405:eq16}.
\end{proof}

\begin{proof}[Proof of Theorem \ref{20120126:prop2}]
By Lemma \ref{20120405:lem2},
we only need to prove that, if condition \eqref{20120405:eq3} holds,
then equation \eqref{20120405:eq16} holds
for every $X,X^\prime,Y,Y^\prime,W,Z^\prime$
satisfying the first three identities in \eqref{20120405:eq6}.
It follows by some straightforward computation that
we can rewrite each term in the LHS of \eqref{20120405:eq16} as follows
\begin{equation}\label{20120405:X1}
\begin{array}{l}
\displaystyle{
\vphantom{\Big(}
-\tint (A(\partial)Y)\cdot D_{B(\partial)X}(\partial)C(\partial)Z^\prime
=-\tint (A(\partial)Y)\cdot B(\partial)D_X(\partial)C(\partial)Z^\prime
} \\
\displaystyle{
\vphantom{\Big(}
-\int \!\!\!\!\sum_{\substack{i',j',k'\in I \\ j,k,\in I,\, n\in\mb Z_+}}\!\!\!\!
Y_{k'}A^*_{k'k}(\lambda\!+\!\mu\!+\!\partial)
\frac{\partial B_{ki'}(\lambda)}{\partial u_j^{(n)}} (\mu+\partial)^n C_{jj'}(\mu)
\big(\big|_{\lambda=\partial}X_{i'}\big)\big(\big|_{\mu=\partial}Z^\prime_{j'}\big)
\,,
}\end{array}
\end{equation}
\begin{equation}\label{20120405:X2}
\begin{array}{l}
\displaystyle{
\vphantom{\Big(}
\tint (A(\partial)Y)\cdot D^*_{D(\partial)X^\prime}(\partial) C(\partial)Z^\prime
=\tint (A(\partial)Y)\cdot D^*_{X^\prime}(\partial)D^*(\partial)C(\partial)Z^\prime
} \\
\displaystyle{
\vphantom{\Big(}
+\int \!\!\!\!\!\!\sum_{\substack{i',j',k'\in I \\ j,k,\in I, n\in\mb Z_+}}\!\!\!\!\!\!
Y_{k'}A^*_{k'k}\!(\lambda\!+\!\mu\!+\!\partial)
(\!-\!\lambda\!-\!\mu\!-\!\partial)^n
\frac{\partial D_{ji'}(\lambda)}{\partial u_k^{(n)}} C_{jj'}(\mu)
\big(\big|_{\lambda=\partial}\!X^\prime_{i'}\big)\big(\big|_{\mu=\partial}\!Z^\prime_{j'}\big)
\,,
}\end{array}
\end{equation}
\begin{equation}\label{20120405:X3}
\begin{array}{l}
\displaystyle{
\vphantom{\Big(}
\tint (D(\partial)Y^\prime)\cdot D_{C(\partial)Z^\prime}(\partial)A(\partial)X
=\tint (D(\partial)Y^\prime)\cdot C(\partial)D_{Z^\prime}(\partial)A(\partial)X
} \\
\displaystyle{
\vphantom{\Big(}
+\int \!\!\!\sum_{\substack{i',j',k'\in I \\ i,k,\in I, n\in\mb Z_+}}\!\!\!
Y^\prime_{k'}D^*_{k'k}(\lambda\!+\!\mu\!+\!\partial)
\frac{\partial C_{kj'}(\mu)}{\partial u_i^{(n)}} (\lambda+\partial)^n A_{ii'}(\lambda)
\big(\big|_{\lambda=\partial}X_{i'}\big)\big(\big|_{\mu=\partial}Z^\prime_{j'}\big)
\,,
}\end{array}
\end{equation}
\begin{equation}\label{20120405:X4}
\begin{array}{l}
\displaystyle{
\vphantom{\Big(}
-\tint (B(\partial)Y)\cdot D_{A(\partial)X}(\partial) C(\partial)Z^\prime
=-\tint (B(\partial)Y)\cdot A(\partial)D_X(\partial)C(\partial)Z^\prime
} \\
\displaystyle{
\vphantom{\Big(}
-\int \!\!\!\sum_{\substack{i',j',k'\in I \\ j,k,\in I, n\in\mb Z_+}}\!\!\!
Y_{k'}B^*_{k'k}(\lambda\!+\!\mu\!+\!\partial)
\frac{\partial A_{ki'}(\lambda)}{\partial u_j^{(n)}} (\mu+\partial)^n C_{jj'}(\mu)
\big(\big|_{\lambda=\partial}X_{i'}\big)\big(\big|_{\mu=\partial}Z^\prime_{j'}\big)
\,,
}\end{array}
\end{equation}
\begin{equation}\label{20120405:X5}
\begin{array}{l}
\displaystyle{
\vphantom{\Big(}
\tint (C(\partial)Y^\prime)\cdot D_{D(\partial)Z^\prime}(\partial)A(\partial)X
=\tint (C(\partial)Y^\prime)\cdot D(\partial)D_{Z^\prime}(\partial)A(\partial)X
} \\
\displaystyle{
\vphantom{\Big(}
+\int \!\!\!\sum_{\substack{i',j',k'\in I \\ i,k,\in I, n\in\mb Z_+}}\!\!\!
Y^\prime_{k'}C^*_{k'k}(\lambda\!+\!\mu\!+\!\partial)
\frac{\partial D_{kj'}(\mu)}{\partial u_i^{(n)}} (\lambda+\partial)^n A_{ii'}(\lambda)
\big(\big|_{\lambda=\partial}X_{i'}\big)\big(\big|_{\mu=\partial}Z^\prime_{j'}\big)
\,,
}\end{array}
\end{equation}
\begin{equation}\label{20120405:X6}
\begin{array}{l}
\displaystyle{
\vphantom{\Big(}
\tint (C(\partial)Y^\prime)\cdot D^*_{A(\partial)X}(\partial) B(\partial)W
=\tint (C(\partial)Y^\prime)\cdot D^*_{X}(\partial)A^*(\partial)B(\partial)W
} \\
\displaystyle{
\vphantom{\Big(}
+\int \!\!\!\!\!\!\sum_{\substack{i',j',k'\in I \\ j,k,\in I, n\in\mb Z_+}}\!\!\!\!\!\!\!
Y^\prime_{k'}C^*_{k'k}\!(\lambda\!+\!\mu\!+\!\partial)
(\!-\!\lambda\!-\!\mu\!-\!\partial)^n
\frac{\partial A_{ji'}(\lambda)}{\partial u_k^{(n)}} B_{jj'}(\mu)
\big(\big|_{\lambda=\partial}\!X_{i'}\big)\big(\big|_{\mu=\partial}\!W_{j'}\big)
,
}\end{array}
\end{equation}
\begin{equation}\label{20120405:X7}
\begin{array}{l}
\displaystyle{
\vphantom{\Big(}
\tint (A(\partial)Y)\cdot D_{B(\partial)W}(\partial)C(\partial)X^\prime
=\tint (A(\partial)Y)\cdot B(\partial)D_W(\partial)C(\partial)X^\prime
} \\
\displaystyle{
\vphantom{\Big(}
+\int \!\!\!\sum_{\substack{i',j',k'\in I \\ i,k,\in I, n\in\mb Z_+}}\!\!\!
Y_{k'}A^*_{k'k}(\lambda\!+\!\mu\!+\!\partial)
\frac{\partial B_{kj'}(\mu)}{\partial u_i^{(n)}} (\lambda+\partial)^n C_{ii'}(\lambda)
\big(\big|_{\lambda=\partial}X^\prime_{i'}\big)\big(\big|_{\mu=\partial}W_{j'}\big)
\,,
}\end{array}
\end{equation}
\begin{equation}\label{20120405:X8}
\begin{array}{l}
\displaystyle{
\vphantom{\Big(}
\tint (A(\partial)Y)\cdot D^*_{C(\partial)X^\prime}(\partial) D(\partial)Z^\prime
=\tint (A(\partial)Y)\cdot D^*_{X^\prime}(\partial)C^*(\partial)D(\partial)Z^\prime
} \\
\displaystyle{
\vphantom{\Big(}
+\int \!\!\!\!\!\!\sum_{\substack{i',j',k'\in I \\ j,k,\in I, n\in\mb Z_+}}\!\!\!\!\!\!
Y_{k'}A^*_{k'k}\!(\lambda\!+\!\mu\!+\!\partial)
(\!-\!\lambda\!-\!\mu\!-\!\partial)^n
\frac{\partial C_{ji'}(\lambda)}{\partial u_k^{(n)}} D_{jj'}(\mu)
\big(\big|_{\lambda=\partial}\!X^\prime_{i'}\big)\big(\big|_{\mu=\partial}\!Z^\prime_{j'}\big)
\,,
}\end{array}
\end{equation}
\begin{equation}\label{20120405:X9}
\begin{array}{l}
\displaystyle{
\vphantom{\Big(}
-\tint (C(\partial)Y^\prime)\cdot D_{D(\partial)X^\prime}(\partial)A(\partial)W
=-\tint (C(\partial)Y^\prime)\cdot D(\partial)D_{X^\prime}(\partial)A(\partial)W
} \\
\displaystyle{
\vphantom{\Big(}
-\int \!\!\!\!\sum_{\substack{i',j',k'\in I \\ j,k,\in I, n\in\mb Z_+}}\!\!\!\!
Y^\prime_{k'}C^*_{k'k}(\lambda\!+\!\mu\!+\!\partial)
\frac{\partial D_{ki'}(\lambda)}{\partial u_j^{(n)}} (\mu+\partial)^n A_{jj'}(\mu)
\big(\big|_{\lambda=\partial}X^\prime_{i'}\big)\big(\big|_{\mu=\partial}W_{j'}\big)
\,,
}\end{array}
\end{equation}
\begin{equation}\label{20120405:X10}
\begin{array}{l}
\displaystyle{
\vphantom{\Big(}
\tint (C(\partial)Y^\prime)\cdot D^*_{B(\partial)X}(\partial) A(\partial)W
=\tint (C(\partial)Y^\prime)\cdot D^*_{X}(\partial)B^*(\partial)A(\partial)W
} \\
\displaystyle{
\vphantom{\Big(}
+\int \!\!\!\!\!\!\sum_{\substack{i',j',k'\in I \\ j,k,\in I, n\in\mb Z_+}}\!\!\!\!\!\!
Y^\prime_{k'}C^*_{k'k}\!(\lambda\!+\!\mu\!+\!\partial)
(\!-\!\lambda\!-\!\mu\!-\!\partial)^n
\frac{\partial B_{ji'}(\lambda)}{\partial u_k^{(n)}} A_{jj'}(\mu)
\big(\big|_{\lambda=\partial}\!X_{i'}\big)\!\big(\big|_{\mu=\partial}\!W_{j'}\big)
,
}\end{array}
\end{equation}
\begin{equation}\label{20120405:X11}
\begin{array}{l}
\displaystyle{
\vphantom{\Big(}
\tint (B(\partial)Y)\cdot D_{A(\partial)W}(\partial)C(\partial)X^\prime
=\tint (B(\partial)Y)\cdot A(\partial)D_W(\partial)C(\partial)X^\prime
} \\
\displaystyle{
\vphantom{\Big(}
+\int \!\!\!\!\sum_{\substack{i',j',k'\in I \\ i,k,\in I, n\in\mb Z_+}}\!\!\!\!
Y_{k'}B^*_{k'k}(\lambda\!+\!\mu\!+\!\partial)
\frac{\partial A_{kj'}(\mu)}{\partial u_i^{(n)}} (\lambda+\partial)^n C_{ii'}(\lambda)
\big(\big|_{\lambda=\partial}X^\prime_{i'}\big)\big(\big|_{\mu=\partial}W_{j'}\big)
\,,
}\end{array}
\end{equation}
\begin{equation}\label{20120405:X12}
\begin{array}{l}
\displaystyle{
\vphantom{\Big(}
-\tint (D(\partial)Y^\prime)\cdot D_{C(\partial)X^\prime}(\partial) A(\partial)W
=-\tint (D(\partial)Y^\prime)\cdot C(\partial)D_{X^\prime}(\partial)A(\partial)W
} \\
\displaystyle{
\vphantom{\Big(}
-\int \!\!\!\!\sum_{\substack{i',j',k'\in I \\ j,k,\in I, n\in\mb Z_+}}\!\!\!\!
Y^\prime_{k'}D^*_{k'k}(\lambda\!+\!\mu\!+\!\partial)
\frac{\partial C_{ki'}(\lambda)}{\partial u_j^{(n)}} (\mu+\partial)^n A_{jj'}(\mu)
\big(\big|_{\lambda=\partial}X^\prime_{i'}\big)\big(\big|_{\mu=\partial}W_{j'}\big)
\,.
}\end{array}
\end{equation}
It follows from the skewadjointness conditions \eqref{20120405:skew}
that the first term in the RHS of \eqref{20120405:X1} 
cancels with the first term in the RHS of \eqref{20120405:X4},
the first term in the RHS of \eqref{20120405:X2} 
cancels with the first term in the RHS of \eqref{20120405:X8},
the first term in the RHS of \eqref{20120405:X3} 
cancels with the first term in the RHS of \eqref{20120405:X5},
the first term in the RHS of \eqref{20120405:X6} 
cancels with the first term in the RHS of \eqref{20120405:X10},
the first term in the RHS of \eqref{20120405:X7} 
cancels with the first term in the RHS of \eqref{20120405:X11},
and the first term in the RHS of \eqref{20120405:X9} 
cancels with the first term in the RHS of \eqref{20120405:X12}.
Furthermore, 
combining the second terms of the RHS's of \eqref{20120405:X7} and \eqref{20120405:X11},
we get, thanks to \eqref{20120405:A1},
$$
\begin{array}{l}
\displaystyle{
\vphantom{\Big(}
\int 
\sum_{\substack{i',j',k'\in I \\ i,j,k\in I}}
\big(B_{kk'}(\partial)Y_{k'}\big)
\{{u_i}_\lambda\{{u_j}_\mu {u_k}\}_H\}_K
\big(\big|_{\lambda=\partial}D_{ii'}(\partial)X^\prime_{i'}\big)
\big(\big|_{\mu=\partial}B_{jj'}(\partial)W_{j'}\big)
\,.
}\end{array}
$$
Combining the second terms of the RHS's of \eqref{20120405:X3} and \eqref{20120405:X5},
we get, thanks to \eqref{20120405:A2},
$$
\begin{array}{l}
\displaystyle{
\vphantom{\Big(}
\int \sum_{\substack{i',j',k'\in I \\ i,j,k,\in I}}
\big(D_{kk'}(\partial)Y^\prime_{k'}\big)
\{{u_i}_\lambda\{{u_j}_\mu {u_k}\}_K\}_H
\big(\big|_{\lambda=\partial}B_{ii'}(\partial)X_{i'}\big)
\big(\big|_{\mu=\partial}D_{jj'}(\partial)Z^\prime_{j'}\big)
\,.
}\end{array}
$$
Combining the second terms of the RHS's of \eqref{20120405:X1} and \eqref{20120405:X4},
we get, thanks to \eqref{20120405:B1},
$$
\begin{array}{l}
\displaystyle{
\vphantom{\Big(}
-\int \!\!\sum_{\substack{i',j',k'\in I \\ i,j,k,\in I}}\!\!
\big(B_{kk'}(\partial)Y_{k'}\big)
\{{u_j}_\mu\{{u_i}_\lambda {u_k}\}_H\}_K
\big(\big|_{\lambda=\partial}B_{ii'}(\partial)X_{i'}\big)
\big(\big|_{\mu=\partial}D_{jj'}(\partial)Z^\prime_{j'}\big)
\,.
}\end{array}
$$
Combining the second terms of the RHS's of \eqref{20120405:X9} and \eqref{20120405:X12},
we get, thanks to \eqref{20120405:B2},
$$
\begin{array}{l}
\displaystyle{
\vphantom{\Big(}
-\int \!\!\sum_{\substack{i',j',k'\in I \\ i,j,k,\in I}}\!\!
\big(D_{kk'}(\partial)Y^\prime_{k'}\big)
\{{u_j}_\mu\{{u_i}_\lambda {u_k}\}_K\}_H
\big(\big|_{\lambda=\partial}D_{ii'}(\partial)X^\prime_{i'}\big)
\big(\big|_{\mu=\partial}B_{jj'}(\partial)W_{j'}\big)
\,.
}\end{array}
$$
Combining the second terms of the RHS's of \eqref{20120405:X6} and \eqref{20120405:X10},
we get, thanks to \eqref{20120405:C1},
$$
\begin{array}{l}
\displaystyle{
\vphantom{\Big(}
-\int \!\!\!\sum_{\substack{i',j',k'\in I \\ i,j,k,\in I}}\!\!\!
\big(D_{kk'}(\partial)Y^\prime_{k'}\big)
\{{\{{u_i}_\lambda{u_j}\}_H}\!_{\lambda+\mu}\!{u_k}\}_K
\big(\big|_{\lambda=\partial}B_{ii'}(\partial)X_{i'}\big)
\!\big(\big|_{\mu=\partial}B_{jj'}(\partial)W_{j'}\big)
.
}\end{array}
$$
Finally, 
combining the second terms of the RHS's of \eqref{20120405:X2} and \eqref{20120405:X8},
we get, thanks to \eqref{20120405:C2},
$$
\begin{array}{l}
\displaystyle{
\vphantom{\Big(}
-\int \!\!\!\sum_{\substack{i',j',k'\in I \\ i,j,k,\in I}}\!\!\!
\big(B_{kk'}(\partial)Y_{k'}\big)
\{{\{{u_i}_\lambda{u_j}\}_K}_{\lambda+\mu}\!{u_k}\}_H
\big(\big|_{\lambda=\partial}D_{ii'}(\partial)X^\prime_{i'}\big)
\!\big(\big|_{\mu=\partial}D_{jj'}(\partial)Z^\prime_{j'}\big)
.
}\end{array}
$$
Putting together all the above results, we conclude that the LHS of \eqref{20120405:eq16}
is equal to
$$
\begin{array}{l}
\displaystyle{
\vphantom{\Big(}
\int  \sum_{i,j,k\in I}
\big(B(\partial)Y\big)_k
\Big(
\{{u_i}_\lambda\{{u_j}_\mu {u_k}\}_H\}_K
+\{{u_i}_\lambda\{{u_j}_\mu {u_k}\}_K\}_H
} \\
\displaystyle{
\vphantom{\Big(}
-\{{u_j}_\mu\{{u_i}_\lambda {u_k}\}_H\}_K
-\{{u_j}_\mu\{{u_i}_\lambda {u_k}\}_K\}_H
-\{{\{{u_i}_\lambda{u_j}\}_H}_{\lambda+\mu}{u_k}\}_K
} \\
\displaystyle{
\vphantom{\Big(}
-\{{\{{u_i}_\lambda{u_j}\}_K}_{\lambda+\mu}\!{u_k}\}_H
\Big)
\big(\big|_{\lambda=\partial}B(\partial)X\big)_i
\big(\big|_{\mu=\partial}B(\partial)W\big)_j
\,,
}\end{array}
$$
which is zero by \eqref{20120405:eq3}.
\end{proof}

In view of Theorem \ref{20120126:prop2},
we can translate Theorem \ref{mtst} in terms of compatible non-local Hamiltonian structures.
\begin{theorem}\label{mtst-nonloc}
Let $\mc V$ be an algebra of differential functions in $u_1,\dots,u_\ell$, which is a domain.
Let $H,K\in\Mat_{\ell\times\ell}\mc V(\partial)$
be compatible non-local Hamiltonian structures on $\mc V$.
Let $H=AB^{-1},$ and $K=CD^{-1}$ be their minimal fractional decompositions
(cf. Definition \ref{def:minimal-fraction}),
with $A,B,C,D\in\Mat_{\ell\times\ell}\mc V[\partial]$, $\det B\neq0$, $\det D\neq0$.
Let $F_0=B(\partial)Z,\,F_1=D(\partial)Z^\prime=B(\partial)W,\,F_2=D(\partial)W^\prime$,
with $Z,Z^\prime,W,W^\prime\in\mc V^{\oplus\ell}$, be such that
$$
D(\partial)Z^\prime=B(\partial)W
\,\,,\,\,\,\,
A(\partial)Z=C(\partial)Z^\prime
\,\,,\,\,\,\,
A(\partial)W=C(\partial)W^\prime
\,,
$$
and
$$
D_{F_0}^*(\partial)=D_{F_0}(\partial)
\,\,,\,\,\,\,
D_{F_1}^*(\partial)=D_{F_1}(\partial)
\,.
$$
Then:
\begin{enumerate}[(a)]
\item
For all $X,\,Y\in\mc V^{\oplus\ell}$, such that
$D(\partial)X,D(\partial)Y\in\im(B)$, we have
\begin{equation}\label{20120405:eq2b}
\tint Y\cdot C^*(\partial)\big(D_{F_2}(\partial)-D_{F_2}^*(\partial)\big)C(\partial)X
=0\,.
\end{equation}
\item
If we also assume that $\det K\neq0$, then $D_{F_2}^*(\partial)=D_{F_2}(\partial)$.
\item
If, moreover, we assume that $\mc V$ is a normal algebra of differential functions,
then $F_2$ is exact: $F_2=\frac{\delta f_2}{\delta u}$ for some $\tint f_2\in\mc V/\partial\mc V$.
\end{enumerate}
\end{theorem}
\begin{proof}
By Theorem \ref{20120126:prop2}, $\mc L_{A,B}(\mc K)$ and $\mc L_{C,D}(\mc K)$
are compatible Dirac structures over $\mc K$, the field of fractions of $\mc V$.
Recalling the expressions \eqref{20120405:eq5} of $\mc N_{\mc L,\mc L^\prime}\wcheck$ 
and $\mc N_{\mc L,\mc L^\prime}$ for these Dirac structures,
we get by Theorem \ref{mtst} that equation \eqref{20120405:eq2b} holds over $\mc K$,
hence over $\mc V$, proving (a).
Let us prove part (b).
It is proved in \cite{CDSK12b} that any two matrices $B(\partial)$ and $D(\partial)$
with non zero determinant have a right common multiple 
$B(\partial)D_1(\partial)=D(\partial)B_1(\partial)$,
where $B_1(\partial),D_1(\partial)\in\Mat_{\ell\times\ell}\mc K[\partial]$
have non-zero determinant.
By clearing denominators, we can assume that $B_1(\partial)$ and $D_1(\partial)$
have coefficients in $\mc V$.
Hence, if $X,Y\in\im(B_1)$, we have $D(\partial)X,D(\partial)Y\in\im(B)$.
Therefore, by part (a) we have
$$
\int G\cdot B_1^*(\partial)C^*(\partial)\big(D_{F_2}(\partial)-D_{F_2}^*(\partial)\big)C(\partial)B_1(\partial)F
=0\,,
$$
for all $F,G\in\mc V^{\oplus\ell}$.
Since, by assumption, $C$ and $B_1$ have non-zero determinants,
it follows that $D_{F_2}^*(\partial)=D_{F_2}(\partial)$, as we wanted.
Finally, part (c) follows by the fact that, under the assumption that $\mc V$ is normal,
the variational complex is exact (see \cite[Thm.3.2]{BDSK09}).
\end{proof}


\section{Hamiltonian equations associated to a non-local Hamiltonian structure}
\label{sec:7}

\subsection{Local functionals, Hamiltonian equations and integrability}
\label{sec:7.1}

Let $\mc V$ be an algebra of differential functions in the variables $u_i,\,i\in I$,
assume that it is a domain, and let $\mc K$ be it field of fractions.
Let $H\in\Mat_{\ell\times\ell}\mc V(\partial)$ be a non-local Hamiltonian structure on $\mc V$.
Recall that, since $H$ has rational entries, it admits a minimal fractional decomposition
$H=AB^{-1}$, where $A,B\in\Mat_{\ell\times\ell}\mc V[\partial]$ and $\det B\neq0$
(cf. Definition \ref{def:minimal-fraction}).
Throughout this section we fix such a minimal fractional decomposition for $H$.
Recall that $\mc L_{A,B}$ is a Dirac structure on $\mc V$ by Theorem \ref{20111020:thm}.
\begin{definition}\label{20120124:def}
A \emph{Hamiltonian functional} (for $H=AB^{-1}$)
is an element $\tint h\in\mc V/\partial\mc V$
such that $\frac{\delta h}{\delta u}=B(\partial)F$
for some $F\in\mc V^{\oplus\ell}$.
Then, $P=A(\partial)F\in\mc V^\ell$ is called a \emph{Hamiltonian vector field}
associated to $\tint h$.
We denote by $\mc F(H)\subset\mc V/\partial\mc V$ the subspace of all Hamiltonian functionals,
and by $\mc H(H)\subset\mc V^\ell$ the subspace of all Hamiltonian vector fields:
$$
\mc F(H)=\Big(\frac{\delta}{\delta u}\Big)^{-1}(\im B)\,\subset\mc V/\partial\mc V
\quad,\qquad
\mc H(H)=A\Big(B^{-1}\Big(\im\frac{\delta}{\delta u}\Big)\Big)\,\subset\mc V^\ell\,.
$$
We say that $\tint f\in\mc F(H)$ and $P\in\mc H(H)$ are \emph{associated} 
(with respect to the Hamiltonian structure $H=AB^{-1}$)
if there exists $F\in\mc V^\ell$ such that 
$\frac{\delta f}{\delta u}=B(\partial)F,\,P=A(\partial)F$.
In this case we write $\tint f\ass{H} P$.
\begin{remark}\label{20120201:rem3}
Note that, the spaces $\mc F(H)\subset\mc V/\partial\mc V$ 
and $\mc H(H)\subset\mc V^\ell$,
as well as the relation $\tint h\ass{H}P$, for $\tint h\in\mc F(H)$ and $P\in\mc H(H)$,
may depend not only on the non-local Hamiltonian structure $H$,
but also on its fractional decomposition $H=AB^{-1}$.
However, if $A$ and $B$ are multiplied on the right by a matrix differential operator $D$
which is invertible in the algebra $\Mat_{\ell\times\ell}\mc V[\partial]$,
the spaces $\mc F(H)$ and $\mc H(H)$,
as well as the relation $\tint h\ass{H}P$, are unchanged.
In particular, in the special case when $\mc V=\mc K$ is a field,
by Proposition \ref{prop:minimal-fraction},
the spaces $\mc F(H), \mc H(H)$, and the relation $\tint h\ass{H}P$,
are independent of the minimal fractional decomposition of $H$.
\end{remark}

In the local case, when $H\in\Mat_{\ell\times\ell}\mc V[\partial]$ 
is a (local) Hamiltonian structure on $\mc V$,
then
$\tint f\in\mc F(H)=\mc V/\partial\mc V$ 
and $P\in\mc H(H)=H(\partial)\Big(\im\frac{\delta}{\delta u}\Big)\,\subset\mc V^\ell$ 
are associated if and only if
$P=H(\partial)\frac{\delta\tint f}{\delta u}$.
\end{definition}
\begin{lemma}\label{20120124:lem}
\begin{enumerate}[(a)]
\item
The space $\mc V^\ell$ is a Lie algebra with bracket 
\eqref{20120126:eq1},
and $\mc H(H)\subset\mc V^\ell$ is its subalgebra.
\item
We have a representation $\phi$ of the Lie algebra $\mc V^\ell$ 
on the space $\mc V/\partial\mc V$ given by
$$
\phi(P)\big(\tint h\big)=\int P\cdot\frac{\delta h}{\delta u}\,,
$$
and the subspace $\mc F(H)\subset\mc V/\partial\mc V$
is preserved by the action of the Lie subalgebra $\mc H(H)\subset\mc V^\ell$.
\item
If $\tint h\ass{H}P$ and $\tint h\ass{H}Q$ for some $\tint h\in\mc F(H)$,
then the action of $P,Q\in\mc H(H)$on $\mc F(H)$ is the same:
$$
\int P\cdot\frac{\delta g}{\delta u}
=\int Q\cdot\frac{\delta g}{\delta u}
\,\,\text{ for all } \tint g\in\mc F(H)\,.
$$
\end{enumerate}
\end{lemma}
\begin{proof}
It follows immediately from \cite[Lem.4.7-8]{BDSK09}.
\end{proof}

Thanks to Lemma \ref{20120124:lem},
we have a well-defined map 
$\{\cdot\,,\,\cdot\}_H:\,\mc F(H)\times\mc F(H)\to\mc F(H)$
given by
\begin{equation}\label{20120124:eq4}
\{\tint f,\tint g\}_H
=
\int P\cdot\frac{\delta g}{\delta u}
\quad
\Big(
=\int \frac{\delta g}{\delta u}\cdot A(\partial) B^{-1}(\partial) \frac{\delta f}{\delta u}
\,\,\Big)\,,
\end{equation}
where $P\in\mc H(H)$ is such that $\tint f\ass{H}P$.
\begin{proposition}\label{20120124:prop}
\begin{enumerate}[(a)]
\item
The bracket \eqref{20120124:eq4} is a Lie algebra bracket on the space 
of Hamiltonian functionals $\mc F(H)$.
\item
The Lie algebra action of $\mc H(H)$ on $\mc F(H)$ is by derivations
of the Lie bracket \eqref{20120124:eq4}.
\item
The subspace 
$$
\mc A(H)=\Big\{(\tint f,P)\in\mc F(H)\times\mc H(H)\,\Big|\,\tint f\ass{H}P\Big\}\,
$$
is a subalgebra of the direct product of Lie algebras $\mc F(H)\times\mc H(H)$.
\end{enumerate}
\end{proposition}
\begin{proof}
It follows immediately from \cite[Prop.4.9, Rem.4.6]{BDSK09}.
\end{proof}

A \emph{Hamiltonian equation} associated to the Hamiltonian structure $H$
and to the Hamiltonian functional $\tint h\in\mc F(H)$,
with an associated Hamiltonian vector field $P\in\mc H(H)$,
is, by definition, the following evolution equation on the variables $u=\big(u_i\big)_{i\in I}$:
\begin{equation}\label{20120124:eq5}
\frac{du}{dt}
=P\,.
\end{equation}
By the chain rule, any element $f\in\mc V$ evolves according to the equation
$$
\frac{df}{dt}=\sum_{i\in I}\sum_{n\in\mb Z_+}(\partial^nP_i)\frac{\partial f}{\partial u_i^{(n)}}\,,
$$
and, integrating by parts,
a local functional $\tint f\in\mc V/\partial\mc V$
evolves according to
$$
\frac{d\tint f}{dt}=\int P\cdot\frac{\delta f}{\delta u}\,.
$$
An \emph{integral of motion} for the Hamiltonian equation \eqref{20120124:eq5}
is a Hamiltonian functional $\tint f\in\mc F(H)$ such that
$$
\frac{d\tint f}{dt}=\{\tint h,\tint f\}_H=0\,,
$$
i.e. $\tint f$ lies in the centralizer of $\tint h$ in the Lie algebra $\mc F(H)$.
The Hamiltonian equation \eqref{20120124:eq5} is said to be \emph{integrable}
if there is an infinite sequence of pairs $(\tint h_n,P_n)\in\mc F(H)\times\mc H(H),\,n\geq0$,
such that $\tint h_0=\tint h$, $P_0=P$,
we have $\tint h_n\ass{H}P_n$ for every $n\in\mb Z_+$,
the sequence $\{\tint h_n\}_{n\in\mb Z_+}$ spans 
an infinite-dimensional abelian subalgebra of the Lie algebra $\mc F(H)$,
and the sequence $\{P_n\}_{n\in\mb Z_+}$ spans an infinite-dimensional 
abelian subalgebra of the Lie algebra $\mc H(H)$.
(Equivalently, if there exists an abelian subalgebra $\mf h$
of the Lie algebra $\mc A(H)$ defined in Proposition \ref{20120124:prop}(c),
such that both its canonical projections $\pi_1(\mf h)$ and $\pi_2(\mf h)$ 
in $\mc F(H)$ and $\mc H(H)$ are infinite dimensional, 
and $\tint h\in\pi_1(\mf h)$.)
In this case, we have an \emph{integrable hierarchy} of Hamiltonian equations
$$
\frac{du}{dt_n} = P_n\,,\,\,n\in\mb Z_+\,.
$$

\subsection{The Lenard-Magri scheme of integrability}
\label{sec:7.2}

Throughout the section, we
let $H,K\in\Mat_{\ell\times\ell}\mc V(\partial)$ be a pair of compatible Hamiltonian structures
on the algebra of differential functions $\mc V$.
We also let $H=AB^{-1}$ and $K=CD^{-1}$,
with $A,B,C,D\in\Mat_{\ell\times\ell}\mc V[\partial]$
and $\det B\neq0,\det D\neq0$,
be their minimal fractional decomposition,
and let $\mc L_{A,B},\mc L_{C,D}\subset\mc V^{\oplus\ell}\oplus\mc V^\ell$
be the corresponding compatible Dirac structures
(see Theorems \ref{20111020:thm} and \ref{20120126:prop2}).

The \emph{Lenard-Magri scheme of integrability} 
consists in finding  sequences 
$$
\begin{array}{l}
\displaystyle{
\vphantom{\Big(}
0=\tint h_{-1},\,\tint h_0,\,\tint h_1,\,\tint h_2,\dots\in\mc F(H)\cap\mc F(K)\,,
} \\
\displaystyle{
\vphantom{\Big(}
P_0,\,P_1,\,P_2,\dots\in\mc H(H)\cap\mc H(K)\,,
}
\end{array}
$$
such that 
\begin{equation}\label{maxi}
\UseTips
\xymatrix{
& P_{n-1} \ar@{<->}[dl]_{H} \ar@{<->}[dr]^{K} & & 
P_n \ar@{<->}[dl]_{H} \ar@{<->}[dr]^{K} & & 
P_{n+1} \ar@{<->}[dl]_{H} \ar@{<->}[dr]^{K} & \\
\dots & & \tint h_{n-1} & & \tint h_n & & \dots 
}
\end{equation}
%
Explicitly, diagram \eqref{maxi} holds if and only if there exists a sequence 
$\{F_n\}_{n\geq-1}$ in $\mc V^{\oplus\ell}$ such that the following equations hold:
\begin{equation}\label{20120315:eq1}
B(\partial)F_{-1}=0
\,\,,\,\,\,
C(\partial)F_{2n}=A(\partial)F_{2n-1}=P_n
\,\,,\,\,\,
B(\partial)F_{2n+1}=D(\partial)F_{2n}=\frac{\delta h_n}{\delta u}
\end{equation}
for all $n\geq 0$.
\begin{proposition}\label{20120126:prop3}
\begin{enumerate}[(a)]
\item
If the sequence of pairs $\big(\tint h_n,P_n\big)$, $0\leq n\leq N+1$,
where $\tint h_n\in\mc F(H)\cap\mc F(K)$,
$P_n\in\mc H(H)\cap\mc H(K)$,
satisfies 
$\tint h_{n-1}\ass{H}P_{n}\ass{K}\tint h_n$
for $0\leq n\leq N$,
then
\begin{equation}\label{20120126:eq3}
\{\tint h_m,\tint h_n\}_H=\{\tint h_m,\tint h_n\}_K=0
\,\,,\,\,\,\,
0\leq m,n\leq N\,.
\end{equation}
\item
We have 
\begin{equation}\label{20120126:eq4}
[P_m,P_n]
\,\in\,\ker B^*\cap\ker D^*
\,\,\,\,
\text{ for all } m,n\leq N\,.
\end{equation}
\end{enumerate}
\end{proposition}
\begin{proof}
For $m=n$ equation \eqref{20120126:eq3} holds trivially,
by the skew symmetry of the Lie brackets $\{\cdot\,,\,\cdot\}_{H}$ and $\{\cdot\,,\,\cdot\}_{K}$.
Assuming $m<n$, 
we prove \eqref{20120126:eq3} by induction on $n-m$.
We have, by the assumption on the $\tint h_n$'s and $P_n$'s
and the definition \eqref{20120124:eq4} of the Lie bracket on $\mc F(H)$,
$$
\big\{\tint h_m,\tint h_n\big\}_H=\int P_{m+1}\cdot\frac{\delta h_n}{\delta u}
=\big\{\tint h_{m+1},\tint h_n\big\}_K=0\,,
$$
by the inductive assumption.
Similarly,
$$
\begin{array}{l}
\vphantom{\Big(}
\big\{\tint h_m,\tint h_n\big\}_K
=-\big\{\tint h_n,\tint h_m\big\}_K
=-\int P_n\cdot\frac{\delta h_m}{\delta u}
\\
\vphantom{\Big(}
=-\big\{\tint h_{n-1},\tint h_m\big\}_H
=\big\{\tint h_{m},\tint h_{n-1}\big\}_H
=0\,.
\end{array}
$$
By Theorem \ref{20111020:thm}, $\mc L_{A,B}$ and $\mc L_{C,D}$
are Dirac structures, hence they are closed under 
the Courant-Dorfman product \eqref{20111020:eq5}.
By formula \eqref{eq:dirac} and 
the assumption on the $\tint h_n$'s and $P_n$'s,
we have that
$\frac{\delta\tint h_{n-1}}{\delta u}\oplus P_n\in\mc L_{A,B}$,
and $\frac{\delta\tint h_n}{\delta u}\oplus P_n\in\mc L_{C,D}$, for every $0\leq n\leq N$.
It then follows, by equation \eqref{20120127:eq1} and formulas 
\eqref{20120126:eq1} and \eqref{20120124:eq4}, that
$$
\begin{array}{l}
\displaystyle{
\Big(\frac{\delta\tint h_{m-1}}{\delta u}\oplus P_m\Big)
\circ
\Big(\frac{\delta\tint h_{n-1}}{\delta u}\oplus P_n\Big)
=
\frac{\delta}{\delta u}\{\tint h_{m-1},\tint h_{n-1}\}_H\oplus[P_m,P_n]
} \\
\displaystyle{
\Big(\frac{\delta\tint h_{m}}{\delta u}\oplus P_m\Big)
\circ
\Big(\frac{\delta\tint h_{n}}{\delta u}\oplus P_n\Big)
=
\frac{\delta}{\delta u}\{\tint h_{m},\tint h_{n}\}_K\oplus[P_m,P_n]
}
\end{array}
$$
Hence, by \eqref{20120126:eq3}, we get that
$0\oplus[P_m,P_n]\in\mc L_{A,B}\cap\mc L_{C,D}$.
Namely, there exist $F\in\ker B\subset\mc V^\ell$ and $G\in\ker D\subset\mc V^\ell$
such that $[P_m,P_n]=AF=CG\in A(\ker B)\cap C(\ker D)$.
To conclude, we finally observe that, by skewadjointness of $H$ and $K$,
we have $B^*A=-A^*B$ and $D^*C=-C^*D$, which immediately implies
$A(\ker B)\subset\ker B^*$ and $C(\ker D)\subset\ker D^*$.
\end{proof}
\begin{corollary}\label{20120127:cor}
Suppose that the sequences $\{\tint h_n\}_{n\in\mb Z_+},\,\{P_n\}_{n\in\mb Z_+}$,
span infinite dimensional subspaces of $\mc V/\partial\mc V$ and $\mc V^\ell$ respectively,
and satisfy conditions \eqref{maxi} (where $\tint h_{-1}=0$).
Suppose, moreover, that $\ker B^*\cap\ker D^*=0$.
Then,
the hierarchy of bi-Hamiltonian equations
$$
\frac{du}{dt_n}=P_n\,\,,\,\,\,\,n\in\mb Z_+\,.
$$
is integrable.
\end{corollary}
\begin{proof}
It is immediate from Proposition \ref{20120126:prop3}
and the definition of integrability.
\end{proof}

Next, we want to discuss in which situations one can apply successfully the Lenard-Magri scheme.
In order to do so, recall that an algebra of differential functions $\mc V$ 
is called \emph{normal} \cite{BDSK09}
if $\frac{\partial}{\partial u_i^{(m)}}(\mc V_{m,i})=\mc V_{m,i}$ for every $m\in\mb Z_+,i\in I$, 
where
$$
\mc V_{m,i}
=\Big\{f\in\mc V\,\Big|\,\frac{\partial f}{\partial u_j^{(n)}}=0
\text{ for } (n,j)>(m,i) \quad (\text{in lexicographic order})
\Big\}\,.
$$
For example, the algebra of differential polynomials in $\ell$ variables is normal,
and any algebra of differential functions $\mc V$ can be included in a normal one.
\begin{theorem}\label{20120127:thm}
Let $\mc V$ be a normal algebra of differential functions.
Let $H,K\!\in\!\Mat_{\ell\times\ell}\!\mc V(\partial)$ be a pair of compatible
non-local Hamiltonian structures on $\mc V$
and assume that $\det K\neq0$.
Let $H=AB^{-1}$, $K=CD^{-1}$, be their minimal fractional decompositions,
with $A,B,C,D\in\Mat_{\ell\times\ell}\mc V[\partial]$.
Let  $N\geq0$,
and let $0=\tint h_{-1},\tint h_0,\tint h_1,\dots,\tint h_N\in\mc F(H)\cap\mc F(K)$
and $P_0,P_1,\dots,P_N\in\mc H(H)\cap\mc H(K)$
be such that
$\tint h_{n-1}\ass{H}P_{n}\ass{K}\tint h_n$ for $0\leq n\leq N$.
Suppose, moreover, that the following orthogonallity conditions hold:
$$
\big(\Span{}_{\mc C}\big\{\frac{\delta h_n}{\delta u}\big\}_{n=0}^N\big)^\perp\subset\im C
\,\,,\,\,\,\,
\big(\Span{}_{\mc C}\big\{P_n\big\}_{n=0}^N\big)^\perp \subset\im B\,,
$$
where the orthogonal complement is with respect to the pairing
$\mc V^{\oplus\ell}\times\mc V^\ell\to\mc V/\partial\mc V$,
defined by $(F,P)\mapsto\tint F\cdot P$.
Then, there exist infinite sequences
$\{\tint h_n\}_{n\in\mb Z_+}\subset\mc F(H)\cap\mc F(K)$
and $\{P_n\}_{n\in\mb Z_+}\subset\mc H(H)\cap\mc H(K)$,
extending the given finite sequences,
satisfying conditions \eqref{maxi}.
In other words, the Lenard-Magri scheme of integrability can be applied.
\end{theorem}
\begin{proof}
By Proposition \ref{20120126:prop3}, we have
$$
\tint\frac{\delta h_N}{\delta u}\cdot P_n=\{\tint h_N,\tint h_n\}_K=0\,,
$$
for all $n=0,\dots,N$. 
Hence, by the second orthogonality condition we get that
$\frac{\delta h_N}{\delta u}=B(\partial)X$ for some $X\in\mc V^{\oplus\ell}$.
Let $P_{N+1}=A(\partial)X\in\mc V^\ell$.
Again by Proposition \ref{20120126:prop3}, we have
$$
\tint\frac{\delta h_n}{\delta u}\cdot P_{N+1}=\{\tint h_n,\tint h_N\}_H=0\,,
$$
for all $n=0,\dots,N$. 
Hence, by the first orthogonality condition we get that
$P_{N+1}=C(\partial)Y$ for some $Y\in\mc V^{\oplus\ell}$.
Let $F_{N+1}=D(\partial)Y\in\mc V^{\oplus\ell}$.
It follows by Theorem \ref{mtst-nonloc}(c) 
(with $\frac{\delta h_{N-1}}{\delta u}$ in place of $F_0$,
$\frac{\delta h_{N}}{\delta u}$ in place of $F_1$,
and $F_{N+1}$ in place of $F_2$)
that $F_{N+1}$ is exact,
i.e. $F_{N+1}=\frac{\delta h_{N+1}}{\delta u}$
for some $\tint h_{N+1}\in\mc V/\partial\mc V$.
By construction, 
$\tint h_{N}\ass{H}P_{N+1}\ass{K}\tint h_{N+1}$.
Hence, we extended the sequences  $\{\tint h_n\}_{n=-1}^N$ and $\{P_n\}_{n=0}^N$
by one step.
The claim follows by induction.
\end{proof}
%
%
\begin{remark}\label{20120201:rem2}
Let $\mc V$ be an algebra of differential function in $u_i,\,i\in I$.
Assume that $\mc V$ is a domain, and let $\mc K$ be its field of fractions.
Let $H=AB^{-1},\,K=CD^{-1}\in\Mat_{\ell\times\ell}\mc V(\partial)$
be skewadjoint rational matrix differential operators with coefficients in $\mc V$,
with $A,B,C,D\in\Mat_{\ell\times\ell}\mc V[\partial],\,\det B,\det D\neq0$.
Let $\xi_{-1}=0,\xi_0,\dots \xi_N\in\mc V^{\oplus\ell}$
and $P_0,\dots,P_N\in\mc V^\ell$ be such that
\begin{equation}\label{20120201:eq1}
\xi_{n-1}=B(\partial)F_n
\,\,,\,\,\,\,
P_n=A(\partial)F_n
\,\,,\,\,\,\,
\xi_n=D(\partial)G_n
\,\,,\,\,\,\,
P_n=C(\partial)G_n
\,,
\end{equation}
for $0\leq n\leq N$ and some $F_0,\dots,F_N,G_0,\dots,G_N\in\mc V^\ell$, and
$$
\big(\Span{}_{\mc C}\{\xi_n\}_{n=0}^N\big)^\perp
\subset\im C
\,\,,\,\,\,\,
\big(\Span{}_{\mc C}\{P_n\}_{n=0}^N\big)^\perp\subset\im B\,.
$$
By the same arguments as in the proof of Theorem \ref{20120127:thm} we can extend the 
given seguences to infinite sequences 
$\{\xi_n\}_{n\in\mb Z_+}\subset\mc V^{\oplus\ell}$ and $\{P_n\}_{n\in\mb Z_+}\subset\mc V^\ell$
satisying conditions \eqref{20120201:eq1} for every $n\in\mb Z_+$.
Next, suppose that $H=AB^{-1}$ and $K=CD^{-1}$
are minimal fractional decompositions (cf. Definition \ref{def:minimal-fraction})
and that $H$ and $K$ are compatible Hamiltonian structures on $\mc V$.
Suppose, moreover, that $\det K\neq0$ and $\xi_0=\frac{\delta h_0}{\delta u}$
for some $\tint h_0\in\mc V/\partial\mc V$.
Then it follows by Theorem \ref{mtst-nonloc} that all $\xi_n$'s are closed in $\mc K^{\oplus\ell}$,
i.e. $D_{\xi_n}(\partial)$ is self-adjoint for every $n\in\mb Z_+$.
Hence, all the $\xi_n$'s are closed in $\mc V^{\oplus\ell}$.
Finally, taking, if necesary, a normal extension $\tilde{\mc V}$ of $\mc V$,
one proves that all the $\xi_n$'s are exact,
i.e. there exist elements $\tint h_n\in\tilde{\mc V}/\partial\tilde{\mc V}$
such that $\xi_n=\frac{\delta\tint h_n}{\delta u},\,n\in\mb Z_+$,
\cite[Thm.3.2]{BDSK09}.
Hence, the Lenard-Magri scheme can be applied.
\end{remark}

\begin{remark}\label{20121004:rem}
Recall from Theorem \ref{20111021:thm} that, 
given two compatible non-local Hamiltonian structures $H$ and $K$, 
with $\det(K)\neq0$ ,
we get an infinite family of compatible non-local Hamiltonian structures
given by $H^{[0]}=K$, and $H^{[s]}=(H\circ K^{-1})^{s-1}\circ H$, for $s\geq1$.
Suppose that the infinite sequences
$\{\tint h_n\}_{n\in\mb Z_+}\subset\mc F(H)\cap\mc F(K)$,
$\{P_n\}_{n\in\mb Z_+}\subset\mc H(H)\cap\mc H(K)$
satisfy the Lenard-Magri recursive relations
$\tint h_{n-1}\ass{H} P_n\ass{K}\tint h_n$ for every $n\geq0$ 
(we let $\tint h_{-1}=\tint 0$).
Then, for every $n,s\in\mb Z_+$ we have the relations
$\tint h_n\ass{H^{[s]}} P_{n+s}$.
This will be proved in \cite{DSK12}.
Consequently, all the evolutionary equations $\frac{du}{dt_n}=P_n$
are Hamiltonian with respect to all the non-local Hamiltonian structures $H^{[s]},\,s\in\mb Z_+$.
\end{remark}

\subsection{Example: NLS integrable hierarchy}
\label{sec:7.3}

Consider the algebra of differential polynoamials in two variables $R_2=\mb F[u,v,u',v',\dots]$,
and its extension $\mc V_u=\mb F[u^{\pm1},v,u',v',\dots]$.
Let also $\mc K$ be the field of fractions of $R_2$ (which is the same as that of $\mc V_u$).
Consider the following pair of compatible non-local Hamiltonian structures
with coefficients in $R_2$ from Example \ref{20110922:ex5}:
$$
H=\left(\begin{array}{cc} 
v\partial^{-1}\circ v & -v\partial^{-1}\circ u \\
-u\partial^{-1}\circ v & u\partial^{-1}\circ u
\end{array}\right)+c\partial\id
\,\,,\,\,\,\,
K=\left(\begin{array}{cc} 0 & -1 \\ 1 & 0 \end{array}\right)\,,
$$
where $c\in\mb F$ is a fixed constant.
The operators $H$ and $K$ admit the following fractional decompositions:
$H=AB^{-1},\,K=CD^{-1}$, where $C=K$, $D=\id$, and
$$
A=\left(\begin{array}{cc} 
c\partial\circ u & -u^2v \\
c\partial\circ v & u^3+c\partial\circ(u\partial+2u') 
\end{array}\right)
\,\,,\,\,\,\,
B=\left(\begin{array}{cc} 
u & 0 \\
v & u\partial+2u'
\end{array}\right)
\,.
$$
This fractional decomposition of $H$ is minimal over $\mc K$
by Proposition \ref{prop:minimal-fraction}
since  $\ker\bar B\subset\bar{\mc K}^2$ is spanned by 
$\left(\begin{array}{c} 0 \\ u^{-2} \end{array}\right)$,
which does not lie in $\ker\bar A$.

Let $\tint h_{-1}=0,\,\tint h_0=\frac12(u^2+v^2)\in\mc V_u/\partial\mc V_u$
and $P_0=\left(\begin{array}{c} -v \\ u \end{array}\right)\in\mc V_u^2$.
We have $0=\tint h_{-1}\ass{H}P_0\ass{K}\tint h_0$. 
Indeed,
letting $F_{-1}=\left(\begin{array}{c} 0 \\ u^{-2} \end{array}\right)$
and $F_0=\left(\begin{array}{c} u \\ v \end{array}\right)$,
we have
$$
\begin{array}{l}
\displaystyle{
B(\partial)F_{-1}
=0=\frac{\delta\tint h_{-1}}{\delta u}
\,\,,\,\,\,\,
A(\partial)F_{-1}
=\left(\begin{array}{c} -v \\ u \end{array}\right)=P_0
} \\
\displaystyle{
\id F_0
=\frac{\delta\tint h_0}{\delta u}
\,\,,\,\,\,\,
KF_0
=\left(\begin{array}{c} -v \\ u \end{array}\right)=P_0
\,.
}
\end{array}
$$
Note that, even though $0=\tint h_1,\tint h_0\in R_2/\partial R_2$ and $P_0\in R_2^2$,
they are NOT associated within the algebra $R_2$,
since $F_{-1}$ does not lie in $R_2^2$.

Next, we prove the orthogonality conditions:
$\big(\frac{\delta\tint h_0}{\delta u}\big)^\perp\subset\im\id$,
$\big(P_0\big)^\perp\subset\im(B:\,\mc V_u^2\to\mc V_u^2)$.
The first condition is obvious.
For the second, $F=\left(\begin{array}{c} f \\ g \end{array}\right)\in\mc V_u^2$
is orthogonal to $P_0$ if and only if
$$
-vf+ug=\partial h\in\partial\mc V_u\,.
$$
In this case, 
$F=B(\partial)\left(\begin{array}{c} f/u \\ h/u^2 \end{array}\right)\in\im B$.

By the observations in Remark \ref{20120201:rem2},
there exist sequences $\{\xi_n\}_{n\in\mb Z_+}$ and $\{P_n\}_{n\in\mb Z_+}$ in $\mc V_u^2$
satisfying condition \eqref{20120201:eq1} for all $n\in\mb Z_+$,
and all the $\xi_n$'s are closed, in the sense that 
their Frechet derivative $D_{\xi_n}(\partial)$ is self-adjoint for every $n$.

We argue that, in fact, all these elements $\xi_n,P_n,\,n\in\mb Z_+$,
lie in $R_2^2$.
Indeed, it is not hard to show, using the explicit form of the matrices $A$ and $B$, 
that, if $\xi_{n-1}$ lies in $R_2^2$, then necessarily 
both $P_n$ and $\xi_n$ lie in $R_2^2$.
Alternatively, we can exchange the roles of $u$ and $v$
to prove that all $\xi_n$'s and $P_n$'s lie in $\mc V_v^2$, hence in $R_2^2$.
First, we observe that,
since $\ker B$ is one dimensional over $\mb F$ and $C=K$ is invertible,
the flag of subspaces
\begin{equation}\label{20120202:eq1}
\begin{array}{l}
U_0=\Span_{\mb F}\{\xi_0\}\,,
U_1=\Span_{\mb F}\{\xi_0,\xi_1\}\,,\dots
\,,\\
V_0=\Span_{\mb F}\{P_0\}\,,
V_1=\Span_{\mb F}\{P_0,P_1\}\,,\dots
\subset \mc{V}_u^2\,,
\end{array}
\end{equation}
associated to the sequences
$\{\xi_n\}_{n\in\mb Z_+}$ and $\{P_n\}_{n\in\mb Z_+}$
satisfying condition \eqref{20120201:eq1} for every $n\in\mb Z_+$,
is uniquely defined.
Consider now
the new fractional decomposition $H=A_1B_1^{-1}$,
where
$$
A_1=\left(\begin{array}{cc} 
v^3+c\partial\circ(v\partial+2v') & c\partial\circ u \\
-uv^2 & c\partial\circ v
\end{array}\right)
\,\,,\,\,\,\,
B_1=\left(\begin{array}{cc} 
v\partial+2v' & u \\
0 & v
\end{array}\right)
\,.
$$
By the same argument as above,
we get that there exist sequences
$\{\xi^1_n\}_{n\in\mb Z_+}$ and $\{P^1_n\}_{n\in\mb Z_+}$ in $\mc V_v^2$,
satysfying condition \eqref{20120201:eq1}
with respect to the new fractional decomposition, i.e.
$$
\xi^1_{n-1}=B_1(\partial)F^1_n
\,\,,\,\,\,\,
P^1_n=A_1(\partial)F^1_n
\,\,,\,\,\,\,
\xi^1_n=D(\partial)G^1_n
\,\,,\,\,\,\,
P^1_n=C(\partial)G^1_n
\,,
$$
for some elements $F^1_n,\,G^1_n\in\mc V_v^2$.
Note that the two fractional decompositions $H=AB^{-1}$ and $H=A_1B_1^{-1}$
are equivalent,
in the sense that $A_1$ and $B_1$ are obtained from $A$ and $B$
by multiplication on the right by the matrix
$$
D=B^{-1}B_1=
\left(\begin{array}{cc} 
u^{-1}v^{-1}\partial\circ v^2 & 1 \\
-u^{-2}v^2 & 0
\end{array}\right)
\,,
$$
which is invertible in $\Mat{}_{\ell\times\ell} \mc K[\partial]$.
Hence, by the arguments in Remark \ref{20120201:rem2}
we get that both the sequences
$\{\xi_n\}_{n\in\mb Z_+},\,\{P_n\}_{n\in\mb Z_+}\subset\mc V_u^2$,
and the sequences
$\{\xi^1_n\}_{n\in\mb Z_+},\,\{P^1_n\}_{n\in\mb Z_+}\subset\mc V_v^2$,
considered as sequences in $\mc K^2$,
satisfy condition \eqref{20120201:eq1} for every $n\in\mb Z_+$.
Therefore, by the uniqueness of the flag \eqref{20120202:eq1},
we get that
the subspaces
$U_n=\Span_{\mb F}\{\xi_0,\dots,\xi_n\}\subset\mc V_u^2$
and $U^1_n=\Span_{\mb F}\{\xi^1_0,\dots,\xi^1_n\}\subset\mc V_v^2$
must be the same subspace of $\mc K^2$,
and similarly the subspaces
$V_n=\Span_{\mb F}\{P_0,\dots,P_n\}\subset\mc V_u^2$
and $V^1_n=\Span_{\mb F}\{P^1_0,\dots,P^1_n\}\subset\mc V_v^2$
must be the same subspace of $\mc K^2$.
In conclusion, all $\xi_n$'s and $P_n$'s lie in $\mc V_u^2\cap\mc V_v^2=R_2^2$.

Since $R_2$ is a normal algebra of differential functions,
it follows from \cite[Thm.4.13]{BDSK09} that all $\xi_n$'s are exact,
i.e. there exists elements $\tint h_n\in\mc V/\partial\mc V,\,n\in\mb Z_+$,
such that $\xi_n=\frac{\delta\tint h_n}{\delta u}$.
Hence, the Lenard scheme can be applied.

It is not hard to show, by induction on $n$, that both entries of $\xi_n\in R_2^2$
have differential order $n$.
Therefore, the vectors $\xi_n$ are linearly independent.
It follows that the elements $\tint h_n\in\mc V/\partial\mc V,\,n\in\mb Z_+$
and the elements $P_n=K\xi_n\in R_2^2,\,n\in\mb Z_+$,
are linearly independent as well.
Since $D=\id$, we can apply Proposiion \ref{20120126:prop3} and Corollary \ref{20120127:cor}
to deduce that 
we have an infinite hierarchy of compatible integrable equations
$\frac{du}{dt_n}=P_n,\,n\in\mb Z_+$,
for which $\tint h_n\in\mc F(H)\cap\mc F(K)$ are integrals of motion.

We can compute the first few equations of this integrable hierarchy, 
called the non-linear Schroedinger (NLS) hierarchy:
$$
\begin{array}{l}
\displaystyle{
\frac{d}{dt_0}
\left(\begin{array}{c}u \\ v \end{array}\right)
=
\left(\begin{array}{c}-v \\ u \end{array}\right)\,,
} \\
\displaystyle{
\frac{d}{dt_1}
\left(\begin{array}{c}u \\ v \end{array}\right)
=
c\left(\begin{array}{c}u' \\ v' \end{array}\right)\,,
} \\
\displaystyle{
\frac{d}{dt_2}
\left(\begin{array}{c}u \\ v \end{array}\right)
=
\left(\begin{array}{c} c^2v''+\frac{c}{2} v(u^2+v^2) \\ -c^2u''-\frac{c}{2} u(u^2+v^2) \end{array}\right)\,,
} \\
\displaystyle{
\frac{d}{dt_3}
\left(\begin{array}{c}u \\ v \end{array}\right)
=
\left(\begin{array}{c}
-c^3u'''-\frac32 c^2 (u^2+v^2)u' \\
-c^3 v'''-\frac32 c^2 (u^2+v^2)v'
\end{array}\right)\,.
} \\
\dots
\end{array}
$$
The third equation is called the (coupled) NLS equation.
The first four integrals of motion are
$$
\begin{array}{l}
\displaystyle{
\tint h_0=\int \frac12 (u^2+v^2)\,,\,\,
\tint h_1=c\int uv'\,,\,\,
} \\
\displaystyle{
\tint h_2=\int\Big(\frac{c^2}{2}(u'^2+v'^2)-\frac{c}{8}(u^2+v^2)^2\Big)\,,
} \\
\displaystyle{
\tint h_3=\int\Big(-c^3uv'''-\frac{c^2}{2}(u^3v'-v^3u')\Big)
\,.
}
\end{array}
$$
Of course this is a very well known hierarchy, studied by many different methods,
see e.g. \cite{TF86,Dor93,BDSK09}.
The approach of the last two papers, based on a Dirac structure, is close to ours 
(even though in \cite{Dor93} the proposed Dirac structure is not quite a Dirac structure).



\end{document}